\documentclass[12pt, reqno]{amsart}
\usepackage{amsthm, amscd, amsfonts, amssymb, graphicx, xcolor,mathtools}
\usepackage{amsmath}
\usepackage[a4paper, left=0.85in, right=0.85in, top=1in, bottom=1in]{geometry}
\usepackage[round,authoryear]{natbib}

\usepackage{xcolor}      
\usepackage[ bookmarksnumbered, plainpages, backref = page,
            colorlinks=true,   
            citecolor=blue,    
            linkcolor=red,    
            urlcolor=blue]{hyperref} 

\usepackage{bookmark}

\usepackage{subcaption}
\usepackage{setspace}

\usepackage{adjustbox}     
\usepackage{caption}       
\usepackage{floatrow}         
\usepackage{tabularray}    
\usepackage{booktabs}  
\usepackage{ragged2e}
\usepackage{multirow}
\usepackage[shortlabels]{enumitem}
\usepackage{pdflscape}
\usepackage{rotating}
\usepackage{placeins}

\usepackage{csquotes}
\newtheorem{theorem}{Theorem}[section]
\newtheorem{lemma}[theorem]{Lemma}
\newtheorem{proposition}{Proposition}[section]
\newtheorem{corollary}[theorem]{Corollary}
\theoremstyle{definition}

\newtheorem{assumption}{Assumption}[section]

\newtheorem{remark}{Remark}[section]
\numberwithin{equation}{section}


\begin{document}

\setcounter{page}{1}
\title[\tiny Proxy-SVARs with non-Gaussian shocks]{Identification and Inference in proxy-SVARs with non-Gaussian Shocks}

\author[]{Paritosh Shankarrao Junare}
\thanks{Correspondence: Department of Economics, University of Bologna, Italy. Email: paritosh.junare@unibo.it. I am immensely grateful to Luca Fanelli, Giuseppe Cavaliere and Giovanni Angelini for their guidance and support throughout this work. I am also thankful to Markku Lanne, Jani Luoto, Mika Meitz, Savi Virolainen and Filippo Ferroni for their insightful comments. Discussions with Raffaela Giacomini and Lutz Kilian at Nordic Econometric Meeting 2026; Silvia Sarpietro and Kirill Ponomarev at CEPR Job Market Bootcamp 2026, Tuscany; Giovanni Caggiano as a discussant at SIdE WEEE 2026, and other participants at these conferences are gratefully acknowledged.}

\dedicatory{\today}

\begin{abstract}

Two frequent approaches for identifying structural VARs are external instruments, which carry economic content but are often weak, and non-Gaussianity of the shocks which provides statistical identification but carries no economic meaning. We combine the two strategies in a single generalized method of moments framework. This hybrid approach point-identifies the target shocks while also identifying the non-target shocks up to sign and ordering. Under suitable rank conditions, the higher-order moments of the structural shocks anchor the identification uniformly over the instrument strength. Consequently, under local-to-zero proxy relevance, estimators of the dynamic causal effects remain consistent, and standard asymptotic inference remains valid. Moreover, the Anderson-Rubin confidence sets are substantially narrower than their instrument-only counterparts. The hybrid estimator is also more efficient than either source of identification used in isolation: at any fixed proxy relevance, even a weak instrument increases efficiency of the estimator through its covariance with the non-Gaussian moment block. Under local deviations from proxy exogeneity, we provide asymptotic bias bounds and show that stronger non-Gaussianity of the shocks compresses the bias. Finally, the over-identified structure yields two mutually orthogonal specification tests, for proxy exogeneity and validity of the higher-order moment conditions. We derive their limiting distributions and provide a bootstrap procedure for finite-sample critical values. Monte Carlo simulations and two applications with identification of oil news-shock and a Euro-area MP shock demonstrate the potential of our framework.
\end{abstract}
\maketitle

\section{Introduction}
\label{sec:introduction}
The identification and estimation of dynamic causal effects is crucial to understand the transmission mechanisms of structural shocks in empirical macroeconomics. Identification through structural vector autoregressions (SVARs) generally requires restrictions on the structural impact matrix, denoted by $B$. It captures the instantaneous effects of the structural shocks on variables, and hence it is essential for the estimation of impulse response functions (IRFs). Two identification strategies have become increasingly prevalent. The first strategy employs economic theory-based external instruments or proxies: variables which are \enquote{external} to the VAR estimation, and are correlated with the target shocks (relevance condition) and orthogonal to other structural shocks (exogeneity condition), see \citet{mertens_2013, stock_identification_2018}. The second strategy exploits the non-Gaussianity of the structural shocks to obtain statistical identification: with at most one Gaussian shock, certain higher-order moment conditions recover $B$ up to sign and column permutations without additional restrictions, see \citet{comon_independent_1994, lanne_gmm_2021, keweloh_generalized_2021}.


On one hand, in proxy-SVARs, an external instrument is only as good as its correlation with the shock it is
meant to isolate. Under local-to-zero relevance, see \citet{staiger_instrumental_1997}, the estimates are inconsistent,
and standard asymptotic inference is unreliable as detailed in \citet{montiel_olea_inference_2021}. \citet{ramey_chapter_2016} notes that weak instruments are pervasive in the literature, and \citet{lewis_weak_2026} show that conventional
critical-value thresholds are ill-suited for assessing weak-instrument bias in
IRFs. Weak instruments are therefore empirically typical, and routinely
under-diagnosed; Appendix~\ref{appendix:evidencelit}
(Table~\ref{tab:weak_instruments}) lists prominent applications in which this
sensitivity is visible.
On the other hand, a large body of evidence in macro-financial literature finds that the structural shocks are pronouncedly non-Gaussian. The recovered structural shocks reject Gaussianity in favor of heavy-tailed distributions as shown in \citet{gourieroux_statistical_2017}, \citet{ludvigson_uncertainty_2021}; see
Appendix~\ref{appendix:evidencelit} (Table~\ref{tab:nongaussianity}). However, as noted earlier, this is a source of statistical identification where we still rely on labeling conventions to obtain meaningful economic interpretation of the shocks. Both of these identification strategies therefore face an obstacle that the other can overcome.

We propose a \enquote{hybrid} identification\footnote{We term the identification \emph{hybrid} because the two merged identification channels are of different kinds: one economic (external instruments), and one statistical (the non-Gaussianity of the shocks).} scheme where we combine the two channels within a single generalized method of moments (HGMM) framework that stacks proxy exclusion restrictions with higher-order moment conditions into a single
over-identified system. With $k$ instruments, under standard sign normalizations, we show global point identification of the target $k$ columns of $B$, and the non-target columns are identified up to signed permutations. The identification is {global} in the sense of \citet{komunjer_global_2012}: non-Gaussianity
eliminates the continuous rotational manifold of $B$, reducing identification
to a finite discrete set $\mathcal{B}$; and then the valid proxies select a
unique element from $\mathcal{B}$ delivering global point
identification of the target columns of $B$.

We provide three main results that directly motivate the
HGMM estimator as a preferred alternative to proxy-only approach for applied
researchers. 

\textit{First}, the HGMM estimator remains robust to weak proxies. Under
local-to-zero relevance, the proxy-only
estimator is inconsistent, whereas under certain rank conditions, the
HGMM estimator remains consistent with uniformly valid standard asymptotic inference. This result is explained by the fact that the non-Gaussian (NG) channel acts as an identification anchor. Since the Jacobian matrix pertaining to the NG moment block can remain full rank regardless of proxy strength, the joint information matrix remains bounded away from
zero even as the proxy relevance degenerates. This allows practitioners to rely on standard inference methods even in the presence of weak instruments where proxy-SVAR inference breaks down.

\textit{Second}, the joint estimator is asymptotically at least as efficient as the
proxy-only estimator: the difference in their asymptotic variances is positive semidefinite. The efficiency gain has a closed form. The joint information matrix decomposes additively into the
information obtained from the proxy exclusion conditions and the information the
non-Gaussian moments add {after} partialing out the identifying directions
already spanned by the valid proxies. The gain is strict when the non-Gaussian moments identify parameters along
directions that the proxy block does not capture. Interestingly, one might expect the hybrid estimator to collapse onto the non-Gaussian GMM estimator if the proxy is very weak. It does not. The covariance between the two moment blocks continues to sharpen the weighting matrix, so at any fixed proxy relevance its variance is smaller than that of the non-Gaussian estimator alone. This generates an efficiency gain that a researcher who entirely discards the proxies would forego.

\textit{Third}, when the instrument locally violates the exogeneity condition, the resulting first-order bias is inversely proportional to the minimum eigenvalue of the
NG information matrix. Stronger deviations from Gaussianity therefore compress the
upper bound on the bias directly, delivering a partial mis-specification
robustness that proxy-only approaches do not have.
On all three margins, therefore, augmenting a proxy-SVAR with non-Gaussian identification weakly dominates the standalone alternatives: precision improves under correct specification, consistency and standard inference survive a weak proxy, and the bias is bounded when the proxy is locally invalid.

Another contribution is of separate, valid specification
tests for each identification channel. The standard $J$-statistic is invalid in this setting: restricting the criterion
to the NG block alone admits multiple observationally equivalent
solutions. This implies that the rotation parameters would lack a unique probability limit
and the $\chi^2$ limiting distribution of the test statistic is not
restored without an anchor. Hence, for both identification channels, we construct
mutually orthogonal test statistics from the \emph{joint} GMM estimate. By orthogonalizing each block's sample moments against the other, we obtain
separate statistics for each identification channel, $J_{\mathrm{prx}}$ and
$J_{\mathrm{NG}}$, such that neither test is distorted by the other moment block \citep{newey_generalized_1985}. We also correct for the shared parameter estimation uncertainty in their asymptotic variances. This implies that, under the joint null of valid identification, the limiting distributions of the test statistics follow $\chi^2$ with degrees of freedom equal to
the number of moment conditions in the relevant block.

However, we note that asymptotic $\chi^2$ critical values can be unreliable because higher-order cumulant estimation in the
optimal weighting matrix may induce substantial finite-sample distortion.
Hence, we adopt the residual-based
moving block bootstrap of \citet{bruggemann_inference_2016} that jointly
resamples blocks of estimated residuals-proxy pairs,
preserving the cross-dependence structure that proxy identification
exploits. It also allows for conditional heteroskedasticity in the structural shocks. \citet{hall_bootstrap_1996} recentering ensures
that the bootstrap data-generating process satisfies the over-identifying
restrictions exactly at the estimated value, so that bootstrap
critical values adapt to the effective degrees of freedom.


Monte Carlo evidence from a three-variable VAR model validates the theoretical
results. The estimates are consistent and standard inference remains valid irrespective of the strength of the proxies. For identification robust inference, we also show the HGMM estimator provides narrower Anderson-Rubin (AR) confidence sets relative to the inference based only on proxy exclusion conditions. The joint criterion eliminates the failure of proxy identification without sacrificing 
the validity of the identification-robust coverage. For the specification tests, under the null hypothesis of valid identification, all three test statistics
$J_{\mathrm{joint}}$, $J_{\mathrm{prx}}$, and $J_{\mathrm{NG}}$
exhibit close to nominal size. Under failure of proxy exogeneity, the power ordering
$J_{\mathrm{prx}} \gg J_{\mathrm{NG}}$ holds
uniformly across sample sizes: the dominance of $J_{\mathrm{prx}}$
is consistent with violation of proxy exogeneity as the primary source of
mis-specification, with $J_{\mathrm{NG}}$ inheriting relatively lower, albeit false, power
only through contamination via the shared estimate. Also, $J_{\mathrm{prx}}$ provides a valid formal test of
instrument orthogonality even when the instrument is weak. 

We demonstrate the framework by revisiting identification of an oil-news shock and a Euro-area monetary policy shock. For the oil news shock of
\citet{kanzig_macroeconomic_2021}, estimated under the instrument correction of
\citet{kilian_how_2024}, the hybrid criterion delivers substantially narrower
Anderson--Rubin confidence sets than the proxy alone. The median reduction in set length is $71.6\%$, and are sharp enough to
characterize the shock as predominantly oil {demand} news shock. Notably for the real oil price, the confidence set for the on-impact estimate contracts by a factor of nearly ten once the higher-order moment conditions are added. For a Euro-area monetary policy shock identified from the
high-frequency surprises of \citet{altavilla_measuring_2019}, the gains are more modest with median reduction in set length of $26\%$. Together the illustrations show that the non-Gaussianity of the structural shocks sharpen the inference when the higher-order
moment conditions are informative enough to offset the cost of the additional
degrees of freedom.

The remainder of the paper is organized as follows.
Section~\ref{sec:method} presents the hybrid
GMM identification, asymptotic properties of the estimator and derives joint limiting distributions of the structural parameters; Section \ref{sec:weak_misp} proves consistency with standard asymptotic inference under weak proxies, and asymptotic bias bounds under local violations of proxy exogeneity; 
Section~\ref{sec:tests} develops the mutually orthogonal specification
tests and the required bootstrap procedure for finite-sample critical values;
Section~\ref{sec:results} reports Monte Carlo evidence; 
Section~\ref{sec:empirical} presents the empirical illustrations and
Section~\ref{sec:conclusion} concludes. Appendix \ref{appendix:preliminaries} consolidates preliminary notations and conventions used in the paper. A brief review of related
literature follows next.
\subsection{Related Literature} \label{sec:related_lit}
Several recent papers combine multiple SVAR identification strategies. Within a
Bayesian framework, \citet{baumeister_2015} combine sign restrictions with prior
information; \citet{arias_inference_2021} develop an algorithm for proxy-SVARs with
traditional zero and sign restrictions; \citet{giacomini_narrative_2022} provide
robust inference for proxy-SVARs with narrative restrictions;
\citet{carriero_blended_2024} blend heteroskedasticity with sign/narrative
restrictions and instrumental variables; and \citet{angelini_invalid_2024} combine external instruments with changes in volatility regimes.

Most closely related studies are the over-identified GMM proxy-SVAR of
\citet{gregory_us_2024}, and its
generalization by \citet{bruns_avoiding_2025} which relaxes the proxy count
by focusing the criterion on the proxy-targeted columns. Both, however,
draw their over-identifying content solely from \emph{second-moment}
uncorrelatedness, which vanishes when only a single shock is targeted and offers no identifying power for the
non-targeted columns. We instead draw identifying information from
higher-order moments under non-Gaussianity. This implies that even with a single proxy we obtain a test for proxy validity as the non-Gaussian moment block supplies over-identifying information. We further show robustness to weak proxy relevance where uncorrelatedness-based
approaches lose identifying content entirely.
Relative to the Bayesian non-Gaussian SVAR of \citet{keweloh_estimating_2025},
which combines statistical identification with potentially endogenous proxies,
ours is a frequentist approach which provides explicit weak instrument robustness,
efficiency results, formal specification tests, and an explicit bias bound under
local exogeneity violations. 

When instruments are weak, \citet{montiel_olea_inference_2021} document severe
under-coverage of standard proxy-SVAR confidence sets and propose Anderson--Rubin
(AR) test inversion for robust inference, \citet{angelini_identification_2024}
restore standard asymptotics for weakly instrumented target shocks by requiring
strong auxiliary instruments for the non-target shocks. The HGMM framework remains
robust under local-to-zero proxy relevance through a different mechanism: the
non-Gaussian Jacobian can retain full column rank as the proxy relevance
collapses, keeping the GMM information matrix bounded away
from zero in the limit. We consequently obtain substantially narrower AR confidence sets than the proxy-only criterion, comparable in length to the standard HGMM intervals. An applied researcher therefore needs neither additional instruments nor the heavy computation of AR sets to retain valid inference.

In statistical identification, \citet{comon_independent_1994} establishes that
non-singular linear mixing matrices are identified up to sign and column
permutation when the sources are mutually independent with at most one Gaussian.
\citet{keweloh_generalized_2021} provides a unified co-skewness and co-kurtosis
moment structure, and \citet{lanne_gmm_2021} construct a GMM estimator based on
fourth-order co-kurtosis conditions; \citet{lanne_identifying_2023} show that
$n(n-1)/2$ symmetric co-kurtosis conditions globally identify the impact matrix
when at least $n-1$ shocks have excess kurtosis of a common sign; the condition
we adopt for global identification.
In our moment set the co-skewness conditions deliver local identification and the
common-sign co-kurtosis conditions deliver global identification. We combine them
with proxy exclusion restrictions into a joint criterion and develop
Neyman-orthogonal specification tests. \citet{lewis_identification_2024} provides a
comprehensive review of statistical identification.

\section{Hybrid GMM methodology}\label{sec:method}
This section develops the hybrid GMM (HGMM) framework. We parameterize the impact matrix $B$ as a function of Givens rotation angles, $\phi$, as in \citet{matteson_dynamic_2011}, and estimate the
structural parameters by minimizing an efficiently weighted moment
criterion using residuals from a first-stage OLS estimation.
Combining the two channels delivers global point identification of the target
columns of $B$ (Proposition~\ref{prop:identification_unified}). The joint asymptotic distribution of the full parameter vector,
including autoregressive coefficients, is derived. We characterize the efficiency gains (Proposition \ref{prop:efficiency}) of the hybrid estimator relative to the
proxy-only estimator. The
asymptotic variance of $\hat{\phi}$ is corrected for the
generated-regressor distortion from first-stage estimation error
as noted in \citet{murphy_estimation_1985}, and the standard errors for all elements of $B$ and the structural
impulse-response functions follow with standard derivations from \citet{lutkepohl_2005}.

\subsection{SVAR setup} \label{subsec:svarparam}
We consider a structural VAR(p) model:
\begin{equation}
    y_t = \nu + A_1 y_{t-1} + A_2 y_{t-2} + \cdots + A_p y_{t-p}
    + B\varepsilon_t, \qquad t = 1, \dots, T,
\end{equation}
where $y_t$ is an $n \times 1$ vector of endogenous variables, $\nu$ is a vector of intercepts and $A_j$ are $n \times n$ autoregressive
matrices. The reduced-form innovations, $u_t$ and structural shocks are linearly related through the
$n \times n$ non-singular impact matrix $B$:
\begin{equation}
    u_t = B \varepsilon_t
\end{equation}
where $u_t$ is a sequence of serially uncorrelated reduced-form
disturbances with $\mathbb{E}[u_t] = 0$ and covariance matrix $\Sigma_u := \mathbb{E}(u_t u^{'}_t)$.
We further assume $y_t$ is stable, i.e.
\begin{assumption}\label{as:stable_var}
 $\det A(z) = \det (I_n - A_1z - \cdots - A_pz^p) \neq 0, \quad |z| \leq 1$ 
\end{assumption}
To span the space of observationally equivalent impact matrices, let $H$
denote the lower-triangular Cholesky factor satisfying $HH' = \Sigma_u$.
We adapt the notation of \citet{matteson_dynamic_2011} and parameterize $B = HQ(\phi)$,
where $\phi \in \Phi \equiv (-\pi, \pi]^{p_\phi}$, $Q(\phi) \in SO(n)$ is an orthogonal rotation matrix governed by the
$p_\phi \times 1$ ($p_\phi = n(n-1)/2$) vector of rotation angles $\phi$. $SO(n)$ is a subgroup within the group of all $n \times n$ orthogonal matrices $O(n)$, with determinant equal to $+1$. There exists a unique mapping of $Q \in SO(n)$ into $\phi \in (-\pi, \pi]^{p_{\phi}}$, continuous if either all elements on the main diagonal of $Q$ are positive or all elements of $Q$ are non-zero. For a detailed exposition, see \citet{matteson_dynamic_2011,matteson_independent_2013}. Hence, the structural
shocks are recovered as
\begin{equation}
    \varepsilon_t(\phi) = Q(\phi)' H^{-1} u_t = B^{-1} u_t.
\end{equation}

\subsection{Proxy identification}\label{subsec:proxyidentification}

Consider $m_t$ as a $k \times 1$ vector of external instruments. Identification through instruments requires that each proxy $m_{l,t}$ satisfies a relevance
condition, i.e. nonzero correlation with its target shock $\varepsilon_{j,t}$ and
an exogeneity condition to all remaining shocks \citep{stock_identification_2018}:
\begin{equation}
    \mathbb{E}[m_{l,t}\, \varepsilon_{j,t}] = \psi_l \neq 0, \qquad
    \mathbb{E}[m_{l,t}\, \varepsilon_{i,t}] = 0 \quad \text{for } i \neq j.
    \label{eq:proxy_conditions}
\end{equation}
Hence, we formally assume:
\begin{assumption} \label{as:validproxy}
    \textit{There exists a $k\times1$
    instrument vector $m_t$, with $k \geq 1$, satisfying
    \[
        \mathbb{E}[m_t \varepsilon_t'] = \bigl(\,\Psi \; ,\; 
        0_{k\times(n-k)}\,\bigr),
    \]
    with targeted shocks indexed $1,\ldots,k$ and
    $\Psi = \operatorname{diag}(\psi_1,\ldots,\psi_k)$ non-singular}
\end{assumption}
These conditions yield $q_{\mathrm{prx}} = k(n-1)$ (strict) exclusion restrictions. The corresponding sample moment block is
\begin{equation}
    g_{P,t}(\phi) =
    \begin{pmatrix}
    m_{1,t}\cdot\varepsilon_{i,t}(\phi)\\
    \cdots\\
    m_{k,t}\cdot\varepsilon_{i,t}(\phi)\\
    \end{pmatrix}, 
    \qquad l \in \{1,\dots,k\},\quad i \neq \mathrm{target}(l).
\end{equation}

\subsection{Identification through non-Gaussianity}\label{subsec:nongaussianconditions}

The statistical identification of SVAR models via non-Gaussianity exploits the information contained in the higher-order moments of the structural shocks.
This approach yields identification of $B$ up to sign and column permutation, as established by \citet{comon_independent_1994}. For a detailed review, we refer the reader to \citet{hyvarinen_independent_2013} and \citet{lewis_identification_2024}. Usually, the key maintained assumption is that the structural shocks, $\varepsilon_t$ are mutually independent, but recent literature has shown that it is not a necessary condition for identification, see \citet{mesters_non-independent_2024,lanne_gmm_2021, guay_identification_2021}. The identifying restrictions require only that the structural shocks be mutually uncorrelated with vanishing higher-order co-moments, together with a non-Gaussianity condition on their marginals; full statistical independence is not needed\footnote{Mutual independence would additionally force all cross-cumulants of every order to vanish, whereas only the third- and fourth-order co-moments enter the moment conditions, so these are imposed directly. This is strictly weaker than independence and does not require recovering independent components from the reduced-form innovations \citep{kilian_structural_2017}; see Remark~\ref{rem:scope_nongaus}.}.  Formally, we assume:
\begin{assumption}\label{as:indep}
\textit{$\varepsilon_t$ is strictly stationary,
with components of mean zero, unit variance, mutual uncorrelatedness
$\mathbb E[\varepsilon_{i,t}\varepsilon_{j,t}]=0$ for $i\neq j$, finite fourth moments,
and diagonal third- and fourth-order cumulant tensors,}
\begin{equation}
\operatorname{cum}(\varepsilon_{i,t},\varepsilon_{j,t},\varepsilon_{k,t})=0 \ \ \text{unless } i=j=k,    
\end{equation}
\begin{equation}
\operatorname{cum}(\varepsilon_{i,t},\varepsilon_{j,t},\varepsilon_{k,t},\varepsilon_{l,t})=0 \ \ \text{unless } i=j=k=l,    
\end{equation}
\textit{with diagonal elements $\kappa_{3,k}=\mathbb{E}[\varepsilon_{k,t}^3]$ and $\kappa_{4,k}=\mathbb{E}[\varepsilon_{k,t}^4]-3$.}
\end{assumption}
\begin{assumption}\label{as:gau}
\textit{The structural shocks satisfy: \\
\textnormal{(a)} at most one has zero
skewness, i.e.\ $\kappa_{3,m}=\mathbb{E}[\varepsilon_{m,t}^3]\neq 0$ for at least
$n-1$ shocks; and \\
\textnormal{(b)} at least $n-1$ have nonzero excess kurtosis of a
common sign, i.e.\ $\kappa_{4,m}=\mathbb{E}[\varepsilon_{m,t}^4]-3$ is of equal sign
across the non-mesokurtic shocks.}
\end{assumption}

Both parts of Assumption \ref{as:gau} can be read as deviations from the Stein
identity function: for a standard Gaussian $X$, $\mathbb{E}[Xh(X)]-\mathbb{E}[h'(X)]=0$
for every continuously differentiable $h$, the quantities $\kappa_{3,k}$ and $\kappa_{4,k}$ in Assumption
\ref{as:indep} are exactly the second and third-order deviations. Part (i) of the Assumption \ref{as:gau} is standard, Part (ii) is slightly stronger but mild in practice. Apart from the listed studies in Table \ref{tab:nongaussianity} and as noted by \citet{lanne_identifying_2023}, macro-financial data are mostly leptokurtic with few exceptions. Excess kurtosis is
positive for asset returns, exchange rates and high-frequency policy surprises,
as well as for output and commodity prices at monthly and quarterly
frequency, over essentially any sample containing a
disruption to the global economy, such as Global Financial Crisis 2007-08, COVID-19. Also, with conditionally heteroskedastic shocks, positive excess kurtosis follows mechanically\footnote{For $\varepsilon_{m,t}=\sigma_{m,t}
z_{m,t}$ with $z_{m,t}$ independent of the scale, $\kappa_{4,m}\ge
\kappa_{4,z_m}$ with strict inequality whenever $\sigma_{m,t}$ is
non-degenerate, and $\kappa_{4,m}>0$ when $z_{m,t}$ is Gaussian. The variations in shocks' volatility therefore shift excess kurtosis upward for every shock and cannot by
itself generate a mixture of signs.}. Negative excess kurtosis instead requires
short-tailed or bounded shocks, and part (b) fails only if at least two of these
occur in the same system, since $n-1$ shocks need only share a sign.

We operationalize these assumptions by mapping the relevant set of higher-order moment conditions into vector $g_{NG,t}(\phi)$. For an $n$-variable VAR, we adopt a baseline parameterization with
$q_{\mathrm{NG}} = n(n-1) + n(n-1)/2$ moment conditions:

\begin{itemize}
    \item \textit{Co-skewness (Group 1):} The $n(n-1)$ third-order
          cross-moment conditions
          $\mathbb{E}[\varepsilon_{i,t}^{2}\,\varepsilon_{j,t}] = 0$
          for all ordered pairs $(i,j)$ with $i \neq j$.  

    \item \textit{Symmetric co-kurtosis (Group 2):} The
          $n(n-1)/2$ fourth-order conditions
          $\mathbb{E}[\varepsilon_{i,t}^{2}\,\varepsilon_{j,t}^{2}] = 1$
          for all $i < j$.  
\end{itemize}

Two properties of these moment conditions are worth making explicit. First, they are {unconditional} moment restrictions that hold under Assumption~\ref{as:indep} without independence: uncorrelatedness and the diagonal third- and fourth-order cumulant tensors give $\mathbb{E}[\varepsilon_{it}^2\varepsilon_{jt}]=0$ and $\mathbb{E}[\varepsilon_{it}^2\varepsilon_{jt}^2]=1$ directly. Under the dependence conditions of Section~\ref{ass:regularity}, we allow for
idiosyncratic conditional heteroskedasticity in the shocks (ARCH/GARCH, stochastic
volatility): each $\varepsilon_{i,t}$ is a martingale difference sequence, serially
uncorrelated, while its conditional variance is serially dependent. Second, the
central limit theorem for the quartic conditions requires
$\mathbb E\lVert\varepsilon_t\rVert^{8+4\delta}<\infty$; under cross-sectional
independence the relevant cross-moments factorize and $4+2\delta$ marginal moments
suffice. The general requirement is satisfied by the non-Gaussian parametric families
considered (normal-inverse Gaussian, $\chi^2$, and Student-$t$ with degrees of freedom
exceeding $8+4\delta$).

\begin{remark}[Scope of the non-Gaussian moment block] \label{rem:scope_nongaus}
Assumption~\ref{as:indep} already relaxes full statistical independence: it requires only mutual uncorrelatedness and diagonal third- and fourth-order cumulant tensors. There are other sets of moment conditions in the literature, resting on different maintained assumptions: \citet{lanne_gmm_2021} use co-kurtosis conditions under mutual orthogonality, and \citet{keweloh_generalized_2021} uses all third- and fourth-order moments under independence.

We do not rank these strategies by identification strength. The central claim is that \emph{any} non-Gaussian identified SVAR, once expressed as sufficiently identifying moment conditions under appropriate assumptions, admits augmentation to proxy exclusion conditions within a unified GMM criterion. To illustrate, consider the assumptions of \citet{liu_non-gaussian_2025}: (i) all structural shocks have zero co-skewness and zero co-kurtosis, except that symmetric co-kurtoses may be nonzero and no two shocks share identical fourth-moment relationships, i.e.\ $\sum_{i=1}^{n}\mathbb{E}(\varepsilon_{it}^2 \varepsilon_{jt}^2) \neq \sum_{i=1}^{n}\mathbb{E}(\varepsilon_{it}^2 \varepsilon_{kt}^2)$ for all $j \neq k$; and (ii) at most one structural shock has both zero skewness and zero excess kurtosis. Under these assumptions, our symmetric co-kurtosis conditions can be replaced by asymmetric co-kurtosis conditions. The resulting moment conditions identify $B$ up to sign and column permutation, see \citet{liu_non-gaussian_2025, mesters_non-independent_2024}. These assumptions are strictly weaker than full independence and also allow for common volatility processes where the structural shocks exhibit correlated volatilities.

The analysis of the next section therefore takes $g_{NG,t}(\phi)$ as given under Assumptions \ref{as:indep} and \ref{as:gau}, and characterizes its incremental identifying information relative to the proxy exclusion conditions.
\end{remark}

\subsection{Global identification of $B$}\label{subsec:indeterminacy}

By combining the proxy exclusion and sufficient higher-order moment conditions within the GMM objective, we achieve global point identification of the target columns of $B$, as
established in Proposition \ref{prop:identification_unified}.

\begin{proposition}[Point Identification with external instruments and non-Gaussian shocks]
\label{prop:identification_unified}
Let $u_t = B_0\varepsilon_t$ with $B_0 \in \mathbb{R}^{n\times n}$ non-singular, the VAR satisfies Assumption~\ref{as:stable_var}, and let $\varepsilon_t$ and
the $k\times 1$ proxy vector $m_t$, $1 \leq k \leq n$, satisfy
Assumptions~\ref{as:validproxy}--\ref{as:gau}. 
Partition
$B_0 = [B_{0,1}\;\; B_{0,2}]$, where $B_{0,1}$ collects the $k$ target columns
and $B_{0,2}$ the remaining $n-k$ columns, and impose the sign normalization
\begin{equation}
  [B_0]_{ii} > 0, \qquad i = 1,\ldots,k,
  \label{ass:norm_u}
\end{equation}
For $d \geq 1$, let $\mathcal{P}_d$ denote the set of $d\times d$ permutation
matrices and $\mathcal{D}_d^{\pm}$ the set of $d\times d$ diagonal matrices with
diagonal entries in $\{-1,1\}$. Then:
\begin{enumerate}[label=(\roman*)]
  \item\label{res:target} \textbf{(Target shocks identification)} $B_{0,1}$ is globally
  point-identified for all $1 \leq k \leq n$.

  \item\label{res:residual} \textbf{(Non-target shocks identification)} If $k \geq n-1$,
  $B_0$ is globally point-identified. If $1 \leq k < n-1$, $B_{0,2}$ is
  identified up to the sign and column permutation, with identified set
  \[
    \mathcal{E}_{n-k} = \bigl\{ B_{0,2}P\Lambda :
      P \in \mathcal{P}_{n-k},\; \Lambda \in \mathcal{D}_{n-k}^{\pm} \bigr\},
    \qquad |\mathcal{E}_{n-k}| = (n-k)!\cdot 2^{\,n-k}.
  \]
\end{enumerate}
\end{proposition}
\noindent \textit{Proof}: See Appendix \ref{appendix:proof_identification_unified}.

Proposition \ref{prop:identification_unified} subsumes both \enquote{full} and \enquote{partial} identification\footnote{The notion of identification employed is that of
{global} identification in the sense of \citet{komunjer_global_2012}: $\phi_0$ is globally identified if the population moment
conditions admit $\phi_0$ as their unique solution over all of $\Phi$. This
is strictly stronger than the local identification condition of
\citet{rothenberg_identification_1971}, which only rules out a continuum through $\phi_0$} as special cases of a unified 
framework\footnote{\citet{angelini_exogenous_2019} establish 
necessary and sufficient conditions for point identification in 
SVAR models with $r$ instruments identifying $g \leq r$ structural 
shocks. We do not pursue the general case.}.
When $k < n-1$, 
the target sub-matrix $B_{0,1}$ is point-identified while the 
remaining $B_{0,2}$ remains identified only up to a finite signed 
permutation equivalence class. In both cases, identification of 
the target columns is achieved through the joint restrictions 
imposed by proxy exclusion and higher-order moment conditions. As we will see in Section \ref{subsec:weakproxy}, this result 
remains valid under local-to-zero proxy relevance, that is, 
when the proxy instrument is arbitrarily weak provided that the 
Jacobian of the non-Gaussian moment block retains full column rank at the true parameter.

This proposition also resolves an interpretive limitation of pure 
statistical identification. Through non-Gaussianity, we recover the 
columns of $B_0$ up to signed permutation indeterminacy, with 
no intrinsic mapping between statistically identified components 
and economically meaningful structural shocks. The proxy 
instrument conditions eliminate this indeterminacy for 
the target columns. The external instrument, whose economic 
interpretation is grounded in theoretical 
prior knowledge, selects the unique permutation and sign 
normalization consistent with both the higher-order moment 
conditions and the proxy relevance and orthogonality 
conditions. Economic identification of the target shocks is 
therefore achieved as a direct consequence of the moment 
conditions themselves, rather than through post-estimation 
labeling conventions applied to anonymous statistical 
components.

\subsection{GMM estimation}\label{subsec:twostepgmm}

The unified, over-identified system combines $q = q_{NG} +
q_{\mathrm{P}}$ moment equations into the stacked average:
\begin{equation}\label{eq:moment_conditions}
    \bar{g}_T(\phi) = \frac{1}{T} \sum_{t=1}^T
    \begin{pmatrix}
        g_{NG,t}(\phi) \\[2pt] g_{P,t}(\phi)
    \end{pmatrix}.
\end{equation}
We can estimate it using standard \citet{hansen_1982} two-step GMM. The
first-step estimator minimizes the unweighted criterion:
$\hat{\phi}_1 = \arg\min_\phi\, \bar{g}_T(\phi)'\bar{g}_T(\phi)$.
The first-step residuals are used to construct the optimal weighting matrix:
\begin{equation}
    \hat{S}_T = \frac{1}{T} \sum_{t=1}^T
    \bigl(g_t(\hat{\phi}_1) - \bar{g}_T(\hat{\phi}_1)\bigr)
    \bigl(g_t(\hat{\phi}_1) - \bar{g}_T(\hat{\phi}_1)\bigr)'.
\end{equation}
The second-step estimator minimizes the weighted objective function, yielding:
\begin{equation}\label{eq:estimator}
    \hat{\phi}_T = \arg\min_\phi\; \bar{g}_T(\phi)'\,\hat{S}_{T}^{-1}\,\bar{g}_T(\phi).
\end{equation}
Under the mixing condition of Section~\ref{ass:regularity}, the optimal weighting matrix is the \citet{newey_simple_1987} HAC\footnote{$\hat{S}_{HAC, T} = \hat{\Gamma}_0 + \sum_{i=1}^{T-1}\omega_{i,T}(\hat{\Gamma}_i+ \hat{\Gamma}'_i)$, where $\hat{\Gamma}_i$ is a consistent estimator of ${\Gamma}_i$, the $i$th autocovariance matrix of the moment conditions, ${g}_t(\phi_0)$, with some kernel weights $\omega_{i,T}$.} estimator with data-dependent bandwidth of \citet{andrews_heteroskedasticity_1991}. Under serial independence of the shocks, the autocovariances vanish and $\hat{S}_T$ reduces to the covariance matrix above.
For subsequent derivations, we conformably partition the optimal inverse-weighting matrix as:
\begin{equation}
  S = \begin{pmatrix} S_{NN} & S_{NP} \\ S_{PN} & S_{PP} \end{pmatrix},
  \qquad S_{NP} = S_{PN}',
\end{equation}
Also, we define the population Jacobian as
\begin{equation} \label{eq:jacobian}
    G(\phi) = \frac{\partial g(\phi)}{\partial\phi'}
            \in \mathbb{R}^{q\times p_\phi}, \qquad
  G \equiv G(\phi_0) = \bigl(G_{NG}', G_P' \bigr)',
\end{equation}
where $G_P = \partial g_P(\phi_0)/\partial\phi'$ and
$G_{NG} = \partial g_{NG}(\phi_0)/\partial\phi'$, and a consistent estimator of Jacobian is $\hat{G}_T(\phi)$.

\subsection{Joint inference with VAR parameters}\label{sec:inference}

With a standard set of regularity conditions, we show
consistency and asymptotic normality of the HGMM estimator of the rotation parameters $\phi$, and its efficiency gains relative to proxy-only criterion. Also, we provide joint asymptotic distribution of the estimates of autoregressive and structural rotation parameters. Throughout, let $\psi = (\theta', \phi')'$ denote the full
parameter vector, with $\theta = \mathrm{vec}(\nu, A_1, \ldots, A_p)$
collecting the reduced-form VAR coefficients and $\phi$ the structural
rotation parameters.

\subsubsection{Regularity conditions}\label{ass:regularity}
\begin{enumerate}
  \item[(1)] The parameter space for the rotation angles $\phi\in(-\pi,\pi]^{p_\phi}$
  is compact.
  \item[(2)] For some $\delta>0$, with
  $\mathcal F_{t}=\sigma\bigl((u_s',m_s')' : s\le t\bigr)$:
  \begin{enumerate}
    \item[(a)] $\{(u_t',m_t')'\}$ is strictly stationary and
    $\mathbb E[u_t\mid\mathcal F_{t-1}]=0$ almost surely;
    \item[(b)] $\{(u_t',m_t')'\}$ is $\alpha$-mixing with
    $\sum_{h\ge1}\alpha(h)^{\delta/(2+\delta)}<\infty$;
    \item[(c)] $\mathbb E\lVert u_t\rVert^{8+4\delta}<\infty$ and
    $\mathbb E\lVert m_t\rVert^{4+2\delta}<\infty$.
  \end{enumerate}
  \item[(3)] The long-run variance
  $S=\sum_{h}\mathbb E\bigl[g_t(\phi_0)g_{t-h}(\phi_0)'\bigr]$ is positive definite.
  \item[(4)] The fourth-order cumulants of $\{g_t(\phi_0)\}$ are absolutely summable,
  so that the HAC estimator $\hat{S}_T$ is consistent.
\end{enumerate}

Condition~1 guarantees that the extremum estimator is
well-defined and that a uniform LLN applies to the criterion, delivering
consistency once $\phi_0$ is the unique minimizer of the criterion function, as provided by the
identification result in Proposition \ref{prop:identification_unified}. Condition 2 governs the temporal dependence of the data. Part (a) implies $u_t$ is a martingale
difference sequence, allowing idiosyncratic ARCH/GARCH or
stochastic volatility. Parts (b)--(c) deliver a central limit theorem for
$T^{-1/2}\sum_t g_t(\phi_0)$. Condition~3 concerns the resulting optimal
weighting matrix. Since the moment conditions are serially correlated, the normalized sum has the {long-run}
variance $S=\sum_h\mathbb{E}[g_t(\phi_0)g_{t-h}(\phi_0)']$; positive definiteness of $S$ ensures the efficient weighting $S^{-1}$ and the asymptotic variance
$(G'S^{-1}G)^{-1}$ are well-defined \citep{andrews_heteroskedasticity_1991}. 

Under the regularity conditions in \ref{ass:regularity} and
Assumptions \ref{as:validproxy} - \ref{as:gau}, the estimator \eqref{eq:estimator} is
consistent and asymptotically normal, i.e.
\begin{equation} \label{eq:can}
    \sqrt{T}\,(\hat{\phi}_T - \phi_0)
  \;\xrightarrow{d}\;
  \mathcal{N}\!\bigl(0,\; \mathcal{V}_{joint}\bigr)
\end{equation}
where, $\mathcal{V}_{joint} = \bigl(G'S^{-1}G\bigr)^{-1}$ is the asymptotic covariance matrix\footnote{This result assumes the reduced-form innovations ($u_t$) are known, however in practice the innovations are estimated and hence, we include the generated-regressor correction while deriving the asymptotic distribution of $\psi$ in Section \ref{sec:joint_asymp}.} of the estimator.
See \citet{hall_generalized_2005} for details.

\subsubsection{Efficiency gains}\label{subsec:efficiencygains}
The hybrid GMM estimator combines two distinct identification strategies, yielding efficiency gains relative to proxy-based identification alone\footnote{Analogous gains
hold relative to non-Gaussian identification alone, provided
$\operatorname{rank}(\widetilde{G}_{P}) = p_\phi$. This condition is restrictive
in practice, as enough valid proxies to fully identify a moderate-dimensional
VAR are rarely available.}. The following proposition makes this comparison precise.
\begin{proposition}[Efficiency of joint GMM]
\label{prop:efficiency}
Under Assumptions \ref{as:validproxy}-\ref{as:gau}, and if $G_{NG}$ has full column rank and
$\operatorname{rank}(G_P) = p_\phi$:
\begin{equation}
  \mathcal{V}_P \;-\; \mathcal{V}_{joint} \;\succeq\; 0.
  \label{eq:eff_ineq}
\end{equation} 
where,
\begin{equation}
 \mathcal{V}_{joint} \;\equiv\; \mathcal{I}_{joint}^{-1} \;=\; \bigl(G'S^{-1}G\bigr)^{-1} \qquad \mathcal{V}_P \equiv \bigl(G_P'S_{PP}^{-1}G_P\bigr)^{-1}
  \label{eq:Vjoint}
\end{equation}
$\mathcal{V}_P$ is the asymptotic covariance matrix of the proxy-only GMM estimator.
\end{proposition}
The corollary states strict efficiency gains only if the higher-order moment conditions identify $\phi$ in directions not already identified by the proxy moments.
\begin{corollary}[Strict efficiency gain]\label{cor:strict_eff}
$ \mathcal{V}_P -  \mathcal{V}_{joint} \succ 0$ $\iff$ $\operatorname{rank}(\tilde{G}_{NG})
= p_\phi$
\end{corollary}
where
\begin{equation}\label{eq:GNG_tilde}
 \tilde{G}_{NG}\equiv G_{NG} - S_{NP}S_{PP}^{-1}G_P
\end{equation}
is the \enquote{partialed} NG Jacobian. The proof rests on an exact additive decomposition of the joint information matrix, given by:
\begin{equation}
  \mathcal{I}_{joint}
  \;=\; \underbrace{G_P'S_{PP}^{-1}G_P}_{\mathcal{I}_P}
       \;+\; \underbrace{\tilde{G}_{NG}'\Sigma_{NG}^{-1}\tilde{G}_{NG}}_{\Delta \mathcal{I}_{NG}},
  \label{eq:info_decomp}
\end{equation}
where $\Sigma_{NG} \equiv S_{NN} - S_{NP}S_{PP}^{-1}S_{PN} \succ 0$ is
the Schur complement of $S_{PP}$ in $S$, see \citet{newey_generalized_1985}. Since $\Sigma_{NG}^{-1} \succ 0$,
the increment $\Delta \mathcal{I}_{NG} \succeq 0$, and monotonicity of
inversion delivers \eqref{eq:eff_ineq}. A detailed proof can be found in Appendix \ref{appendix:efficiency}.

The decomposition \eqref{eq:info_decomp} is a GMM analog of the
Frisch-Waugh-Lovell theorem: $\Delta \mathcal{I}_{NG}$ is the information provided by the higher-order moments on $\phi$, {after} projecting out the identifying
directions already spanned by the proxy moments.
The efficiency gain is thus governed by the extent to which the two channels
identify the parameters along distinct directions.

\subsubsection{Asymptotic distribution}\label{sec:joint_asymp}

Note that the two-step estimation of Section \ref{subsec:twostepgmm}, requires the autoregressive estimates $\hat{\theta}_T$ which yield the reduced-form innovations $\hat{u}_t$ and the Cholesky
factor of the covariance matrix $\hat{\Sigma}_u$, $\hat{H} = \mathrm{chol}(\hat{\Sigma}_u)$. The second stage minimizes
the HGMM criterion over $\phi$ using the generated residuals
$\hat{\varepsilon}_t(\phi) = Q(\phi)'\hat{H}^{-1}\hat{u}_t$. Since,
the second-stage moment conditions depend on first-stage autoregressive estimates, the asymptotic
variance of $\hat{\phi}_T$ requires a generated-regressor
correction, as in \citet{murphy_estimation_1985}.

Under the regularity conditions in \ref{ass:regularity} and Assumptions \ref{as:stable_var} and \ref{as:indep},
\begin{equation}
    \sqrt{T}\,(\hat{\theta} - \theta_0)
    \xrightarrow{d} \mathcal{N}(0,\, \mathcal{V}_\theta),
    \qquad
    \mathcal{V}_\theta = \mathbb{E}[\zeta_t\zeta_t'],
\end{equation}
where $\zeta_t \equiv (\mathbb{E}[x_t x_t'])^{-1}x_t \otimes u_t$, $\otimes$ denotes the Kronecker product, $x_t = (1, y_{t-1}',\ldots,y_{t-p}')' \in \mathbb{R}^{np+1}$ are VAR regressors (including the intercept). $x_t$ is $\mathcal{F}_{t-1}$-measurable with $\mathcal{F}_{t-1} := 
\sigma(y_{t-1}, y_{t-2}, \ldots)$ is the natural filtration of $\{y_t\}$. Because $\zeta_t$ is MDS under condition 2(a) of
Section~\ref{ass:regularity}, $V_\theta=\mathbb E[\zeta_t\zeta_t']$ is also its long-run
variance; for homoscedastic residuals it
simplifies to $V_\theta=\bigl(\mathbb E[x_tx_t']\bigr)^{-1}\otimes\Sigma_u$.

The joint estimator $\hat{\psi}_T$ solves a stacked moment system identifying $\phi$
via the proxy exclusion and higher-order moment conditions. Following the two-step
estimator theory of \citet{murphy_estimation_1985} and \citet{newey_mcfadden_1994},
under the regularity conditions of Section~\ref{ass:regularity} and
Assumptions~\ref{as:stable_var}--\ref{as:gau},
\begin{equation}
    \sqrt{T}
    \begin{pmatrix} \hat{\theta} - \theta_0 \\ \hat{\phi} - \phi_0 \end{pmatrix}
    \xrightarrow{d}
    \mathcal{N}\!\left(0,\;
    \mathcal{V}_\psi \equiv
    \begin{bmatrix} \mathcal{V}_\theta & \mathcal{V}_{\theta\phi} \\ \mathcal{V}_{\theta\phi}' & \mathcal{V}_\phi \end{bmatrix}
    \right).
    \label{eq:joint_asymp}
\end{equation}
Since the second-stage moments depend on $\theta$ through the generated residuals,
first-stage estimation error propagates into $\hat{\phi}_T$. The relevant score is
the combined score $g_t+\Gamma_\theta\zeta_t$, with
$\Gamma_\theta\equiv\mathbb{E}[\partial g_t(\psi_0)/\partial\theta']$, whose long-run
variance we denote $\Omega_\psi$. The generated-regressor corrected variance block is
\begin{equation}
  \mathcal{V}_\phi
  = (G_0'S^{-1}G_0)^{-1}\,G_0'S^{-1}\,\Omega_\psi\,
    S^{-1}G_0\,(G_0'S^{-1}G_0)^{-1},
  \label{eq:vphi_corrected}
\end{equation}
estimated by its sample analogue
\begin{equation}
  \hat{\mathcal{V}}_\phi
  = (\hat{G}'\hat{S}_T^{-1}\hat{G})^{-1}\,\hat{G}'\hat{S}_T^{-1}\,
    \hat{\Omega}_\psi\,
    \hat{S}_T^{-1}\hat{G}\,(\hat{G}'\hat{S}_T^{-1}\hat{G})^{-1}.
  \label{eq:vphi_est}
\end{equation}
Since the higher-order moment conditions of the structural shocks are correlated
with the reduced-form innovations that drive the first-stage error, so the
cross-covariance between $g_t$ and $\zeta_t$, and with it $\mathcal{V}_{\theta\phi}$,
is generically nonzero. Appendix~\ref{app:genreg} gives $\Omega_\psi$,
$\mathcal{V}_{\theta\phi}$ and their consistent estimators explicitly.
\subsubsection{Inference on the impact matrix}
\label{sec:b0_inference}

Since $B = H(\theta)Q(\phi)$ is a smooth differentiable function of $\psi$, we can apply the delta method to \eqref{eq:joint_asymp}. Consider the composite Jacobian
\begin{equation}
    \mathcal{J}_{B} =
    \frac{\partial\,\mathrm{vec}(B)}{\partial\psi'} =
    \left(\frac{\partial\,\mathrm{vec}(B)}{\partial\theta'}
    \;,\;
    \frac{\partial\,\mathrm{vec}(B)}{\partial\phi'}\right)
    \in \mathbb{R}^{n^2 \times (n^2 p + n + p_\phi)}.
    \label{eq:jac_B0}
\end{equation}
Both blocks are analytically tractable. For the rotation component:
\begin{equation}
    \frac{\partial\,\mathrm{vec}(B)}{\partial\phi'}
    = (I_n \otimes H)\,\frac{\partial\,\mathrm{vec}(Q(\phi))}{\partial\phi'},
    \label{eq:jac_B0_phi}
\end{equation}
where the Givens Jacobian $\partial\,\mathrm{vec}(Q(\phi))/\partial\phi'$
has explicit closed form as a sum of antisymmetric generator matrices, as shown in
\citet{matteson_dynamic_2011}. For the Cholesky component:
\begin{equation}
    \frac{\partial\,\mathrm{vec}(B)}{\partial\theta'}
    = (Q(\phi)'\otimes I_n)\,
      \frac{\partial\,\mathrm{vec}(H)}{\partial\,\mathrm{vech}(\Sigma_u)'}\cdot
      \frac{\partial\,\mathrm{vech}(\Sigma_u)}{\partial\theta'},
    \label{eq:jac_B0_theta}
\end{equation}
where the Cholesky derivative uses the implicit differentiation formula of
\citet{magnus_neudecker_1989} via the duplication matrix $D_n$, and
$\partial\,\mathrm{vech}(\Sigma_u)/\partial\theta'$ follows from standard
theory.

The delta method applied to \eqref{eq:joint_asymp} with Jacobian $\mathcal{J}_{B}$ in
\eqref{eq:jac_B0} yields
\begin{equation}
    \sqrt{T}\,\mathrm{vec}(\hat{B}_0 - B)
    \xrightarrow{d}
    \mathcal{N}\bigl(0,\; \mathcal{J}_{B} \mathcal{V}_\psi \mathcal{J}_{B}'\bigr).
    \label{eq:b0_asymp}
\end{equation}
We can estimate the standard errors for $[\hat{B}_0]_{ij}$ as the square roots of the corresponding diagonal entries of $\mathcal{J}_{B}\hat{\mathcal{V}}_\psi \mathcal{J}_{B}'$. Once we have established the inference for the impact matrix $B$, the derivation of the asymptotic distribution of the structural impulse response functions (SIRFs) follows directly from \citet{lutkepohl_2005}, Chapter 3. We defer it to the Appendix \ref{subsec:structural_irfs}.

\section{Behavior under weak and contaminated proxies}\label{sec:weak_misp}
In this section, we provide three main results. First, Proposition~\ref{prop:localtozero} shows that under local-to-zero weak proxies, the HGMM estimator remains consistent and asymptotically normal, with the non-Gaussian moment block acting as the identifying anchor. The estimator also attains a smaller asymptotic variance than the standalone non-Gaussian GMM estimator, the gain arising from the cross-block long-run covariance between the proxy exclusion restrictions and higher-order moment conditions. Second, Proposition~\ref{prop:uniform} shows that when the higher-order moments are sufficiently informative, standard Wald inference is uniformly, rather than only pointwise, valid over the proxy relevance. Third, Proposition~\ref{prop:bias_bound_local} shows that the estimator incurs a finite
asymptotic bias when the proxies show local violations of exogeneity condition through
correlation with non-target structural shocks. We derive an explicit bound in terms of the degree of
violation of exogeneity and the strength of the non-Gaussian identification.

\subsection{Robust identification with weak proxies} \label{subsec:weakproxy}

In proxy-SVARs, identification rests on the relevance condition
$\mathbb{E}[m_t\varepsilon_{1,t}] = \psi_1 \neq 0$. As $\lVert\psi_1\rVert$
shrinks, the Jacobian of the proxy moments loses rank, and a proxy-only
estimator inherits the problems of weak identification
\citep{newey_generalized_2009, montiel_olea_inference_2021}. By augmenting the external instrument with higher-order moments of the structural shocks, we can resolve this rank deficiency. The resulting estimator is consistent, and at least as efficient as
the standalone NG-GMM estimator, even when proxy relevance is local-to-zero. The following proposition formalizes this
result.

\begin{proposition}[Behavior under local-to-zero proxy relevance]\label{prop:localtozero}
Suppose the conditions of Proposition \ref{prop:identification_unified}
hold, but proxy relevance follows the local-to-zero (Pitman) sequence
\[
    \mathbb{E}[m_{t,T}\varepsilon_t']
    = \frac{1}{\sqrt{T}}\bigl(\,C\;,\;0_{k\times(n-k)}\,\bigr),
    \qquad C = \operatorname{diag}(c_1,\ldots,c_k),\ \ c_i > 0,
\]
and let $\Sigma_{NG}$ be the Schur complement defined in
\eqref{eq:info_decomp}, suppose
$\mathcal{I}^{*}\equiv G_{NG}'\,\Sigma_{NG}^{-1}\,G_{NG}$ is positive definite. Then:
\begin{enumerate}[label=(\roman*)]
  \item \label{localtozero_infolimit} The population
  proxy Jacobian satisfies $G_{P,T}=O(T^{-1/2})$, and the joint
  information matrix:,
  \begin{equation}
      \mathcal{I}_{\mathrm{joint},T}\ \longrightarrow\ \mathcal{I}^{*}  
  \end{equation}

  \item \label{localtozero_can} The hybrid
  GMM estimator satisfies
  \begin{equation}
       \hat{\phi}\ \xrightarrow{p}\ \phi_0,
  \end{equation}
     \begin{equation}
 \sqrt{T}\,(\hat{\phi}-\phi_0)\ \xrightarrow{d}\
      \mathcal{N}\!\bigl(0,\ \mathcal{V}_{NG}^{*}\bigr),
  \end{equation}
     \begin{equation}
          \mathcal{V}_{NG}^{*}=(\mathcal{I}^{*})^{-1}
     \end{equation}
 
  \item \label{localtozero_effoverng}
  $\mathcal{V}_{NG}^{*}\preceq \mathcal{V}_{NG}
  \equiv (G_{NG}'\,S_{NN}^{-1}\,G_{NG})^{-1}$, $\mathcal{V}_{NG}$ being the asymptotic variance of
  the optimal standalone NG-GMM estimator. Equality holds if the
  cross-block long-run covariance is such that $S_{NP}=0$.
\end{enumerate}
\end{proposition}
\noindent\textit{Proof.} See Appendix \ref{appendix:proof_localtozero}.

Parts~\ref{localtozero_infolimit} and~\ref{localtozero_can} show that, even when the relevance of the proxies is local-to-zero, the hybrid estimator remains consistent and asymptotically
normal. This result is explained by the fact that identification is anchored by the non-Gaussian moment
block, whose Jacobian $G_{NG}$ does not depend on the proxy relevance.

Part~\ref{localtozero_effoverng} is an interesting result. One might expect
that, since the proxy is weak, the hybrid estimator reduces to
a standalone NG-GMM estimator and attains asymptotic variance $\mathcal{V}_{NG}$. Part~\ref{localtozero_effoverng} shows that this conjecture is an understatement. $\mathcal{V}_{NG}^{*}$
weights the NG moment block by the Schur complement,
$\Sigma_{NG} = S_{NN} - S_{NP}S_{PP}^{-1}S_{PN} \preceq S_{NN}$, i.e. the variance
after projecting out the proxy moment conditions. Hence, at any fixed proxy relevance,
the proxy provides efficiency through its covariance structure with the higher-order moments, even
when its direct identifying information is negligible. This correlation need not necessarily vanish in the
local-to-zero limit, because the proxy can be weakly correlated with the target shock
yet correlated with its higher-power transformations. For example, a proxy of the form
$m_t = \psi_{1,T}\,\varepsilon_{1,t} + \kappa\,(\varepsilon_{1,t}^{2}-1) + \eta_t$,
with $\kappa \neq 0$, responds
to the size of the realized shock regardless of its sign, so that large surprises
of either sign move the proxy in the same direction. This implies that
$\mathbb{E}[m_t\varepsilon_{1,t}^{2}] \neq 0$ even as $\psi_{1,T} \to 0$.

The pointwise results above hold for a fixed sequence of data-generating
processes. The next subsection shows that under certain conditions where identification through non-Gaussianity is sufficiently strong, standard
Wald inference is valid uniformly over a class of data-generating processes in
which proxy relevance may be arbitrarily small.

\subsection{Uniform validity of standard inference}\label{subsec:uniformvalid}
Let $\gamma=(\phi_0,\theta_0,\Psi,F)\in\mathcal G$ denote the data-generating
process, where $F$ is the joint distribution of the structural shocks and the
proxy. All population objects below ($G_{NG}$, $\Sigma_{NG}$, $S$) are
functionals of $\gamma$.

\begin{assumption}[Uniform non-Gaussian identification]\label{as:uniform_ng}
Assumptions~\ref{as:stable_var}--\ref{as:gau} hold, and $\gamma\in\mathcal P$, where
\[
    \mathcal{P}=\Big\{\gamma\in\mathcal G:\
    \sigma_{\min}\!\big(G_{NG}(\phi_0)\big)\ge \underline g,\
    \operatorname{eig}\!\big(S(\gamma)\big)\subset[\underline s,\bar s],\
    \min_{l\le k}|\psi_l|\ \ge\ \underline\psi \Big\}
\]
for some constants $\underline g>0$, $0<\underline s\le\bar s<\infty$ and
$\underline\psi>0$, with $\sigma_{\min}(\cdot)$ the smallest singular value and
$\operatorname{eig}(\cdot)$ the set of eigenvalues. 
\end{assumption}
Assumption~\ref{as:uniform_ng} requires the NG moment block to be uniformly
strongly identifying. The proxy is relevant\footnote{In the Appendix, Corollary \ref{cor:irrelevance} establishes that strict proxy irrelevance ($\Psi = 0$)
collapses identification to the non-Gaussian identified set $\mathcal{B}$. So indeed, the proxy is still required to select the target column, but only out of a discrete (not a continuum) number of choices. This allows the proxy relevance to be arbitrarily close to zero and NG moment block anchors the identification with uniformly valid standard inference.} but $\underline\psi>0$ can be arbitrarily small. Together, these conditions bound the information available for inference away
from zero uniformly over $\mathcal P$: the smallest eigenvalue of the joint
information matrix $G'S^{-1}G$ is at least $\underline g^{\,2}/\bar s$ for every
$\gamma\in\mathcal P$. This bound does not involve $\underline\psi$, so the
asymptotic variance of $\hat{\phi}_T$ stays bounded for arbitrarily weak proxies. Moreover, the conditions of Section~\ref{ass:regularity} hold for every $\gamma\in\mathcal P$, with
mixing coefficients, moment bounds and cumulant sums bounded uniformly over $\mathcal P$.

\begin{proposition}[Uniform validity of Wald inference over proxy strength]\label{prop:uniform}
Suppose Assumption~\ref{as:uniform_ng} and the regularity conditions of
Section~\ref{ass:regularity} hold. Let $R$ be a fixed $d\times p_\phi$ matrix of
full row rank, and let $\hat{\mathcal V}_\phi$ be the generated-regressor
corrected variance estimator of \eqref{eq:vphi_est}. Then, for
$H_0: R\phi_0=r$, the Wald statistic
\[
   W_T \;=\; T\,(R\hat\phi-r)'\,\big(R\,\hat{\mathcal V}_\phi\,R'\big)^{-1}\,(R\hat\phi-r)
\]
is asymptotically $\chi^2_d$ uniformly over proxy strength:
\[
   \lim_{T\to\infty}\ \sup_{\gamma\in\mathcal P}\
   \sup_{x\in\mathbb R}\
   \big|\,\mathbb P_\gamma(W_T\le x)-F_{\chi^2_d}(x)\,\big|\;=\;0 .
\]
Accordingly, the Wald confidence set
$\mathcal C_{1-\alpha}=\{r:W_T(r)\le \chi^2_{d,1-\alpha}\}$ has correct asymptotic size,
$\displaystyle \lim_{T\to\infty}\inf_{\gamma\in\mathcal P}\mathbb P_\gamma\big(R\phi_0\in\mathcal C_{1-\alpha}\big)=1-\alpha.$
\end{proposition}
\noindent\textit{Proof.} See Appendix~\ref{appendix:proof_uniform}.

Proposition~\ref{prop:uniform} is the inferential counterpart of
Proposition~\ref{prop:localtozero}. Since the information bound holds
uniformly over $\mathcal P$, the estimator is asymptotically normal and
$\hat{\mathcal V}_\phi$ is consistent uniformly in proxy strength, so the Wald
statistic retains its $\chi^2$ limit. In practice,
when the structural shocks are sufficiently non-Gaussian,\footnote{See
Appendix~\ref{appendix:evidencelit} for the empirical evidence.} conventional
delta-method standard errors for the structural parameters are valid.

\begin{remark}[The case of doubly-weak identification]
Proposition \ref{prop:uniform} is uniform over proxy strength while the
non-Gaussian block is uniformly strongly identifying
($\sigma_{\min}(G_{NG})\ge\underline g$). If non-Gaussianity is itself weak ($\sigma_{\min}(G_{NG})\to0$, i.e. near-Gaussian shocks), the DGP leaves
$\mathcal P$, and the limit is
non-standard; Wald inference is then not uniformly valid. For that regime we can rely
on the identification-robust Anderson--Rubin confidence sets which retain correct coverage uniformly
over the strength of {both} channels \citep{montiel_olea_inference_2021, hoesch_locally_2024}.
\end{remark}

\subsection{Bounded asymptotic bias under local violations of proxy exogeneity}\label{subsec:proxycontam}
In applied macroeconomic and financial research, external instruments may be
endogenous by narrative mis-classification, or the
inadvertent capture of confounding information, as discussed in
\citet{bauer_reassessment_2022, miranda-agrippino_transmission_2021}. When a
proxy co-varies with the non-target shocks, the orthogonality condition fails, i.e. the population expectation
$\mathbb{E}[g_P(\phi_0)] = \mu \neq 0$, inducing a non-vanishing asymptotic
bias in the structural parameter estimates.

The hybrid GMM framework permits a precise characterization of this bias and
its determinants. Because identification operates through two distinct
channels, the asymptotic bias is not governed solely by the degree of the violation of proxy
exogeneity but sufficiently strong non-Gaussian identification attenuates the
bias. The following proposition makes this
relationship explicit.

\begin{proposition}[Asymptotic bias bound under local violations of proxy exogeneity]\label{prop:bias_bound_local}
Suppose Assumptions~\ref{as:indep}--\ref{as:gau} hold, so $G_{NG}(\phi_0)$ has
full column rank at $\phi_0$, and let the estimator adopts the efficient weighting matrix
$W=S^{-1}$. The proxy moment block is locally contaminated,
\[
   \mathbb{E}[g_P(\phi_0)]=\mu_T=c/\sqrt{T},\quad \|c\|<\infty,
   \qquad \mathbb{E}[g_{NG}(\phi_0)]=0 .
\]
Let the partialed-out proxy Jacobian be $\tilde G_P=G_P-S_{PN}S_{NN}^{-1}G_{NG}$ and the
Schur complement $\Sigma_P=S_{PP}-S_{PN}S_{NN}^{-1}S_{NP}$. 
Also, suppose $\phi_{*,T}$
minimizes $\mathcal Q_T(\phi)=\mathbb{E}[g(\phi)]'S^{-1}\mathbb{E}[g(\phi)]$ and
denote $\mathbb{B}_T=\sqrt{T}(\phi_{*,T}-\phi_0)$ as the scaled bias. Then
\begin{equation}
  \mathbb{B}_T=-(G'S^{-1}G)^{-1}\,\tilde G_P'\,\Sigma_P^{-1}\,c+O(T^{-1/2}),
  \label{eq:local_bias}
\end{equation}
The Euclidean norm of the
scaled bias is bounded by
\begin{equation}
  \|\mathbb{B}_T\|_2\ \le\
  \frac{\big\|\tilde G_P'\,\Sigma_P^{-1}\,c\big\|_2}
       {\lambda_{\min}\!\big(G_{NG}'S_{NN}^{-1}G_{NG}\big)}
  \ +\ O(T^{-1/2}).
  \label{eq:local_bias_bound}
\end{equation}
\end{proposition}

\noindent\textit{Proof.} See Appendix \ref{appendix:bias_bound_local}.
In \eqref{eq:local_bias_bound}, the numerator $\big\|\tilde G_P'\,\Sigma_P^{-1}\,c\big\|_2$ is the raw contamination, transmitted only through the proxy moment conditions. The denominator $\lambda_{\min}\!\big(G_{NG}'S_{NN}^{-1}G_{NG}\big)$ is the curvature contributed by the non-Gaussian block alone contributes to the population criterion. The bias ceiling is therefore strictly decreasing in the strength of non-Gaussian identification: as $\lambda_{\min}\!\big(G_{NG}'S_{NN}^{-1}G_{NG}\big)$ grows, the higher-order moment conditions anchor the estimator strongly at $\phi_0$ and the distortion from the mis-specified proxy shrinks. A proxy-only estimator has no such denominator. There $\mathcal{I}_{NG} = 0$, and the analogous bound is unbounded whenever $c \neq 0$.

\section{Specification tests}\label{sec:tests}

The asymptotic theory of Section \ref{sec:inference} assumes valid
specification of both the non-Gaussian moment block and the proxy exogeneity moment
conditions. In practice, neither can be taken as given. Proxy exogeneity is
not directly testable in isolation. The exclusion restriction
$\mathbb{E}[m_{l,t}\varepsilon_{i,t}] = 0$ for $i \neq \mathrm{target}(l)$
involves unobserved structural shocks and cannot be evaluated without auxiliary
identifying assumptions. This section constructs the mutually orthogonal specification test statistics for both moment blocks, establishes their
asymptotic distributions, and describes the bootstrap procedure used to
obtain valid finite-sample critical values. Specifically, we construct three test statistics: $J_{\mathrm{joint}}$, $J_{\mathrm{prx}}$, and $J_{NG}$.

\subsection{The joint null hypothesis}\label{sec:joint_null}

Let $\bar{g}_T(\phi)$ denote the sample moment vector of \eqref{eq:moment_conditions} with $q = q_{NG} + q_P$ conditions. The
joint null hypothesis assumes that the proxy instruments satisfy the exogeneity
conditions \eqref{eq:proxy_conditions}, and the structural
innovations satisfy the higher-order moment conditions (\ref{as:indep} and \ref{as:gau}). Under this null hypothesis, the
over-identifying test statistic, as in \citet{hansen_1982}, is given by:
\begin{equation}
    J_{\mathrm{joint}}
    = T\cdot\bar{g}_T(\hat{\phi})'\,\hat{S}_T^{-1}\,\bar{g}_T(\hat{\phi})
    \xrightarrow{d} \chi^2(q - p_\phi),
    \label{eq:j_joint}
\end{equation}
where $\hat{S}_T$ is the long-run covariance estimator from
\eqref{eq:vphi_est}. The degrees of freedom equal $q - p_\phi =
(q_{NG} + q_P) - p_\phi$. Rejection of \eqref{eq:j_joint} signals mis-specification
of at least one moment block, but it does not identify its origin. We develop conditional orthogonal tests to address this lack of attribution.

\subsection{Neyman-Orthogonal tests}\label{sec:subset_tests}

Identification through non-Gaussianity recovers the rotation matrix $Q(\phi)$ only up
to a signed permutation group $\mathcal{SP}(n)$. This implies that the unrestricted non-Gaussian estimator, which excludes
the proxy moment conditions, admits $2^n \cdot n!$ observationally equivalent
solutions and does not converge to a unique probability limit. The standard $J$-statistic requires the unrestricted estimator
to be asymptotically normal around a unique $Q(\phi_0)$. Without a column anchor provided by proxy exclusion conditions to select a single element of the $\mathcal{SP}(n)$-orbit, this
condition fails and the $\chi^2$ limit distribution of the $J$-statistic is not
restored. Therefore, the standard test statistic is invalid in this setting.

Instead, we propose to construct the test statistics from the same joint GMM estimator under the joint null. The key challenge is that the raw sample moments $\bar{g}_{NG,T}$ and $\bar{g}_{P,T}$ are correlated through the shared sample, identification space and the common estimate $\hat{\phi}_T$. This means that their naive quadratic forms do not have standard chi-squared limits under the null.
We address these concerns by constructing mutually orthogonalized score statistics, as in \citet{newey_generalized_1985}.
The notations for the partitioned moment vectors, variance, and Jacobian matrices remain the same as Section \ref{subsec:twostepgmm}.

To construct a test statistic for the validity of the proxy moment conditions that
is asymptotically insensitive (under the null) to the higher-order moment conditions, we orthogonalize proxy moment vector, $\bar{g}_{P,T}$ with respect to the linear span of $\bar{g}_{NG,T}$, i.e.,
\begin{equation}
    \tilde{g}_{P,T} = \bar{g}_{P,T} - S_{PN}S_{NN}^{-1}\bar{g}_{NG,T}.
    \label{eq:gphat}
\end{equation}
The coefficient $S_{PN}S_{NN}^{-1}$ is the linear predictor of
$g_P$ given $g_{NG}$ under the long-run variance matrix $S$.
The resulting score satisfies the Neyman orthogonality condition:
\begin{equation}
    \frac{\partial}{\partial\theta_{NG}}
    \,\mathbb{E}\!\left[\tilde{g}_{P,T}(\phi_0,\theta_{NG})\right] = 0,
    \label{eq:neyman}
\end{equation}
where $\theta_{NG}$ parameterizes local perturbations of the non-Gaussian moment
block. Condition \eqref{eq:neyman} ensures that first-order perturbations in the non-Gaussian moment block does not contaminate the
asymptotic distribution of the test statistic for the proxy conditions.
An analogous construction yields the orthogonalized non-Gaussian moment
vector:
\begin{equation}
    \tilde{g}_{NG,T} = \bar{g}_{NG,T} - S_{NP}S_{PP}^{-1}\bar{g}_{P,T}.
    \label{eq:gnghat}
\end{equation}
The asymptotic behavior of $T^{1/2}\tilde{g}_P$ is governed by two objects
introduced in Proposition~\ref{prop:bias_bound_local}: the projected Jacobian
$\tilde{G}_P = G_P - S_{PN}S_{NN}^{-1}G_{NG}$ and the Schur complement
$\Sigma_P = S_{PP} - S_{PN}S_{NN}^{-1}S_{NP}$, which are the proxy-block
counterparts of $\tilde{G}_{NG}$ in \eqref{eq:GNG_tilde} and $\Sigma_{NG}$ in
\eqref{eq:info_decomp}. 
Since $\tilde{g}_P$ is
evaluated at the estimated value, its asymptotic variance is given by:
\begin{equation}
    \mathcal Q_P = \underbrace{\Sigma_P}_{\text{Schur complement}}
          \;-\;
          \underbrace{\tilde{G}_P\bigl(G'{S}^{-1}G\bigr)^{-1}\tilde{G}_P'}_{
          \text{parameter estimation correction}}.
    \label{eq:QPP}
\end{equation}
The second term removes the part of the uncertainty absorbed by the
estimation of $\phi$. The analogous construction yields $\mathcal Q_{NG}$.
Formally, the two specification tests are
\begin{equation}
    J_{\mathrm{prx}} = T\,\tilde{g}_P'\,\mathcal Q_P^{+}\,\tilde{g}_P,
    \qquad
    J_{NG} = T\,\tilde{g}_{NG}'\,\mathcal Q_{NG}^{+}\,\tilde{g}_{NG},
    \label{eq:subset_stats}
\end{equation}
where $(\cdot)^+$ denotes the Moore-Penrose pseudo-inverse.

\subsection{Limiting distributions}\label{subsec:limitingdist}
 The specification test of validity of the proxy exclusion conditions is conditional on the non-Gaussian moment block identifying $\phi_0$ i.e., formally:
 \begin{equation}
    H_0^{prx}: \mathbb{E}(\tilde{g}_P) = 0 | rank({G}_{NG}) = p_\phi
 \end{equation}
 under the assumptions \ref{as:indep} and \ref{as:gau}. When evaluated at the joint efficient estimator $\hat{\phi}$, the orthogonalized sample moments $T^{1/2}\tilde{g}_P$ converge in distribution to $\mathcal{N}(0, \mathcal Q_P)$. The variance matrix $\mathcal Q_P$ explicitly accounts for the parameter estimation error through the projected Jacobian $\tilde{G}_P$, has generic rank $q_P$. Therefore, with analogous reasoning for the non-Gaussianity block, both test statistics satisfy:

\begin{equation}\label{eq:subset_chi2}
J_{\mathrm{prx}} \xrightarrow{d} \chi^2(q_P), \qquad J_{\mathrm{NG}}  \xrightarrow{d} \chi^2(q_{\mathrm{NG}}).
\end{equation}
The degrees of freedom equal the number of moment conditions in each block: $q_P$ and $q_{NG}$ respectively. The penalty for estimating the $p_\phi$ parameters is already integrated into the construction of $\mathcal Q_P$ and $Q_{\mathrm{NG}}$ via the second term in equation \eqref{eq:QPP}. This allows the generalized inverse to operate on the full $q_P$-dimensional (or $q_{NG}$-dimensional) covariance structure. While the statistics are asymptotically insensitive under the joint null, mis-specification under the alternative (for instance, invalid proxies or dependent structural shocks) in one block will transmit contamination to both statistics through the shared parameter estimate $\hat{\phi}$. This property is utilized as a diagnostic tool in
Section \ref{sec:results}. We formalize the asymptotic distributions in the following proposition:

\begin{proposition}[Asymptotic distributions of the test statistics]
\label{prop:limitdist}
Under Assumptions \ref{as:validproxy}-\ref{as:gau} where $G_{NG}$ has full column rank, with $H_0\colon
\mathbb{E}[g(\phi_0)] = 0$, and consistent estimators
$\hat{\mathcal Q}_P \xrightarrow{p} \mathcal Q_P$ and
$\hat{\mathcal Q}_{NG} \xrightarrow{p} \mathcal Q_{NG}$:
\begin{enumerate}[\upshape(i)]
  \item $J_{\mathrm{prx}} \equiv T\,\tilde{g}_{P,T}'\hat{\mathcal Q}_P^{-1}\tilde{g}_{P,T}
        \xrightarrow{d} \chi^2(q_P)$.
  \item If $\mathrm{rank}(G_P) = p_\phi$, then
        $J_{NG} \equiv T\,\tilde{g}_{NG,T}'\hat{\mathcal Q}_{NG}^{-1}
        \tilde{g}_{NG,T} \xrightarrow{d} \chi^2(q_{NG})$.
  \item In general, let $\hat\Pi$ be the orthogonal projector onto
        the column space of $\hat{\mathcal Q}_{NG}$, consistently estimated from
        the spectral decomposition of $\hat{\mathcal Q}_{NG}$ by retaining
        eigenvectors with eigenvalues above a threshold $\tau_T \to
        0$ with $T^{1/2}\tau_T \to \infty$. Then,
        \begin{equation}
          J_{NG} \equiv T\,\tilde{g}_{NG,T}' \ \hat{\Pi} \ \hat{\mathcal Q}_{NG}^+ \ 
          \hat{\Pi} \ \tilde{g}_{NG,T}
          \xrightarrow{d} \chi^2(r_{NG}),
          \qquad r_{NG} \equiv q_{NG} - (p_\phi - \mathrm{rank}(G_P)).
          \label{eq:JNG}
        \end{equation}
\end{enumerate}
\end{proposition}
\noindent \textit{Proof}: See Appendix \ref{app:limitingdistributions}.

\begin{remark}[Compression of degrees of freedom] \label{rem:df_compress}
Under the null, both statistics have chi-squared limits, but their degrees of
freedom depend on how much of $\phi_0$ the other block identifies. The correction
term in \eqref{eq:QPP} absorbs the directions of $\phi$ that the conditioning
block leaves unidentified. Hence, if $G$ has full column rank $p_\phi$, then the generalized versions of the asymptotic distributions are given by:
\begin{equation}
    J_{\mathrm{prx}} \xrightarrow{d} \chi^2\bigl(q_P - (p_\phi - \mathrm{rank}(G_{NG}))\bigr),
    \qquad
    J_{NG} \xrightarrow{d} \chi^2\bigl(q_{NG} - (p_\phi - \mathrm{rank}(G_P))\bigr).
    \label{eq:df_subset}
\end{equation}
$J_{\mathrm{prx}}$ thus retains all $q_P$ degrees of freedom when the
non-Gaussian moments identify $\phi_0$ on their own. 

This characteristic can be salient in SVAR applications with partial identification through external instruments, where the number of rotation angles $p_\phi$ grows quadratically with the number of variables, while the available valid instruments are typically small. As a result, in empirical applications, the test statistic $J_{NG}$ becomes conservative, and is not always a reliable standalone test of the non-Gaussian
specification. Nonetheless, $J_{\mathrm{prx}}$ remains valid when
non-Gaussian block retains full column rank (Lemma~\ref{lem:ng_set}) and anchors
estimation regardless of proxy strength, so the test for proxy exogeneity stays valid
even when there is a single instrument in the system.

\end{remark}

\subsection{Residual-based moving block bootstrap}\label{sec:bootstrap}

The chi-squared approximations in \eqref{eq:j_joint} and
\eqref{eq:subset_chi2} rely on the asymptotic distribution of the standardized sample moments. However, the weighting matrix involves up to fourth-order cumulants. Also, as discussed in Remark \ref{rem:df_compress}, the asymptotic degrees of freedom for the test statistics also depend on the rank of the Jacobian of each moment block. This may invalidate the use of critical values from asymptotic $\chi^2$ distributions in empirical settings with finite samples.

We adopt the residual-based moving block bootstrap procedure of \citet{bruggemann_inference_2016}, which also allows for conditional heteroskedasticity. To preserve the empirical dependence structure between the residuals and proxies $(\hat{u}_t, m_t)$, we resample
overlapping blocks of the joint residual-proxy pairs \citep{jentsch_asymptotically_2022}. 

To efficiently obtain the bootstrap analog of the sample moments, 
we use the $k$-step iterative estimation procedure of \citet{andrews_2002}.
Moment recentering through \citet{hall_bootstrap_1996} subtracts $\bar{g}_T(\hat{\phi})$ from
each bootstrap moment average, so that the bootstrap data-generating
process satisfies the over-identifying restrictions exactly at
$\hat{\phi}$. With recentering, the bootstrap correctly mimics the non-standard limit
distribution due to the
degrees-of-freedom compression.
The detailed procedure can be found in Appendix \ref{appendix:mbb}.

By repeating the bootstrap procedure over $N_B$ replications, we obtain the empirical distributions of $(J_{\mathrm{joint}}^*, J_{\mathrm{prx}}^*, J_{NG}^*)$.
Hence, we can compute bootstrap $p$-values as the proportion of bootstrap replications
exceeding the corresponding sample statistic.

\section{Monte Carlo evidence}
\label{sec:results}

We assess the finite-sample properties of the hybrid GMM
estimator and the specification tests via a Monte Carlo simulation study.

\subsection{Data-generating process}
The baseline data generating process is a stable VAR(2) with $n = 3$ variables. The true
structural impact matrix is $B_0 = H_0 Q(\phi_0)$, with $H_0 = I_3$ and
rotation angles $\phi_0 = (0.80,\, {-0.40},\, 1.20)'$.

Under the null hypothesis of valid specification, the structural shocks\footnote{The choice of the underlying distributions is not consequential; the simulation results still hold with other skewed, heavy-tailed distributions including generalized Student-$t$ and other hyperbolic distributions.} are non-Gaussian, satisfying the moment conditions for identification via
co-skewness and co-kurtosis, see Section \ref{subsec:nongaussianconditions}. Each shock is normalized to zero mean and
unit variance:
\begin{itemize}
    \item {Shock 1} is drawn from a Normal-Inverse Gaussian
          distribution with parameters $(\alpha, \beta, \delta) =
          (1.6,\, 0.8,\, 1.6)$, yielding skewness
          $\kappa_3 \approx 1.01$ and excess kurtosis $\kappa_4 \approx 2.70$.
    \item {Shock 2} is drawn from a standardized $\chi^2(7)$
          distribution, with $\kappa_3 = \sqrt{8/7} \approx 1.07$ and
          $\kappa_4 = 12/7 \approx 1.71$.
    \item {Shock 3} is drawn from a standardized $\chi^2(3)$
          distribution, with $\kappa_3 = \sqrt{8/3} \approx 1.63$ and
          $\kappa_4 = 4.00$.
\end{itemize}

The instruments are $m_{it} = \psi_{i}\,\varepsilon_{it}
+ \sigma_v v_{it}$, where $v_{it} \overset{\text{i.i.d.}} {\sim}
\mathcal{N}(0,1)$. The relevance matrix is $\Psi = \operatorname{diag}(0.70,\, 0.60,\, 0.80)$ and $\sigma_v = 0.50$ for all the instruments.
Therefore, the dimension of the rotational angle parameters is $p_{\phi} = 3$, with moment conditions $q_{NG} = 9$, $q_P = 6$, giving us the total of $q = 15$ conditions. 
After estimating the rotation angles $\hat{\phi}$, see Section \ref{subsec:twostepgmm}, we construct the test statistics $J_{\mathrm{joint}}$, $J_{\mathrm{prx}}$, and $J_{\mathrm{NG}}$ as described in Section \ref{sec:tests}. To obtain bootstrap critical values, we implement a residual-based moving block bootstrap procedure, detailed in Appendix \ref{appendix:mbb}.
We report the Monte Carlo results over sample sizes $T \in \{200,\, 300,\, 500,\,
1000\}$, with $M = 500$ replications and $N_{B} = 999$ bootstrap replications.

\begin{remark}[Computational Implementation and Asymptotic Invariance]
\label{rem:comp_invariance}
Proposition \ref{prop:identification_unified} establishes global point identification
over the unconstrained rotation space $SO(n)$. In finite samples, however,
optimization over higher-order moments is sensitive to initial conditions near
Givens gimbal lock boundaries ($\phi = \pi/2$). We initialize
the numerical solver at the unweighted non-Gaussian identified preliminary estimate
$\hat{\phi}_{\mathrm{prelim}} = \arg \ \min_\phi \bar{g}_{NG}(\phi)'\bar{g}_{NG}(\phi)$
and restrict the search to the principal geometric branch
$\phi \in \Phi_0 \equiv (-\pi/2,\, \pi/2)^{p_\phi}\subset\Phi$, which eliminates
topologically isomorphic aliases without any loss of generality. The criterion is
then minimized by a constrained solver on $\Phi_0$, started both at
$\hat{\phi}_{\mathrm{prelim}}$ and at draws from the interior of the branch, so
that the reported estimate is the best of several local solutions.
\end{remark}

\begin{remark}[Spectral Regularization]
The matrices $\mathcal Q_P$ and $\mathcal Q_{NG}$ for the test statistics \eqref{eq:subset_stats} are constructed from higher-order moment conditions whose finite-sample covariance is usually ill-conditioned.
Consequently, inversion without regularization may amplify the smallest eigenvalues,
inflating the test statistics and distorting size. We regularize via spectral
truncation: eigenvalues of $\mathcal Q$ below the threshold $\tau = 10^{-4}$ are set
to zero prior to inversion, retaining only the subspace along which the matrix
is well-conditioned \citep{antoine_efficient_2007}. This truncation controls
the effective rank of $\mathcal Q^+$ without distorting the limiting distribution along
the retained eigenvalue directions.
\end{remark}

\subsection{Baseline results}
\label{sec:baseline}

\subsubsection{Estimates under $H_0$ and $H_1$}
Table \ref{tab:phi_estimates_h0} reports the average estimates and standard errors for $\phi_0$, over the simulations, under the
baseline DGP. The estimator is effectively
unbiased at all sample sizes and
standard errors decline monotonically.
Table \ref{tab:phi_estimates_h1} reports estimates under proxy
invalidity, where instrument $m_2$ violates the exogeneity condition:
$m_{2t} = 0.4\,\varepsilon_{1t} + \sigma_v v_{2t}$, i.e. it loads on shock 1 rather
than the target shock. Table~\ref{tab:phi_estimates_h1} shows
severe bias in $\hat{\phi}_2$ and $\hat{\phi}_3$. Both are consistent with the standard result that GMM estimators converge to a
pseudo-true value rather than $\phi_0$ under moment mis-specification, as shown in
\citet{white_1982, hall_inoue_2003}.

\begin{table}[htbp]
\centering
\caption{Estimates of $\phi_0$ under $H_0$: Non-Gaussian Shocks with Valid Proxies}
\label{tab:phi_estimates_h0}
\begin{tabular}{lccccc}
\toprule
 & True Value & $T = 200$ & $T = 300$ & $T = 500$ & $T = 1000$ \\
\midrule
\multicolumn{6}{l}{\textit{Estimates and their standard errors}} \\[3pt]
$\hat{\phi}_1$ & $0.80$    & 0.8046    & 0.8005    & 0.8007    & 0.8006 \\
&  & (0.0300) & (0.0250) & (0.0201) & (0.0145) \\
$\hat{\phi}_2$ & $-0.40$   & $-$0.3958 & $-$0.3955 & $-$0.3996 & $-$0.3985 \\
&  & (0.0315) & (0.0266) & (0.0212) & (0.0153) \\
$\hat{\phi}_3$ & $1.20$    & 1.1991    & 1.2017    & 1.2025    & 1.2002 \\
&  & (0.0310) & (0.0259) & (0.0207) & (0.0150) \\[6pt]
\bottomrule
\end{tabular}
\begin{minipage}{0.92\textwidth}
\vspace{4pt}
\scriptsize
\textit{Notes:} Average estimates and their standard errors are based on $M = 500$ Monte Carlo replications with $B = 999$ bootstrap replications. The DGP is a stationary VAR(2) with $n = 3$ variables and true rotation parameters $\phi_0 = (0.80, -0.40, 1.20)'$. Structural shocks are drawn from NIG, $\chi^2(7)$, and $\chi^2(3)$ distributions. Proxies are valid with relevance $\Psi = (0.70, 0.60, 0.80)'$ and noise $\sigma_v = 0.50$.
\end{minipage}
\end{table}

\begin{table}[htbp]
\centering
\caption{Estimates of $\phi_0$ under $H_1$: Invalid Proxy}
\label{tab:phi_estimates_h1}
\begin{tabular}{lccccc}
\toprule
 & True Value & $T = 200$ & $T = 300$ & $T = 500$ & $T = 1000$ \\
\midrule
\multicolumn{6}{l}{\textit{Estimates and their standard errors}} \\[3pt]
$\hat{\phi}_1$ & $0.80$    & 0.8277    & 0.8372    & 0.8228    & 0.8034 \\
&  & (0.0384) & (0.0296) & (0.0219) & (0.0152) \\
$\hat{\phi}_2$ & $-0.40$   & $-$0.4014 & $-$0.3671 & $-$0.3539 & $-$0.3427 \\
&  & (0.0475) & (0.0402) & (0.0328) & (0.0234) \\
$\hat{\phi}_3$ & $1.20$    & 0.6723    & 0.8381    & 0.9966    & 1.1371 \\
&  & (0.0494) & (0.0401) & (0.0309) & (0.0223) \\[6pt]
\bottomrule
\end{tabular}
\begin{minipage}{0.92\textwidth}
\vspace{4pt}
\footnotesize
\textit{Notes:} Average estimates and their standard errors are based on $M = 500$ Monte Carlo replications with $B = 999$ bootstrap replications. The DGP is identical to the $H_0$ specification except that proxy $m_2$ is mis-specified: $m_{2t} = 0.4 \, \varepsilon_{1t} + \sigma_v \nu_{2t}$, violating exogeneity by responding to shock 1 instead of shock 2. The substantial bias in $\hat{\phi}_3$ reflects the contamination of the proxy moment conditions.
\end{minipage}
\end{table}

\subsubsection{Size under $H_0$}
Table \ref{tab:size_h0}, Panel A reports empirical rejection rates at
the 5\% nominal level. All three tests are slightly conservative at modest sample sizes. The large mean bootstrap $J$-statistics in small samples reflect the inefficiency of the estimates when the underlying shocks are heavy-tailed; they converge toward their sample counterparts as $T$ grows. The mean $J$ statistics, both estimated and bootstrap, track the means of their limiting distributions closely: $J_{\mathrm{prx}}$ against $\chi^2(q_P)$ and $J_{\mathrm{NG}}$ against $\chi^2(q_{NG})$, as stated in Proposition~\ref{prop:limitdist}.

\begin{table}[htbp]
\centering
\caption{Size of Specification Tests under $H_0$: Non-Gaussian Shocks with Valid Proxies}
\label{tab:size_h0}
\begin{tabular}{lcccc}
\toprule
 & $T = 200$ & $T = 300$ & $T = 500$ & $T = 1000$ \\
\midrule
\multicolumn{5}{l}{\textit{Panel A: Empirical Rejection Rates (nominal size $= 0.05$)}} \\[3pt]
Joint specification test      & 0.034 & 0.048 & 0.044 & 0.056 \\
Proxy validity test           & 0.044 & 0.028 & 0.036 & 0.036 \\
Non-Gaussianity test          & 0.024 & 0.044 & 0.028 & 0.044 \\[6pt]
\multicolumn{5}{l}{\textit{Panel B: Mean $J$-statistics}} \\[3pt]
$J_{\text{joint}}$ (bootstrap) & 24.471 & 21.352 & 17.236 & 12.714 \\
$J_{\text{joint}}$ (estimate)  & 20.361 & 18.001 & 14.077 & 12.218 \\[3pt]
$J_{\text{prx}}$ (bootstrap) & 18.379 & 13.680 & 10.315 & 6.555 \\
$J_{\text{prx}}$ (estimate)  & 14.101 & 12.455 & 8.421 & 6.069 \\[3pt]
$J_{\text{NG}}$ (bootstrap)    & 32.052 & 24.031 & 17.549 & 11.941 \\
$J_{\text{NG}}$ (estimate)     & 20.170 & 17.535 & 13.823 & 10.989 \\
\bottomrule
\end{tabular}
\begin{minipage}{0.92\textwidth}
\vspace{4pt}
\footnotesize
\textit{Notes:} Panel A reports the fraction of $M = 500$ Monte Carlo replications in which the null hypothesis is rejected at the 5\% nominal level. Panel B reports average $J$-statistics across replications. ``Bootstrap'' denotes the average bootstrap $J$-statistic ($N_B = 999$); ``estimate'' denotes the average of the $J$-statistic estimates. The joint test evaluates all overidentifying restrictions; the NG test evaluates non-Gaussianity moment conditions conditional on proxy validity; the proxy validity test evaluates proxy validity conditional on non-Gaussianity.
\end{minipage}
\end{table}

\begin{table}[htbp]
\centering
\caption{Power of Specification Tests under $H_1$: Invalid Proxy}
\label{tab:power_h1}
\begin{tabular}{lcccc}
\toprule
 & $T = 200$ & $T = 300$ & $T = 500$ & $T = 1000$ \\
\midrule
\multicolumn{5}{l}{\textit{Panel A: Empirical Rejection Rates (nominal size $= 0.05$)}} \\[3pt]
Joint specification test      & 0.708 & 0.894 & 0.972 & 1.000 \\
Proxy validity test           & 0.606 & 0.810 & 0.926 & 0.994 \\
Non-Gaussianity test          & 0.072 & 0.288 & 0.642 & 0.942 \\[6pt]
\multicolumn{5}{l}{\textit{Panel B: Mean $J$-statistics}} \\[3pt]
$J_{\text{joint}}$ (bootstrap) & 28.821 & 24.334 & 19.889 & 14.223 \\
$J_{\text{joint}}$ (estimate)  & 90.325 & 116.339 & 165.058 & 292.212 \\[3pt]
$J_{\text{prx}}$ (bootstrap) & 30.403 & 22.313 & 17.051 & 9.107 \\
$J_{\text{prx}}$ (estimate)  & 94.347 & 117.791 & 218.212 & 293.421 \\[3pt]
$J_{\text{NG}}$ (bootstrap)    & 57.127 & 37.248 & 23.962 & 14.197 \\
$J_{\text{NG}}$ (estimate)     & 46.739 & 61.246 & 88.026 & 124.158 \\
\bottomrule
\end{tabular}
\begin{minipage}{0.92\textwidth}
\vspace{4pt}
\scriptsize
\textit{Notes:} Panel A reports the fraction of $M = 500$ Monte Carlo replications in which the null hypothesis is rejected at the 5\% nominal level. Panel B reports average $J$-statistics across replications. ``Bootstrap'' denotes the average bootstrap $J$-statistic ($N_B = 999$); ``estimate'' denotes the average of the $J$-statistic estimates. Under $H_1$, the proxy validity and joint test statistics diverge sharply from their bootstrap counterparts, reflecting growing power as $T$ increases. The NG bootstrap critical values are heavily inflated at small $T$ due to the contamination of the rotation estimates, suppressing NG rejection rates despite the underlying shocks remaining non-Gaussian.
\end{minipage}
\end{table}

\subsubsection{Power under $H_1$}
In Table \ref{tab:power_h1}, all three tests accumulate power rapidly as $T$ increases. At $T=500$,
rejection rates are $0.972$, $0.926$, and $0.642$ for the joint, proxy
validity, and non-Gaussianity tests respectively. 

A remark regarding the rejection rates of the non-Gaussianity test is in order.
The test for non-Gaussianity rejects despite the
true shocks are non-Gaussian. This is due to the contamination in the joint estimate, and it is expected. The joint GMM objective pools proxy and higher-order
moment conditions, so the invalid proxy moments displace $\hat{\phi}$ from
$\phi_0$. The recovered structural shocks $\hat{\varepsilon}_t =
\hat{B}_0^{-1}(\hat{\phi})\,u_t$ do not match the true shocks,
and their higher-order cross-moments violate the moment conditions even though the underlying DGP is non-Gaussian. This implies that
$g_{NG}(y_t, \hat{\phi})$ has nonzero population expectation under $H_1$,
generating power in the non-Gaussianity test through an indirect spillover
channel. This mechanism is consistent with the partitioned $J$-statistic
structure: the Schur-complement decomposition yields orthogonal test statistics
only under the joint null. Under a mis-specification of the proxy moment conditions, only parameter contamination propagates through the shared space to all components \citep{newey_generalized_1985}.

The power ordering $J_{\mathrm{prx}} \gg J_{\mathrm{NG}}$ is uniform across all sample sizes and supports the theoretical interpretation. The ordering is even more pronounced for the point-estimates of $J$-statistics, which diverge sharply from their bootstrap counterparts as $T$ increases. The proxy validity statistic reaches $218.1$ at $T=500$ against a bootstrap mean of $17.1$, whereas the non-Gaussianity statistic reaches $88.0$ against a bootstrap mean of $23.9$. When all three
tests reject and the $J_{\mathrm{prx}}$ dominates, the pattern is consistent with
proxy invalidity as the primary source of moment mis-specification. 

\subsection{The case of weak instruments}
\label{sec:weakproxy}

One of the central concerns in a proxy-SVAR is the sensitivity of inference
to the strength of the external instruments. Weak instruments can distort the sampling distribution of GMM-based estimators and invalidate standard asymptotic
approximations, see \citet{staiger_instrumental_1997, stock_gmm_2000}. While the relevance of instrument is
routinely verified in applied work through first-stage $F$-statistics, limitations of such pre-tests in terms of low power and distorted subsequent inference are documented in \citet{,ramey_chapter_2016, lewis_weak_2026}. 

We show that even under weak proxies, the hybrid framework continues to recover the parameters consistently (Proposition~\ref{prop:localtozero}), because the higher-order moment conditions anchor the identification. Though the specification test exhibit asymmetric size distortions and $J_{\mathrm{NG}}$ becomes severely conservative, $J_{\mathrm{prx}}$ retains reliable calibration and remains a valid test of proxy orthogonality.

We adopt the local-to-zero framework for proxy
relevance i.e., a Pitman drift $c/\sqrt{T}$ with
$c = 2.24$ across the sample sizes. At these magnitudes, the proxy
Jacobian $G_{P,T} = O(T^{-1/2})$ contributes negligible identifying
information. Results\footnote{Similar results with strong proxies and weak non-Gaussian shocks can be found in Appendix \ref{appendix:weak_ng_sim}. We also provide results for doubly weak-identification, i.e. weak proxies and weak non-Gaussianity. In this case, the estimates remain unbiased with increasing sample size, albeit with larger standard errors compared to strong identification case.} are reported in
Tables  \ref{tab:phi_estimates_weak} and  \ref{tab:size_weak}.

\begin{table}[htbp]
\centering
\caption{Parameter Estimates under $H_0$: Weak Proxies (Pitman Drift $c = 2.24$)}
\label{tab:phi_estimates_weak}
\begin{tabular}{lccccc}
\toprule
 & True Value & $T = 200$ & $T = 300$ & $T = 500$ & $T = 1000$  \\
\midrule
\multicolumn{6}{l}{\textit{Estimates and their standard errors}} \\[3pt]
$\hat{\phi}_1$ & $0.80$    & 0.7851    & 0.7754    & 0.8019    & 0.8007     \\
               &  & (0.0597) & (0.0526) & (0.0439) & (0.0332)  \\
$\hat{\phi}_2$ & $-0.40$   & $-$0.3621 & $-$0.3956 & $-$0.3872 & $-$0.3959  \\
               &  & (0.0639) & (0.0580) & (0.0507) & (0.0387)  \\
$\hat{\phi}_3$ & $1.20$    & 1.1789    & 1.1870    & 1.2013    & 1.2024    \\
               &  & (0.0563) & (0.0492) & (0.0405) & (0.0307)  \\[6pt]
\bottomrule
\end{tabular}
\begin{minipage}{0.95\textwidth}
\vspace{4pt}
\footnotesize
\textit{Notes:} Results based on $M = 500$ Monte Carlo replications with $N_B = 999$ bootstrap replications. Proxy relevance follows a Pitman drift: $\Psi_i = c / \sqrt{T}$ with $c = 2.24$ for all instruments, yielding $\Psi \approx (0.158, 0.129, 0.100, 0.071)$ for $T = (200, 300, 500, 1000)$. All other DGP parameters are identical to the strong-proxy $H_0$ specification.
\end{minipage}
\end{table}

\subsubsection{Consistent estimation}

Table \ref{tab:phi_estimates_weak} shows that the joint GMM estimator recovers the parameters consistently despite severe proxy weakness. Though we do note some finite-sample bias at $T=200$, it
reflects the elevated variance of higher-order sample moments under
heavy tails. The estimates converge toward $\phi_0$ as $T$ increases.
This is the simulation evidence for Proposition~\ref{prop:localtozero}. As the proxy block Jacobian $G_{P,T} \xrightarrow{T \to \infty} 0$, the identification remains anchored by the
non-Gaussian block. The co-skewness and symmetric co-kurtosis conditions
provide full-rank local curvature in a neighborhood of $\phi_0$,
ensuring the variance matrix, $\mathcal{V}_{NG}^* =\bigl(G_{NG}'\,\Sigma_{NG}^{-1}\,G_{NG}\bigr)^{-1}$ remains asymptotically non-singular. Simultaneously, the
shrinking, but a nonzero proxy relevance imposes a strictly positive
penalty on each element of the finite discrete set
$\mathcal{B} \setminus \{\phi_0\}$. This resolves the global permutation
ambiguity that the higher-order moments alone cannot eliminate. 

\subsubsection{Robust Anderson-Rubin confidence intervals}
\label{subsec:sim_bootstrap_ci}
In the presence of weak instruments, standard practice is to construct Anderson-Rubin (AR) confidence 
sets, see \citet{montiel_olea_inference_2021, 
jentsch_asymptotically_2022}, by inverting the AR statistic over the 
parameter space. They maintain 
correct asymptotic coverage uniformly over proxy strength, 
including the local-to-zero regime.

To assess the inferential contribution of the non-Gaussian 
moment conditions under weak proxy identification, we compare the AR confidence sets for the rotation parameters $\phi$ under two criteria, \textit{Proxy-only} and \textit{Proxy+NG}. The former involves only the proxy moment conditions in the GMM criterion for the AR statistic whereas the latter includes all the moment conditions (non-Gaussian and proxy). This allows us to compare the effect of additional identifying strength of the higher-order moments conditions, in terms of the size of identification-robust confidence sets, under local-to-zero proxy relevance.

 For each simulated dataset, the AR statistic is evaluated over a fine $p_{\phi}$-dimension grid of candidate rotation angle vectors $\phi \in \Phi_0$. At each grid point, the AR statistic takes the form
\begin{equation}
    AR_T(\phi_{(i)}) = T \cdot \bar{g}(\phi_{(i)})' \hat{S}(\phi_{(i)})^{-1} \bar{g}(\phi_{(i)}),
\end{equation}
where $\bar{g}(\phi_{(i)})$ is the sample moment vector at $i^{th}$ grid point and $\hat{S}_T(\phi_{(i)}) = T^{-1}\sum_{t=1}^T (g_t(\phi_{(i)}) - \bar{g}(\phi_{(i)}))(g_t(\phi_{(i)}) - \bar{g}(\phi_{(i)}))'$ is the consistent variance estimator\footnote{In the empirical application, $\widehat{S}(\phi_{(i)})$ is replaced by the residual based moving block bootstrap estimator $\widehat{S}_{MBB}(\phi_{(i)})$ 
to account for serial dependence.}. Under the null that $\phi_{(i)}$ is the true rotation, $AR_T(\phi_{(i)})$ is asymptotically distributed as $\chi^2(q)$, where $q$ is the number of moment conditions.
The AR confidence set is the collection of parameter values not 
rejected by the test at the nominal $1 - \alpha$ level  :
\begin{equation}
    \mathcal{C}_{1-\alpha} = \bigl\{\phi_{(i)} \in \Phi_0 : 
    AR_T(\phi_{(i)}) \leq \chi^2_{q,\,1-\alpha}\bigr\}.
\label{eq:AR_set}
\end{equation}
In practice, $\mathcal{C}$ is approximated by 
evaluating $AR_T(\phi_{(i)})$ over a fine grid of $1,000,000$ ($100^{p_{\phi}}$) 
candidate vectors covering $\Phi_0$, and retaining all grid 
points at which the criterion does not exceed the critical 
value. Given that such grid-based inferential methods are computationally prohibitive, we include the simulation results for a sample size $T = 1000$, for $100$ Monte Carlo replications.

\begin{table}[htbp]
\centering
\caption{Anderson-Rubin 90\% Confidence Sets with HGMM standard CIs}\label{tab:compare_sim_ARCI} 
\begin{tabular}{lcccc}
\toprule
Parameter & Mean Estimate & Proxy-only & Proxy+NG & HGMM standard CI \\
\midrule
$0.80$ & $0.8007$ & $[0.2221, 1.3328]$ & $[0.7140, 0.9044]$ & $[0.7447, 0.8567]$\\[1pt]
\quad  & & $1.1107$ & $0.1904$ & 0.1120 \\[4pt]
$-0.40$ & $-0.3959$ & $[-0.8409, 0.0793]$ & $[-0.5077, -0.2697]$ & $[-0.4603, -0.3305]$\\[1pt]
\quad  & & $0.9202$ & $0.2380$ & 0.1298 \\[4pt]
$1.20$ & $1.2024$ & $[-1.5708, 1.5708]$ & $[1.1265, 1.2852]$ & $[1.1509, 1.2539]$\\[1pt]
\quad  & & $3.1416$ & $0.1587$ & 0.1030 \\
\bottomrule
\end{tabular}

\vspace{8pt}
\begin{minipage}{0.95\textwidth}
\scriptsize
\textit{Notes.} The 90\% AR confidence sets are median sets, across the monte carlo simulations. (with their lengths below). The main comparison is between the \textit{Proxy-only} and \textit{Proxy+NG} criterions. The former involves only the proxy orthogonality moment conditions in the AR statistic whereas the latter includes all the moment conditions (non-Gaussian and proxy). The HGMM standard CI is computed as: $[\hat{\phi} - 1.645 \times \hat{\sigma}_{\phi}, \hat{\phi} + 1.645 \times \hat{\sigma}_{\phi}]$, where $\hat{\phi}$ and $\hat{\sigma}_{\phi}$ are the average of the estimates (across the monte carlo replications) and their standard errors, respectively. Results based on $M = 100$ Monte Carlo replications with $N_B = 999$ bootstrap replications. Proxy relevance follows a Pitman drift: $\Psi_i = c / \sqrt{T}$ and $c = 2.24$, with sample size, $T = 1000$. All other DGP parameters remain unchanged from the $H_0$ specification.
\end{minipage}
\end{table}
The results demonstrate a breakdown of proxy-only 
identification under weak instruments, and its recovery 
through non-Gaussianity. For $\phi_3$, the \textit{Proxy-only} 
AR set spans the entire feasible 
rotation parameter space, indicating that the proxy criterion 
fails to reject any value in the considered parameter space, and is entirely 
uninformative. As shown in Table \ref{tab:compare_sim_ARCI}, all three \textit{Proxy+NG} 
sets are centered near the true parameter values ($0.80$, 
$-0.40$, $1.20$) and are comparable in length to the HGMM 
Wald intervals, indicating that the NG conditions restore the information regarding 
the true parameter. The 
\textit{Proxy+NG} AR criterion remains informative and 
well-centered because the higher-order moments 
conditions identify the rotation parameters independently 
of proxy strength, as established in 
Proposition \ref{prop:localtozero}. The joint hybrid 
criterion eliminates the identification failure without sacrificing the validity of the identification-robust coverage.

\subsubsection{Asymmetric size distortions in specification tests}

While point identification is preserved, Table \ref{tab:size_weak}
reveals a pronounced asymmetry in the
size of the specification tests. The asymmetry
originates in the conditional structure of each test, i.e., which
block of the joint moment system anchors the parameter estimation.

Panel A of Table \ref{tab:size_weak} shows that $J_{\mathrm{NG}}$ almost never rejects at all sample sizes, against a nominal significance level of $5\%$. Panel B identifies the
mechanism where the mean sample statistic $\hat{J}_{\mathrm{NG}}$ is stable between
$12.5$ and $9.9$, while their bootstrap analogs are massively inflated
at small samples. The bootstrap mean of $J_{\mathrm{NG}}$ falls from $147.0$ at
$T=200$ to $49.6$ at $T=1000$.
\begin{table}[htbp]
\centering
\caption{Size of Specification Tests under $H_0$: Weak Proxies (Pitman Drift $c = 2.24$)}
\label{tab:size_weak}
\begin{tabular}{lcccc}
\toprule
 & $T = 200$ & $T = 300$ & $T = 500$ & $T = 1000$  \\
\midrule
\multicolumn{5}{l}{\textit{Panel A: Empirical Rejection Rates (nominal size $= 0.05$)}} \\[3pt]
Joint specification test      & 0.000 & 0.008 & 0.006 & 0.040  \\
Proxy validity test           & 0.018 & 0.032 & 0.030 & 0.040  \\
Non-Gaussianity test          & 0.000 & 0.000 & 0.000 & 0.006  \\[6pt]
\multicolumn{5}{l}{\textit{Panel B: Mean $J$-statistics}} \\[3pt]
$J_{\text{joint}}$ (bootstrap) & 25.313 & 21.309 & 17.749 & 14.020  \\
$J_{\text{joint}}$ (estimate)  & 16.524 & 16.017 & 13.967 & 12.958  \\[3pt]
$J_{\text{prx}}$ (bootstrap) & 11.527 & 9.218 & 7.757 & 6.814  \\
$J_{\text{prx}}$ (estimate)  & 9.384 & 8.522 & 7.344 & 6.794 \\[3pt]
$J_{\text{NG}}$ (bootstrap)    & 147.073 & 113.351 & 79.632 & 49.601  \\
$J_{\text{NG}}$ (estimate)     & 12.571 & 11.910 & 10.714 & 9.939  \\
\bottomrule
\end{tabular}
\begin{minipage}{0.95\textwidth}
\vspace{4pt}
\scriptsize
\textit{Notes:} Panel A reports the fraction of $M = 500$ Monte Carlo replications in which the null hypothesis is rejected at the 5\% nominal level. Panel B reports average $J$-statistics across replications. ``Bootstrap'' denotes the average bootstrap $J$-statistic ($N_B = 999$); ``estimate'' denotes the average of the $J$-statistic estimates. Under weak proxy asymptotics ($\Delta_i = c/\sqrt{T}$, $c = 2.24$), the proxy moments become asymptotically uninformative. The NG bootstrap critical values are massively inflated due to the near-singular contribution of the proxy block to the long-run covariance matrix.
\end{minipage}
\end{table}

This failure to reject is a direct consequence of {anchor failure} in the
conditional testing procedure. The $J_{\mathrm{NG}}$ statistic evaluates the higher-order
moment conditions with a maintained hypothesis that the identifying strength is provided by the proxy block\footnote{This is a specific case where there are sufficient instruments to identify the entire system. As noted in Remark \ref{rem:df_compress}, if (instruments) $k < n-1$ (structural shocks), the proxy block cannot provide the entire identification and the distribution of the $J_{\mathrm{NG}}$ statistic faces compression in its (asymptotic) degrees of freedom.}.
Under strong proxy relevance, $G_P$ is full rank and absorbs the full
parameter estimation penalty, leaving the higher-order moments as over-identifying conditions. Under weak proxy relevance, $G_{P,T} \approx 0$, the joint GMM criterion obtains identifying curvature from the non-Gaussian
block alone. The sample statistic $\hat{J}_{NG}$ is
evaluated over a subspace fitted to the higher-order moments,
compressing the realized statistic toward $q_{NG}-(p_{\phi}-rank(G_{P}))$ (Proposition \ref{prop:limitdist}).

Simultaneously, since the
proxy block contributes near-zero information to the joint long-run
covariance matrix, the restricted variance matrix $\mathcal Q_{NG}^{*}$
inherits extreme instability from these higher-order moments, generating
bootstrap values that are orders of magnitude larger than the
compressed sample statistic.

In contrast, Panel A of Table \ref{tab:size_weak} shows that $J_{\mathrm{prx}}$
maintains better size, improving with $T$. This is because, analogous to the $J_{\mathrm{NG}}$ test, the $J_{\mathrm{prx}}$ test relies on
the non-Gaussian block as its identifying anchor. The true DGP is non-Gaussian, so $G_{NG}$ maintains full column rank independently of
the proxy validity, identifying $\hat{\phi}$ without any contribution from $G_{P}$.
Hence, $J_{\mathrm{prx}}$ evaluates proxy orthogonality
against a correctly specified parameter space.

\section{Empirical Applications}
\label{sec:empirical}

We demonstrate the potential of our HGMM framework with two empirical applications. The first identifies an \emph{oil news shock} from the OPEC-window
oil futures surprise of \citet{kanzig_macroeconomic_2021} after adopting the corrections suggested by \citet{kilian_how_2024}
(Section~\ref{subsec:emp_oil}); the second identifies a Euro-area
\emph{monetary policy shock} from the high-frequency interest-rate surprise of
\citet{altavilla_measuring_2019} (Section~\ref{subsec:emp_mp}). In each case, we
report identification-robust confidence sets for the impulse responses under two
criteria, \textit{Proxy-only} (the exclusion conditions alone) and \textit{HGMM} (the
exclusion conditions combined with the higher-order moment conditions), and we subject the
proxy exclusion restriction to a formal specification test. Each illustration is accompanied by a robustness
exercise, reported in the appendix: for the oil shock, we re-run the analysis on
the \citet{kanzig_macroeconomic_2021} baseline (original instrument and specification) in Appendix~\ref{appendix:kanzig_rob},
and for the monetary policy shock on the \citet{gertler_monetary_2015} specification in Appendix~\ref{appendix:gk}.

\subsection{Identification-robust inference}
\label{subsec:emp_arsets}

The headline comparison in both illustrations is between the
identification-robust confidence sets of the two criteria. For each rotation in
the target-shock subspace, we construct the Anderson--Rubin statistic and retain the
rotation if the statistic does not exceed its per-candidate bootstrap quantile. The confidence set for a given impulse response is the projection of this
acceptance region onto the corresponding coordinate, obtained by optimizing the impulse
response function subject to the acceptance constraint. Because the target-shock
responses depend only on the first column of the impact matrix, the inversion
reduces to the $n-1$ rotation angles that build that column, and the remaining
columns are not individually relevant. This is analogous to the
procedure validated in Section~\ref{subsec:sim_bootstrap_ci}; the full
construction is given in Appendix~\ref{sec:emp_arsets_cons}. The dimension of the search is
therefore $n-1$ rather than $p_\phi=n(n-1)/2$, and it grows linearly rather than
quadratically in the size of the system. For the six-variable oil system the
inversion runs over $5$ angles with $16$ moment conditions rather than over the
$15$ angles of the point estimate, and for the four-variable monetary system over
$3$ angles with $10$ conditions. Throughout,
bootstrap quantiles use a residual-based moving-block bootstrap (MBB) that jointly
resamples the recentered residuals and the proxy \citep{bruggemann_inference_2016, jentsch_asymptotically_2022}
and re-estimates the VAR within each draw, so that first-stage uncertainty
propagates into the structural objects. Confidence sets are reported at the
nominal $90\%$ level.

\subsection{Oil news shock}
\label{subsec:emp_oil}

We revisit the identification of oil news shocks in
\citet{kanzig_macroeconomic_2021}, who use changes in oil futures prices around
pre-scheduled OPEC announcements as an external instrument for the structural
shock to oil price expectations. The setting is well suited to our framework
because the relevance and construction of the monthly instrument have been
questioned, most directly by \citet{kilian_how_2024} but also by
\citet{mori_estimating_2024, cavaliere_bootstrap_2025}, so it is an ideal case
in which the identifying information from the higher-order moment conditions can be
consequential.

\subsubsection{Specification}\label{subsubsec:oil_spec}
Our main specification adopts the corrections of \citet{kilian_how_2024}, so
that the instrument is constructed consistently with the timing of the
reduced-form data. The instrument aggregates daily futures-price surprises
around OPEC announcements into a monthly series. However, a naive within-month sum is inconsistent with a VAR built on monthly
{average} prices, because a permanent surprise of size $\delta$ on trading
day $d_k$ of a month with $T_t$ trading days shifts the current month's average
by $\frac{T_t-d_k+1}{T_t}\,\delta$ and the next month's average by
$\frac{d_k-1}{T_t}\,\delta$. The corrected proxy weights each surprise by its
within-month timing and carries the residual weight into the following month. We
adopt this temporally weighted instrument, restrict the sample to 1989:M4
onwards where the trading-day data are reliable, and estimate a log-levels VAR
with the \citet{kilian_role_2014} oil-market lag order $p=24$.

The system has the six baseline variables: real oil price, world oil
production, world oil inventories, world industrial production, U.S. industrial
production, and U.S. CPI, so that $n=6$ and $k=1$. The reduced-form VAR is
\begin{equation}
  y_t = c + \sum_{j=1}^{p} A_j\, y_{t-j} + u_t,
  \qquad u_t \sim (0,\Sigma_u),
  \label{eq:VAR_emp}
\end{equation}
estimated on 1989:M4-2017:M12; the structural innovations satisfy
$u_t = B\varepsilon_t$, where $\varepsilon_t$ meets Assumption~\ref{as:indep}.
Following Section~\ref{subsec:svarparam}, $B = H\,Q(\phi)$ with $H$ the Cholesky
factor of $\Sigma_u$ and $Q(\phi)\in SO(n)$ governed by the
$p_\phi = n(n-1)/2 = 15$ Givens angles, so $\varepsilon_t(\phi)=Q(\phi)'H^{-1}u_t$.

Indexing the oil news shock as $j=1$ and letting $m_t$ denote the instrument (robust $F = 8.9$)
the stacked moment vector combines the higher-order moment conditions of
Section~\ref{subsec:nongaussianconditions} with the proxy exclusion conditions,
giving $q_P = k(n-1)=5$ proxy conditions and $q_{NG}=75$ non-Gaussian conditions:
the co-skewness and symmetric co-kurtosis blocks of
Section~\ref{subsec:nongaussianconditions}, together
with the asymmetric co-kurtosis\footnote{\textit{Asymmetric co-kurtosis:} the
$n(n-1)$ fourth-order conditions
$\mathbb{E}[\varepsilon_{i,t}^{3}\,\varepsilon_{j,t}]=0$ for all ordered pairs
$(i,j)$ with $i\neq j$. They are not necessary for identification and are not part
of the moment block of Section~\ref{subsec:nongaussianconditions}; appending them
raises the count to $2n(n-1)+n(n-1)/2=75$. Their sample variance involves the cross moment
$\mathbb{E}[\varepsilon_{i,t}^{6}\varepsilon_{j,t}^{2}]$, which reduces to the
sixth-order marginal requirement $\mathbb{E}|\varepsilon_{i,t}|^{6+3\delta}<\infty$
under the factorization discussed in Section~\ref{subsec:nongaussianconditions};
some evidence in Appendix~\ref{appendix:evidence_nongauss} support the requirement
in this specification. Figure~\ref{fig:kl_arsets_nogrp2} shows the
identification-robust sets when the asymmetric co-kurtosis conditions are excluded.} conditions appended as a further source of information, so a total of
$q=80$ moment conditions and $p_\phi=15$ parameters. We estimate $\hat\phi$ by continuous-updating
GMM and normalize the impact matrix to a
positive diagonal; the impulse responses at horizon $i$ are
\begin{equation}
  \hat{\Theta}_i = \bigl(R\,\hat{\mathcal A}^{\,i}\,R'\bigr)\hat{B},
  \qquad i = 0,1,\dots,h,
  \label{eq:IRF_emp}
\end{equation}
with $R=[I_n,\mathbf 0_{n\times np}]$ and $\hat{\mathcal A}$ the companion matrix
of \eqref{eq:VAR_emp}. The identification-robust sets compare \textit{Proxy-only}
(the five proxy conditions that identify the target shock) with \textit{HGMM}
(addition of the target shock's higher-order moment conditions, $16$ conditions in
total), inverting the AR statistic over the five-dimensional target-shock
subspace at $10^5$ candidate rotations; the bootstrap uses block length
$\ell=36$ and $9999$ replications. Responses are normalized to a $10\%$ increase in the
real oil price on impact. Appendix~\ref{appendix:evidence_nongauss} shows some evidence of non-Gaussianity in the reduced-form innovations, and reports the
third and fourth moments of the estimated structural shocks where the joint Gaussianity of the
vector is rejected by the pooled Jarque-Bera test ($JB=35.14$ with  $12 \ d.f.$ and $p=0.0004$).

\subsubsection{Results}\label{subsubsec:oil_results}
Figure~\ref{fig:kl_mbb_vs_ar} plots the 90\% identification-robust confidence
sets along with their standard MBB confidence bands. An adverse oil \emph{supply
news} shock and a positive oil \emph{demand news} shock both raise the real price
of oil, so the two are distinguished only by the accompanying responses of oil
production, global real activity, and inventories \citep{kilian_role_2014}. Under
supply news, the anticipation of future scarcity raises storage demand on
impact, so inventories build immediately, production declines sluggishly but
persistently, and global industrial production falls. Under flow demand news, the oil price rises {together with}
global activity, production responds positively to the price incentive, and
inventories need not build.

Our estimates match the demand-news pattern, and the inventory response is
decisive. World oil inventories are drawn {down}: the response is
broadly negative over the first $12$ months, and turns significantly positive only after two years. This is the opposite of the immediate accumulation that defines a
supply news shock. World industrial production rises on impact by $0.23\%$ and is
significantly positive; U.S. industrial production is significantly positive over first ten
months and U.S. consumer prices increase for more than four years. World oil
production shows no persistent decline: it falls on impact, then rises over the following six months, consistent with
producers responding to the higher price rather than with an anticipated supply
shortfall.

\begin{figure}[H]
    \centering
     \includegraphics[width=0.98\textwidth]{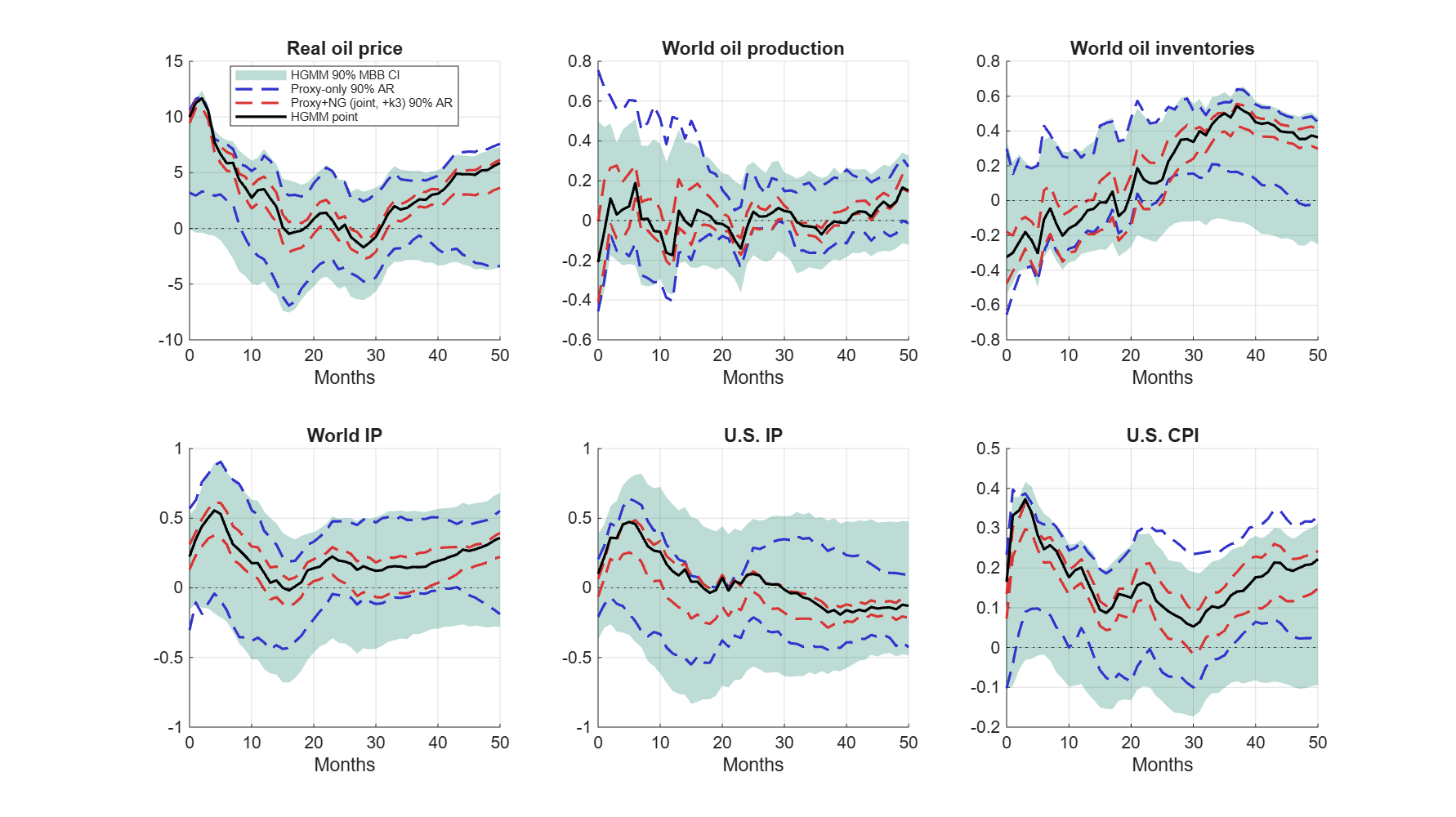}
    \caption{90\% identification-robust confidence sets with \textit{Proxy-only} and \textit{HGMM} criterions, with standard MBB percentile bands.}
    \label{fig:kl_mbb_vs_ar}
    \begin{minipage}{0.92\textwidth}
\vspace{2pt}
\scriptsize
\textit{Notes:} Each panel superimposes three $90\%$ interval estimates for the impulse response of one variable to the oil news shock. The shaded band is the residual-based moving-block-bootstrap (MBB) percentile confidence interval for the HGMM {point} estimator. The two dashed lines are the identification-robust AR sets: \textcolor{blue}{\textit{Proxy-only}} (blue, dashed) inverts the AR statistic using only the five proxy orthogonality conditions that identify the target shock; \textcolor{red}{\textit{HGMM}} (red, dashed) additionally imposes the co-skewness, asymmetric co-kurtosis, and co-kurtosis-aggregate conditions of the target shock ($16$ conditions in total). The solid black line is the HGMM point estimate. Critical values from a residual-based moving-block bootstrap ($9999$ replications, block length $\ell = 36$). Responses are normalized to a $10\%$ increase in the real oil price on impact. The $x$-axis is the horizon in months; the $y$-axis is the response in percent.
    \end{minipage}
\end{figure}

Table~\ref{tab:arlen_kilian} reports the lengths of the confidence sets under
both criteria and Table~\ref{tab:arlen_kilian_ratio} the corresponding length
ratios at horizons $0$ to $4$. The identifying contribution of the higher-order
moment conditions is large: the median reduction in set length is $71.6\%$, and
the \textit{HGMM} sets are narrower than the \textit{Proxy-only} sets at every
variable-horizon pair, with a median length ratio of $3.5$ and
a minimum of $1.6$. The gain is largest for the real oil price itself, where the
impact set contracts from a length of $7.41$ under the instrument alone to $0.75$
once the higher-order moment conditions are added, a ratio of nearly ten; world
industrial production contracts from $0.87$ to $0.18$ on impact and U.S. CPI from
$0.34$ to $0.06$.

The consequence of augmenting non-Gaussian identification is that the demand-shock characterization above exists only under the
hybrid criterion. Table~\ref{tab:arsign_kilian} in
Appendix~\ref{appendix:extra_kilian} records the horizons at which each criterion
determines the sign of a response. The instrument alone never resolves the sign
of the world oil production or U.S. industrial production response; it also never detects the inventory drawdown, since its confidence set for
inventories never lies strictly below zero at any horizon and turns positive only
after two years. A proxy-only analysis of this specification would therefore be unable
to distinguish supply news from demand news at all. Furthermore, the proxy exclusion restriction is not rejected
($\hat J_{\text{prx}}=2.87$, with asymptotic and bootstrap $p$-values of $0.71$
and $0.97$, respectively).

\begin{table}[htbp]
\centering
\caption{Length comparison of 90\% Anderson--Rubin confidence sets: \textit{Proxy-only} vs \textit{HGMM} (Kilian-corrected specification)}
\label{tab:arlen_kilian}
\renewcommand{\arraystretch}{1.15}
\resizebox{\textwidth}{!}{%
\begin{tabular}{lcccccccccc}
\toprule
 & \multicolumn{2}{c}{$h=0$} & \multicolumn{2}{c}{$h=1$} & \multicolumn{2}{c}{$h=2$} & \multicolumn{2}{c}{$h=3$} & \multicolumn{2}{c}{$h=4$} \\
\cmidrule(lr){2-3}\cmidrule(lr){4-5}\cmidrule(lr){6-7}\cmidrule(lr){8-9}\cmidrule(lr){10-11}
Variable & Proxy & HGMM & Proxy & HGMM & Proxy & HGMM & Proxy & HGMM & Proxy & HGMM \\
\midrule
Real oil price & 7.41 & \textbf{0.75} & 8.65 & \textbf{0.63} & 8.55 & \textbf{0.88} & 7.33 & \textbf{0.77} & 5.11 & \textbf{0.91} \\
World oil production & 1.21 & \textbf{0.40} & 0.97 & \textbf{0.43} & 0.70 & \textbf{0.28} & 0.71 & \textbf{0.37} & 0.71 & \textbf{0.29} \\
World oil inventories & 0.95 & \textbf{0.30} & 0.69 & \textbf{0.21} & 0.67 & \textbf{0.24} & 0.59 & \textbf{0.18} & 0.56 & \textbf{0.24} \\
World IP & 0.87 & \textbf{0.18} & 0.73 & \textbf{0.17} & 0.95 & \textbf{0.20} & 0.93 & \textbf{0.21} & 0.92 & \textbf{0.24} \\
U.S. IP & 0.41 & \textbf{0.13} & 0.42 & \textbf{0.12} & 0.53 & \textbf{0.16} & 0.55 & \textbf{0.17} & 0.71 & \textbf{0.22} \\
U.S. CPI & 0.34 & \textbf{0.06} & 0.44 & \textbf{0.07} & 0.34 & \textbf{0.06} & 0.30 & \textbf{0.07} & 0.27 & \textbf{0.05} \\
\bottomrule
\end{tabular}}
\begin{minipage}{0.96\textwidth}
\vspace{4pt}
\footnotesize
\textit{Notes:} Each entry is the length (upper minus lower bound) of the 90\% identification-robust Anderson--Rubin confidence set for the impulse response of the row variable to the target shock at horizon $h$ (months). \textit{Proxy} inverts the AR statistic using only the five proxy exclusion conditions that identify the target shock; \textit{HGMM} (Proxy+NG) additionally imposes the target shock's higher-order moment conditions, for $16$ conditions in total. The smaller length in each (variable, horizon) pair is in \textbf{bold}. Responses are normalized to a $10\%$ increase in the real oil price on impact, so entries are in percent. Sets are obtained by inverting the AR statistic over the target-shock rotation subspace with per-candidate moving-block-bootstrap critical values ($9999$ replications, block length $\ell=36$); see the construction in Appendix~\ref{sec:emp_arsets_cons}.
\end{minipage}
\end{table}

\begin{table}[htbp]
\centering
\caption{Ratio of 90\% Anderson--Rubin confidence-set lengths: \textit{Proxy-only} / \textit{HGMM} (Kilian-corrected specification)}
\label{tab:arlen_kilian_ratio}
\renewcommand{\arraystretch}{1.0}
\resizebox{0.80\textwidth}{!}{%
\begin{tabular}{lccccc|cc}
\toprule
 & \multicolumn{5}{c|}{Ratio at horizon $h$} & \multicolumn{2}{c}{Over $h=0,\dots,50$} \\
\cmidrule(lr){2-6}\cmidrule(lr){7-8}
 Variable & $h=0$ & $h=1$ & $h=2$ & $h=3$ & $h=4$ & Median & Max \\
\midrule
Real oil price & 9.90 & 13.82 & 9.76 & 9.48 & 5.61 & 9.76 & 13.82 \\
World oil production & 3.01 & 2.25 & 2.47 & 1.94 & 2.46 & 2.46 & 3.01 \\
World oil inventories & 3.21 & 3.35 & 2.81 & 3.26 & 2.34 & 3.21 & 3.35 \\
World IP & 4.90 & 4.42 & 4.84 & 4.38 & 3.90 & 4.42 & 4.90 \\
U.S. IP & 3.26 & 3.42 & 3.36 & 3.29 & 3.28 & 3.29 & 3.42 \\
U.S. CPI & 5.42 & 5.97 & 5.64 & 4.55 & 5.56 & 5.56 & 5.97 \\
\bottomrule
\end{tabular}}
\begin{minipage}{0.96\textwidth}
\vspace{4pt}
\footnotesize
\textit{Notes:} Each entry is the ratio $\text{len(Proxy)}/\text{len(HGMM)}$ of the length of the 90\% identification-robust Anderson--Rubin confidence set under \textit{Proxy-only} to the corresponding length under \textit{HGMM}, for the impulse response of the row variable to the target shock. Values above one indicate that the \textit{Proxy-only} set is longer. The last two columns summarize the ratio over all $51$ horizons $h=0,\dots,50$. Construction as in Table~\ref{tab:arlen_kilian}.
\end{minipage}
\end{table}

Other results with the specification without the asymmetric co-kurtosis moment block, the MBB
percentile bands for the proxy-SVAR estimator, 
and sensitivity to the bootstrap block length are reported in
Appendix~\ref{appendix:extra_kilian}. A robustness exercise with application to the \citet{kanzig_macroeconomic_2021} baseline instrument and
specification is reported in Appendix~\ref{appendix:kanzig_rob}.

\subsection{Euro-area monetary policy shock}\label{subsec:emp_mp}

Our second illustration identifies a Euro-area monetary policy shock using the
high-frequency interest-rate surprises of \citet{altavilla_measuring_2019} as external instruments. They
measure surprises as changes in overnight-index-swap (OIS) rates in narrow
windows around ECB policy communications and rotate the resulting surprises into
four orthogonal factors: \emph{target}, \emph{timing}, \emph{forward
guidance}, and \emph{quantitative easing}. For the detailed construction of the
surprises and the factor rotation, we refer the reader to
\citet{altavilla_measuring_2019}.

\subsubsection{Specification}\label{subsubsec:mp_spec}
We follow the specification of four variables ordered with
the policy indicator first: the two-year OIS rate (the target), the log Euro
Stoxx 50 equity index, the log EUR/USD exchange rate, and the two-year
inflation-linked swap rate, so that $n=4$ and $k=1$. The VAR has the form
\eqref{eq:VAR_emp} with $p=15$ lags, estimated at daily frequency over
2005:M3-2007:M12. We use the \emph{forward guidance} factor from the press-conference window as
the external instrument $m_t$ (robust $F = 9.9$), since every variable in this VAR is a
forward-looking asset price, and all four respond to news about the expected path
of policy. A surprise
that is orthogonal to the contemporaneous policy decision isolates this
path component. With $n=4$, the rotation has $p_\phi=n(n-1)/2=6$ Givens angles. We estimate
$\hat\phi$ by continuous-updating GMM, stacking $q_P=k(n-1)=3$ proxy exclusion
conditions with the $q_{NG}=n(n-1)+n(n-1)/2=18$ higher-order moment conditions\footnote{Unlike the oil-news specification of
Section~\ref{subsubsec:oil_spec}, we do not append the asymmetric co-kurtosis
conditions here. Matching a Student-$t$ to the estimated excess kurtosis of the
third structural shock gives $\hat\nu=4.7$, which does not clear the $6+3\delta$ moment bound required
by the asymmetric conditions under the factorization of
Section~\ref{subsec:nongaussianconditions}, whose sample variance would therefore not be
consistently estimable. Identification holds on the conditions above in either
case; see Appendix~\ref{appendix:evidence_nongauss}.} of
Section~\ref{subsec:nongaussianconditions}. Responses are reported for a
one-standard-deviation monetary policy shock. Appendix~\ref{appendix:evidence_nongauss} reports the corresponding moments of the structural shocks, where all four excess kurtoses are positive, and the pooled Jarque-Bera test clearly rejects their joint Gaussianity.

\subsubsection{Results}\label{subsubsec:mp_results}
Figure~\ref{fig:al_mbb_vs_ar} plots the 90\% identification-robust confidence sets along with their standard MBB confidence bands. A contractionary forward guidance surprise raises the two-year OIS rate by
$2.7$ basis points on impact for a one-standard-deviation shock and decays slowly: the response is still
$0.7$ basis points after $200$ business days. Equity prices fall on impact
($-0.09$ log points), and then rise after two weeks. The euro appreciates from about five months onward, the only horizons at which the exchange rate response is signed. Though the medium-run pattern of a
persistent increase in expected rates and a
stronger euro is the conventional transmission of a tightening surprise, the delayed {increase} in equity prices is the
information-effect, which the literature associates with hawkish communication that also
conveys a more favorable assessment of the future economic outlook
\citep{nakamura_high-frequency_2018, jarocinski_deconstructing_2020,
miranda-agrippino_transmission_2021}. 

Table~\ref{tab:arlen_altavilla} reports the lengths of confidence sets at both criteria: \textit{HGMM} and \textit{Proxy-only}, and
Table~\ref{tab:arlen_altavilla_ratio} the corresponding length ratios at horizons
of $0$, $50$, $100$, $150$, and $200$ business days.
The identifying contribution of the higher-order moment conditions is modest relative to oil-news shock in Section \ref{subsec:emp_oil}, but still substantial with a median reduction
in length of $26.4\%$; by
variable, the median reduction is $20\%$ for the two-year OIS rate, $30\%$ for
the equity index, $29\%$ for the exchange rate, and $39\%$ for the
inflation-linked swap, where the sets are more than halved at some horizons. The \textit{HGMM} confidence sets are narrower than the \textit{Proxy-only} sets
at every variable-horizon pair. The two sets are not nested as the per-candidate bootstrap critical value rises with the number of conditions imposed, and Appendix~\ref{appendix:gk} reports a specification in which the added conditions are not enough to offset the additional degrees of freedom. Also, Table \ref{tab:arsign_altavilla} in Appendix \ref{appendix:extra_altavilla} shows these narrower confidence sets allow us to interpret the signs of the point estimates in longer horizons. Furthermore, the proxy exclusion restriction is not rejected ($\hat{J}_{prx} = 3.18$, with asymptotic and bootstrap $p$-values of $0.36$ and $0.75$, respectively). Other results with the
MBB percentile bands (including for the proxy-estimator), AR sets with alternate bootstrap block length are reported in Appendix \ref{appendix:extra_altavilla}. A robustness
exercise on the \citet{gertler_monetary_2015} specification is included in Appendix \ref{appendix:gk}.

\begin{table}[htbp]
\centering
\caption{Length comparison of 90\% Anderson--Rubin confidence sets: \textit{Proxy-only} vs \textit{HGMM} (Altavilla monetary-policy specification)}
\label{tab:arlen_altavilla}
\renewcommand{\arraystretch}{1.15}
\resizebox{\textwidth}{!}{%
\begin{tabular}{lcccccccccc}
\toprule
 & \multicolumn{2}{c}{$h=0$} & \multicolumn{2}{c}{$h=50$} & \multicolumn{2}{c}{$h=100$} & \multicolumn{2}{c}{$h=150$} & \multicolumn{2}{c}{$h=200$} \\
\cmidrule(lr){2-3}\cmidrule(lr){4-5}\cmidrule(lr){6-7}\cmidrule(lr){8-9}\cmidrule(lr){10-11}
Variable & Proxy & HGMM & Proxy & HGMM & Proxy & HGMM & Proxy & HGMM & Proxy & HGMM \\
\midrule
2Y OIS & 0.0180 & \textbf{0.0162} & 0.0308 & \textbf{0.0264} & 0.0302 & \textbf{0.0245} & 0.0238 & \textbf{0.0183} & 0.0193 & \textbf{0.0141} \\
Euro Stoxx 50 & 1.021 & \textbf{0.799} & 0.484 & \textbf{0.278} & 0.363 & \textbf{0.230} & 0.337 & \textbf{0.259} & 0.328 & \textbf{0.258} \\
EUR/USD & 0.485 & \textbf{0.349} & 0.231 & \textbf{0.189} & 0.188 & \textbf{0.111} & 0.192 & \textbf{0.122} & 0.211 & \textbf{0.159} \\
2Y infl.-linked swap & 0.0191 & \textbf{0.0086} & 0.0158 & \textbf{0.0120} & 0.0064 & \textbf{0.0049} & 0.0056 & \textbf{0.0025} & 0.0036 & \textbf{0.0016} \\
\bottomrule
\end{tabular}
}
\begin{minipage}{0.96\textwidth}
\vspace{4pt}
\footnotesize
\textit{Notes:} Each entry is the length (upper minus lower bound) of the 90\% identification-robust Anderson--Rubin confidence set for the impulse response of the row variable to the target shock at horizon $h$ (business days). \textit{Proxy} inverts the AR statistic using only the $q_P=3$ proxy exclusion conditions that identify the target shock; \textit{HGMM} additionally imposes the higher-order moment conditions, for $7$ conditions in total. The smaller length in each (variable, horizon) pair is in \textbf{bold}. The interest-rate and swap responses are in percentage points; the equity and exchange-rate responses are in log points. Sets are obtained by inverting the AR statistic over the target-shock rotation subspace with per-candidate moving-block-bootstrap critical values ($9999$ replications, block length $\ell=50$); see the construction in Appendix~\ref{sec:emp_arsets_cons}.
\end{minipage}
\end{table}

\begin{table}[htbp]
\centering
\caption{Ratio of 90\% Anderson--Rubin confidence-set lengths: \textit{Proxy-only} / \textit{HGMM} (Altavilla monetary-policy specification)}
\label{tab:arlen_altavilla_ratio}
\renewcommand{\arraystretch}{1.0}
\resizebox{0.80\textwidth}{!}{%
\begin{tabular}{lccccc|cc}
\toprule
 & \multicolumn{5}{c|}{Ratio at horizon $h$} & \multicolumn{2}{c}{Over $h=0,\dots,200$} \\
\cmidrule(lr){2-6}\cmidrule(lr){7-8}
 Variable & $h=0$ & $h=50$ & $h=100$ & $h=150$ & $h=200$ & Median & Max \\
\midrule
2Y OIS & 1.11 & 1.17 & 1.23 & 1.30 & 1.37 & 1.25 & 1.38 \\
Euro Stoxx 50 & 1.28 & 1.74 & 1.58 & 1.30 & 1.27 & 1.43 & 1.86 \\
EUR/USD & 1.39 & 1.23 & 1.69 & 1.56 & 1.32 & 1.41 & 1.87 \\
2Y infl.-linked swap & 2.22 & 1.32 & 1.31 & 2.24 & 2.25 & 1.64 & 2.33 \\
\bottomrule
\end{tabular}
}
\begin{minipage}{0.96\textwidth}
\vspace{4pt}
\footnotesize
\textit{Notes:} Each entry is the ratio $\text{len(Proxy)}/\text{len(HGMM)}$ of the length of the 90\% identification-robust Anderson--Rubin confidence set under \textit{Proxy-only} to the corresponding length under \textit{HGMM}, for the impulse response of the row variable to the target shock. Values above one indicate that the \textit{Proxy-only} set is longer. The last two columns summarize the ratio over all $201$ horizons $h=0,\dots,200$. Construction as in Table~\ref{tab:arlen_altavilla}.
\end{minipage}
\end{table}

\begin{figure}[H]
    \centering
    \includegraphics[width=0.75\textwidth]{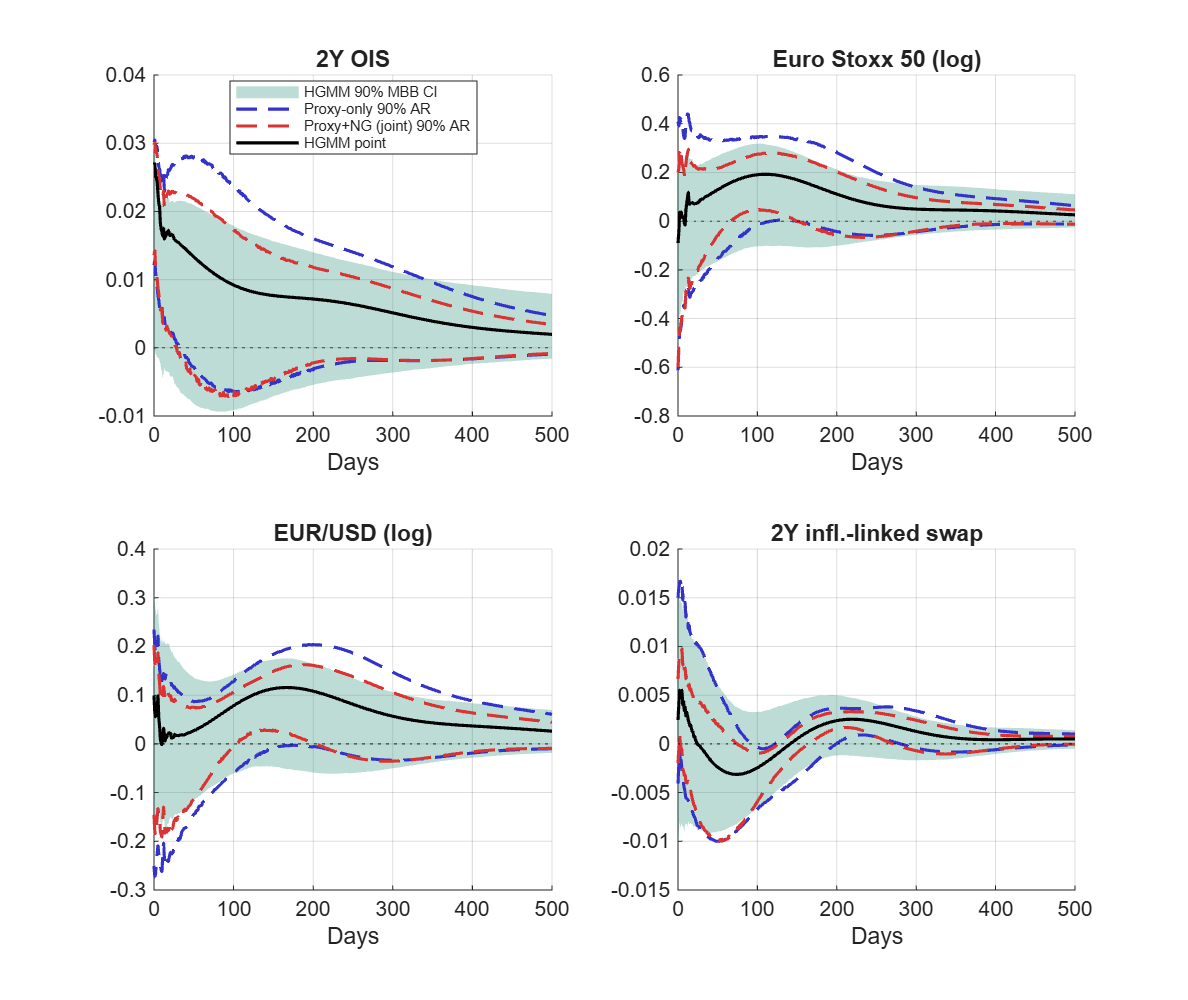}
    \caption{90\% identification-robust confidence sets with \textit{Proxy-only} and \textit{HGMM} criteria, with standard MBB percentile bands.}
    \label{fig:al_mbb_vs_ar}
    \begin{minipage}{0.92\textwidth}
\vspace{2pt}
\scriptsize
\textit{Notes:} Each panel superimposes three $90\%$ interval estimates for the impulse response of one variable to the monetary policy shock. The shaded band is the residual-based moving-block-bootstrap (MBB) percentile confidence interval for the HGMM {point} estimator. The two dashed lines are the identification-robust AR sets: \textcolor{blue}{\textit{Proxy-only}} (blue, dashed) inverts the AR statistic using only the three proxy-exclusion conditions that identify the target shock; \textcolor{red}{\textit{Proxy+NG}} (red, dashed) is the HGMM criterion which additionally imposes the three co-skewness conditions of the target shock and the fourth-order aggregate, for $7$ conditions in total. Critical values are from a residual-based moving-block bootstrap ($9999$ bootstrap replications, block length $\ell = 50$), computed per candidate rotation. The solid black line is the HGMM point estimate. Responses are to a one standard deviation contractionary monetary policy shock; the $x$-axis is the horizon in business days.
    \end{minipage}
\end{figure}

\section{Conclusion}
\label{sec:conclusion}

This paper combines identification through external instruments with identification through non-Gaussianity, stacking higher-order moment conditions and proxy exclusion restrictions into a single over-identified GMM criterion. With non-Gaussian structural shocks, the higher-order moment conditions eliminate the continuous rotational indeterminacy in $\mathcal{O}(n)$ and reduce the identified set to the finite discrete group $\mathcal{B}$; a sufficient number of valid proxies then selects a unique element from $\mathcal{B}$. The proxy-targeted columns of $B_0$ are globally point-identified and the remaining columns are identified up to signed permutation. Neither channel achieves this alone: the estimator is robust to weak instruments, and it does not require an arbitrary labeling convention to give the estimated shocks an economic interpretation.

Under local-to-zero proxy relevance the estimator remains consistent, and the identification-robust Anderson--Rubin confidence sets are substantially narrower than those based on the proxy orthogonality conditions alone. The hybrid estimator is also more efficient than the standalone non-Gaussian GMM estimator, which is less obvious: even when the proxies lose their identifying power, their long-run covariance with the non-Gaussian moment conditions continues to sharpen the effective weighting of the non-Gaussian block.

We propose two mutually orthogonal specification tests, $J_{\mathrm{prx}}$ and $J_{\mathrm{NG}}$, which assess the maintained assumptions of each identification channel separately. Under the joint null both follow $\chi^2$ limits. The effective degrees of freedom compress below the nominal count by $p_\phi - \mathrm{rank}(G_{\text{complement}})$ whenever the complementary Jacobian block is rank-deficient. A residual-based moving block bootstrap with Hall--Horowitz recentering supplies critical values that adapt both to this compression and to the fourth-order cumulants in the weighting matrix.

Some extensions can be pursued. The first is to explore the estimator when either of the identification source is weak: both channels do not enter the criterion at the same informative level, so this problem does not reduce to a familiar weak-identification analysis with two sources, and requires further research. The second is partial identification through non-Gaussianity, where only some shocks satisfy the higher-order moment conditions; this would involve estimation of the target column directly. The reduction used for the
identification-robust sets shows that the target-shock responses are a function of
the unit vector $q_1$ alone. This implies that a criterion built from the moment conditions
involving the target shock could be minimized over the sphere $\mathcal{S}^{n-1}$
without ever forming the full rotation. This would make the dimension of the
estimation problem grow linearly rather than quadratically in $n$, at the cost of
giving up the joint identification of the remaining columns. It is the natural
route to applying the framework to applications larger than those considered here.

\newpage

\bibliographystyle{apalike}
\bibliography{combined}
\pagebreak
\appendix
\counterwithin{figure}{section}
\counterwithin{table}{section}

\section{Preliminary notations and key matrices}
\label{appendix:preliminaries}

We collect preliminary notations and asymptotic conventions used
throughout the appendix. We also derive key matrices used for the proofs, and the moment conditions are stated again for the reader's convenience. We adopt the basic SVAR parameterization of Section \ref{subsec:svarparam}.

\subsection{Preliminaries}

Let $X_T$ be a sequence of random vectors and $X$ a random vector on a common probability space with measure $\mathbb{P}$. We write $X_T \xrightarrow{p} X$ for convergence in probability and $X_T \xrightarrow{d} X$ for convergence in distribution. For a sequence of random variables $Y_T$, $Y_T = o_p(1)$ denotes convergence to zero in probability and $Y_T = O_p(1)$ denotes stochastic boundedness: for every $\epsilon > 0$ there are a finite $M_{\epsilon} > 0$ and an integer $T_0$ with $\mathbb{P}(\|Y_T\| > M_{\epsilon}) < \epsilon$ for all $T > T_0$. For a deterministic sequence $a_T$, $O(a_T)$ denotes a deterministic sequence bounded by a constant multiple of $a_T$. We write $Z \sim \mathcal{N}(\mu, \Sigma)$ for the multivariate Gaussian with mean $\mu$ and covariance $\Sigma$, and $\chi^2(r)$ for the chi-squared distribution with $r$ degrees of freedom.

Let $v$ be a vector and $A$ a real matrix. We write $\|v\| \coloneq \sqrt{v^\top v}$ for the Euclidean norm and $\|A\| \coloneq \sup_{\|x\|=1} \|Ax\|$ for the induced spectral norm, $\mathrm{vec}(A)$ for the column-wise vectorization of $A$, and $\otimes$ for the Kronecker product. The column space and null space of $A$ are $\mathrm{col}(A)$ and $\mathrm{null}(A)$, and $A^{+}$ is the Moore--Penrose pseudo-inverse. For symmetric $A$, $\lambda_{\min}(A)$ and $\lambda_{\max}(A)$ are the smallest and largest eigenvalues.

Definiteness is stated through the quadratic form $x^\top A x$ for symmetric $A$: $A \succ 0$ means $x^\top A x > 0$ for all $x \neq 0$, and $A \succeq 0$ means $x^\top A x \geq 0$ for all $x$. For a stationary process $X_t$, the long-run variance
$\sum_{h=-\infty}^{\infty}\operatorname{Cov}(X_t,X_{t-h})$ is the asymptotic variance of
$T^{-1/2}\sum_t X_t$; it reduces to $\operatorname{Var}(X_t)$ when $X_t$ is serially
uncorrelated.

\subsection{Moment conditions}\label{appendix:momentcondtions}
Let $g_t(\phi) = (g_{NG,t}(\phi)', g_{P,t}(\phi)' )'$ collect the
 {non-Gaussian} ($NG$) and {proxy} ($P$) moment contributions,
with dimensions $q_{NG}$ and $q_P = k(n-1)$, respectively, and total
$q = q_{NG} + q_P$. Define:
\begin{align*}
  \bar{g}_T(\phi)
    &= T^{-1}\sum_{t=1}^T g_t(\phi)
     = \bigl(\bar{g}_{NG,T}(\phi)' ,\; \bar{g}_{P,T}(\phi)' \bigr)',
     \qquad \bar{g}_T : \Phi \to \mathbb{R}^q, \\
  g(\phi) &= \mathbb{E}\bigl[\bar{g}_T(\phi)\bigr], \\
  G(\phi) &= \frac{\partial g(\phi)}{\partial\phi'}
            \in \mathbb{R}^{q\times p_\phi}, \qquad
  G \equiv G(\phi_0) = \bigl(G_{NG}', G_P' \bigr)',
\end{align*}
where $G_P = \partial g_P(\phi_0)/\partial\phi'$ and
$G_{NG} = \partial g_{NG}(\phi_0)/\partial\phi'$. The consistent
sample Jacobian is $\hat{G}_T(\phi)$.

The {long-run variance} of the moments is $S =
\mathrm{Avar}(T^{1/2}\bar{g}_T(\phi_0)) > 0$, estimated consistently
by $\hat{S}_T$. Conformably with the
$q_{NG}/q_{P}$ partition, write
\[
  S = \begin{pmatrix} S_{NN} & S_{NP} \\ S_{PN} & S_{PP} \end{pmatrix},
  \qquad S_{PN} \equiv S_{P,NG},\quad S_{NP} = S_{PN}',
\]
where $S_{PP} \in \mathbb{R}^{q_P\times q_P}$, $S_{NN} \in
\mathbb{R}^{q_{NG}\times q_{NG}}$, $S_{PN} \in \mathbb{R}^{q_P\times
q_{NG}}$. The efficient weighting matrix is $W = S^{-1}$, with
$\hat{W}_T = \hat{S}_T^{-1}$.

\subsection{Key derived matrices}\label{appendix:keymatrices}

\begin{itemize}
  \item \emph{Schur complements} of $S$:
    \begin{equation}\label{eq:app_schurcomplements}
      \Sigma_P \equiv S_{PP} - S_{PN}S_{NN}^{-1}S_{NP} > 0, \qquad
      \Sigma_{NG} \equiv S_{NN} - S_{NP}S_{PP}^{-1}S_{PN} > 0.
    \end{equation}
    Positive definiteness follows from $S > 0$.
  \item \emph{Annihilator matrix}:
    \begin{equation}\label{eq:app_annihilator}
      \Omega \equiv S - G(G'S^{-1}G)^{-1}G' \succeq 0,
      \quad \mathrm{rank}(\Omega) = q - p_\phi,
    \end{equation}
    and the sample-level annihilator operator
    $M \equiv I_q - G(G'S^{-1}G)^{-1}G'S^{-1}$, satisfying
    $T^{1/2}\bar{g}_T(\hat{\phi}_T) = M\,T^{1/2}\bar{g}_T(\phi_0) +
    o_p(1)$.
  \item \emph{Orthogonalizing (Frisch-Waugh) operators}:
    \begin{equation}\label{eq:app_fwloperators}
      L \equiv [I_{q_P},\; -S_{PN}S_{NN}^{-1}]
        \in \mathbb{R}^{q_P\times q}, \qquad
      K \equiv [-S_{NP}S_{PP}^{-1},\; I_{q_{NG}}]
        \in \mathbb{R}^{q_{NG}\times q}.
    \end{equation}
  \item \emph{Partialled (projected) Jacobians}:
   
  \begin{equation}\label{eq:app_partialjacobians}
    \tilde{G}_P \equiv LG \ \text{;} \ \tilde{G}_{NG} \equiv KG
  \end{equation}
  
    \item \emph{Orthogonalized sample moments}:
    \begin{equation}\label{app_orthosamplemoments}
      \tilde{g}_{P,T} \equiv L\bar{g}_T(\hat{\phi}_T), \qquad
      \tilde{g}_{NG,T} \equiv K\bar{g}_T(\hat{\phi}_T).
    \end{equation}

  \item \emph{Asymptotic variance matrices of orthogonalized moments}:
  \begin{equation}\label{app_asympvarorthomoments}
      \mathcal Q_P \equiv L\Omega L', \qquad \mathcal Q_{NG} \equiv K\Omega K'.
  \end{equation}
    These are the limiting covariance matrices of $T^{1/2}
    \tilde{g}_{P,T}(\hat{\phi}_T)$ and $T^{1/2}\tilde{g}_{NG,T}(\hat{\phi}_T)$,
    respectively (see Appendix \ref{appendix:proofs}).
\end{itemize}
\newpage


\section{Proofs} \label{appendix:proofs}
This section contains the proofs for all the lemmas, propositions and corollaries stated in the paper.

\subsection{Proof of Proposition \ref{prop:identification_unified}}\label{appendix:proof_identification_unified}
For the main proof, we use the following Lemma \ref{lem:ng_set}:
\begin{lemma}[Non-Gaussian identification]\label{lem:ng_set}
Let $B_0=HQ_0$ with $Q_0\in SO(n)$, and let $Q(\phi)\in SO(n)$ be the Givens
parametrization with $p_\phi=n(n-1)/2$ angles. Let $\phi_0$ satisfy
$Q(\phi_0)=Q_0$ and the
differential of $\phi\mapsto Q(\phi)$ is injective at $\phi_0$. Let the
non-Gaussian moment block stack the co-skewness and co-kurtosis conditions
$f^{(1)}_{ij}$ and $f^{(2)}_{ij}$ in \eqref{eq:C1}--\eqref{eq:C2} for all
ordered pairs $i\neq j$. Let
$\mathcal{S}_n=\{P\Lambda:P\in\mathcal{P}_n,\ \Lambda\in\mathcal{D}_n^{\pm}\}$
denote the $n!\,2^n$ signed permutation matrices, and let
$\mathcal{B}=\{B_0S:S\in\mathcal{S}_n\}$ be the signed-permutation orbit of
$B_0$. Under Assumptions~\ref{as:indep} and \ref{as:gau}:
\begin{enumerate}
  \item[(i)] the population higher-order moment conditions hold at $\phi$ if and only
  if $B(\phi)\in\mathcal{B}$. Hence, the non-Gaussian identified set within the
  $SO(n)$ parametrization is
  $\{B_0S:S\in\mathcal{S}_n,\ \det S=1\}$, which has $n!\,2^{n-1}$ elements;
  without the determinant restriction it is the full orbit $\mathcal{B}$, with
  $n!\,2^n$ elements.
\item[(ii)] $G_{NG}(\phi_0)$ has full column rank $p_\phi$;

\end{enumerate}
\end{lemma}

\noindent The restriction $\det Q_0=1$ is without loss of generality: if
$\det Q_0=-1$, we can replace $B_0$ by the observationally equivalent $B_0\Lambda_1$,
where $\Lambda_1$ flips the sign of the first shock.

\begin{proof}
Since $u_t=B_0\varepsilon_t=HQ_0\varepsilon_t$, the candidate shocks are
$\varepsilon_t(\phi)=R(\phi)\varepsilon_t$, where
$R(\phi)\equiv Q(\phi)'Q_0\in SO(n)$ and $R(\phi_0)=I$. Write $r_i'$ for the
$i$-th row of $R$, so that $\varepsilon_{i,t}(\phi)=r_i'\varepsilon_t$, and for
$i\neq j$ set $X=\varepsilon_{i,t}(\phi)$ and $Y=\varepsilon_{j,t}(\phi)$. By Assumption~\ref{as:indep}, the components of $\varepsilon_t$ have mean zero, unit variance and are mutually uncorrelated, with finite fourth-order moments.
Hence, $\mathbb{E}[X^2]=\|r_i\|^2=1$, $\mathbb{E}[Y^2]=\|r_j\|^2=1$ and
$\mathbb{E}[XY]=r_i'r_j=0$ for every $\phi$. For mean-zero
variables, the moment--cumulant identities are:
\[
\mathbb{E}[X^2Y]=\operatorname{cum}(X,X,Y),\qquad
\mathbb{E}[X^2Y^2]=\operatorname{cum}(X,X,Y,Y)+\mathbb{E}[X^2]\,\mathbb{E}[Y^2]
  +2\,\mathbb{E}[XY]^2 .
\]
By Assumption~\ref{as:indep}, the third- and fourth-order cumulant tensors of
$\varepsilon_t$ are diagonal, with diagonal entries
$\kappa_{3,k}=\mathbb{E}[\varepsilon_{k,t}^3]$ (the skewness, by unit variance)
and $\kappa_{4,k}=\mathbb{E}[\varepsilon_{k,t}^4]-3$ (the excess kurtosis).
Since $X$ and $Y$ are linear combinations of the components of
$\varepsilon_t$, multilinearity of cumulants gives:
\[
\operatorname{cum}(X,X,Y)=\sum_{k}\kappa_{3,k}\,r_{ik}^2 r_{jk},\qquad
\operatorname{cum}(X,X,Y,Y)=\sum_{k}\kappa_{4,k}\,r_{ik}^2 r_{jk}^2 ,
\]
because every cross term involves an off-diagonal cumulant, which is zero.
Substituting, and using
$\mathbb{E}[X^2]\mathbb{E}[Y^2]+2\mathbb{E}[XY]^2=1$, we obtain for $i\neq j$
\begin{align}
f^{(1)}_{ij}(R) &\equiv \mathbb{E}[\varepsilon_{i,t}^2(\phi)\varepsilon_{j,t}(\phi)]
   = \sum_{k}\kappa_{3,k}\,r_{ik}^2 r_{jk}, \label{eq:C1}\\
f^{(2)}_{ij}(R) &\equiv \mathbb{E}[\varepsilon_{i,t}^2(\phi)\varepsilon_{j,t}^2(\phi)]-1
   = \sum_{k}\kappa_{4,k}\,r_{ik}^2 r_{jk}^2. \label{eq:C2}
\end{align}

\medskip
\emph{Sufficiency.}  Let $R\in\mathcal{S}_n$, so that $r_i'=\pm e_{\pi(i)}'$ for
a permutation $\pi$ of $\{1,\ldots,n\}$. For $i\neq j$ we have
$\pi(i)\neq\pi(j)$, so for every $k$ at most one of $r_{ik}$ and $r_{jk}$ is
nonzero and $r_{ik}r_{jk}=0$. Hence,
$r_{ik}^2r_{jk}=r_{ik}(r_{ik}r_{jk})=0$ and
$r_{ik}^2r_{jk}^2=(r_{ik}r_{jk})^2=0$ for every $k$, so every term in
\eqref{eq:C1} and \eqref{eq:C2} vanishes, for any values of
$\kappa_{3,k}$ and $\kappa_{4,k}$.

\emph{Necessity.} Let $R\in SO(n)$ satisfy $f^{(2)}_{ij}(R)=0$ for all
$i\neq j$, and write $s_{ik}=r_{ik}^2\ge 0$. By Assumption~\ref{as:gau}(b), at
most one shock is mesokurtic and the remaining $n-1$ excess kurtoses share a
common sign. Without loss of generality this sign is positive (otherwise
multiply \eqref{eq:C2} by $-1$), so $\kappa_{4,k}\ge 0$ for all $k$. For any
pair $i\neq j$,
\[
0=f^{(2)}_{ij}(R)=\sum_k\kappa_{4,k}\,s_{ik}s_{jk}
\;\Longrightarrow\;
s_{ik}s_{jk}=0\ \text{ for every non-mesokurtic } k,
\]
since a sum of non-negative terms is zero only if every term is zero. Hence, each
of the $n-1$ non-mesokurtic columns of $R$ has at most one nonzero entry. The
columns of an orthogonal matrix have unit norm, so $\sum_i s_{ik}=1$ and that
entry equals $\pm 1$. These $n-1$ columns are mutually orthogonal signed unit
vectors and are therefore supported on distinct rows. The remaining column, of
unit norm and orthogonal to them, is the signed unit vector on the remaining
row. Thus, $R\in\mathcal{S}_n$. This proves (i).

\emph{Local coordinates.} Since $R(\phi_0)=I$, we can
write $R(\phi)=\exp(\Phi(\phi))$, where $\Phi(\phi)$ is skew-symmetric, smooth
in $\phi$, and $\Phi(\phi_0)=0$. Let
$\vartheta(\phi)=(\Phi_{kl}(\phi))_{k<l}\in\mathbb{R}^{p_\phi}$ and
$J=\partial\vartheta/\partial\phi'\,|_{\phi_0}$. This implies $\partial R/\partial\phi_m|_{\phi_0}=\partial\Phi/\partial\phi_m|_{\phi_0}$.
Moreover, $\partial R/\partial\phi_m=(\partial Q/\partial\phi_m)'Q_0$ with $Q_0$
invertible. Hence, $J$ is nonsingular if and only if the differential of
$\phi\mapsto Q(\phi)$ is injective at $\phi_0$. Writing the moment vector as $g(\vartheta(\phi))$, the chain
rule gives $G_{NG}(\phi_0)=G_\vartheta J$ with
$G_\vartheta=\partial g/\partial\vartheta'\,|_{\vartheta=0}$. It therefore
suffices to show that $G_\vartheta$ has full column rank.

\emph{Expansion.} For $k<l$, set $\vartheta_{kl}=\Phi_{kl}$, so that
$\Phi_{lk}=-\vartheta_{kl}$, and substitute
$r_{ik}=\delta_{ik}+\Phi_{ik}+O(\|\Phi\|^2)$. For the co-kurtosis conditions,
\[
f^{(2)}_{ij}=\sum_k\kappa_{4,k}(\delta_{ik}+\Phi_{ik})^2(\delta_{jk}+\Phi_{jk})^2
+O(\|\Phi\|^2)=O(\|\Phi\|^2),\qquad i\neq j,
\]
because every first-order term carries the factor $\delta_{ik}\delta_{jk}=0$.
Hence, the co-kurtosis rows of $G_\vartheta$ are zero. For the co-skewness
conditions,
\[
f^{(1)}_{ij}=\sum_k\kappa_{3,k}(\delta_{ik}+\Phi_{ik})^2(\delta_{jk}+\Phi_{jk})
+O(\|\Phi\|^2)=\kappa_{3,i}\Phi_{ji}+O(\|\Phi\|^2).
\]
Thus, each coordinate $\vartheta_{kl}$, $k<l$, enters in the following two conditions,
\[
f^{(1)}_{kl}=-\kappa_{3,k}\,\vartheta_{kl}+O(\|\Phi\|^2),
\qquad
f^{(1)}_{lk}=\kappa_{3,l}\,\vartheta_{kl}+O(\|\Phi\|^2),
\]
so the column of $G_\vartheta$ for $\vartheta_{kl}$ equals
$(-\kappa_{3,k},\,\kappa_{3,l})'$ on rows $f^{(1)}_{kl},f^{(1)}_{lk}$ and zero
elsewhere. Each ordered pair $(i,j)$ belongs to exactly one unordered pair
$\{k,l\}$, so distinct columns have disjoint supports. The columns are
therefore linearly independent if and only if each is nonzero, i.e.\
$(\kappa_{3,k},\kappa_{3,l})\neq(0,0)$ for every $k<l$. By
Assumption~\ref{as:gau}(a), at most one skewness is zero, so this condition
holds. Hence $G_\vartheta$ has full column rank $p_\phi$, and so does
$G_{NG}(\phi_0)=G_\vartheta J$. This proves (ii).
\end{proof}

\begin{remark}[Local versus global identification]\label{rem:local_global_ng}
The two parts of Lemma~\ref{lem:ng_set} rest on different moments. The global
characterization in (i) uses only the co-kurtosis conditions; the common-sign
restriction is that of \citet{lanne_identifying_2023}, who assume leptokurtic
shocks, and the signed-permutation indeterminacy is the classical result of
\citet{comon_independent_1994}. Local identification in (ii) comes entirely from
the co-skewness conditions, because the co-kurtosis block is flat at $\phi_0$.
Consequently, when two or more shocks are close to symmetric,
$G_{NG}(\phi_0)$ is close to rank deficient even though (ii) continues to hold:
$\phi_0$ remains globally identified up to the orbit, but first-order local
identification is weak. Applying the same argument at any other element of the
identified set, with the skewness coefficients permuted and sign-changed, shows
that (i) holds there as well, provided the parametrization is regular at that
point.
\end{remark}

Now we provide proof of the main proposition:
\begin{proof}
By Lemma~\ref{lem:ng_set}, any
$\tilde{B}_0 \in \mathcal{B}$ takes the form 
$\tilde{B}_0 = B_0 P\Lambda$ with corresponding alternative 
shocks
\begin{equation}\label{eq:alt_eps_u}
    \tilde{\varepsilon}_t 
    = \tilde{B}_0^{-1}u_t 
    = \Lambda P'\varepsilon_t,
\end{equation}
where $\Lambda^{-1} = \Lambda$ and $P^{-1} = P'$. Since 
$\mathcal{B}$ is the non-Gaussian identified set, it suffices to
show that the proxy exclusion conditions of Assumption \ref{as:validproxy} 
and the sign normalization of Assumption \ref{ass:norm_u} 
jointly select a unique element from $\mathcal{B}$ for the 
target columns, and characterize the indeterminacy 
for the remaining columns.

For $\tilde{B}_0$ to satisfy Assumption \ref{as:validproxy}, the 
cross-moment $\mathbb{E}[m_t\tilde{\varepsilon}_t']$ must 
preserve the required block-zero structure. Substituting 
\eqref{eq:alt_eps_u} and using 
$\mathbb{E}[m_t\varepsilon_t'] = [\Psi\;\;0]$ with $\Psi$,  
a $k \times k$ nonsingular diagonal matrix,
\begin{equation}\label{eq:proxy_blocks_u}
    \mathbb{E}[m_t\tilde{\varepsilon}_t']
    = [\,\Psi\;\;0\,]\,P\Lambda
    = \bigl[\;\Psi P_{11}\Lambda_1 \;\;\;
    \Psi P_{12}\Lambda_2\;\bigr],
\end{equation}
where $P$ and $\Lambda$ are partitioned conformably with block 
sizes $(k, n-k)$. The $k\times(n-k)$ zero-block requirement 
in Assumption \ref{as:validproxy} forces 
$\Psi P_{12}\Lambda_2 = 0$. Since $\Psi$ and $\Lambda_2$ 
are both nonsingular, this gives $P_{12} = 0_{k\times(n-k)}$, 
and the permutation structure of $P$ then implies 
$P_{21} = 0_{(n-k)\times k}$. Hence, $P$ is block-diagonal:
\begin{equation}\label{eq:P_blockdiag}
    P = \begin{pmatrix} P_{11} & 0 \\ 0 & P_{22} \end{pmatrix},
    \qquad P_{11} \in \mathcal{P}_k, \quad 
    P_{22} \in \mathcal{P}_{n-k}.
\end{equation}
The relevance block of \eqref{eq:proxy_blocks_u} reduces to 
$\Psi P_{11}\Lambda_1 = \tilde{\Psi}$, where $\tilde{\Psi}$ 
is $k\times k$ diagonal matrix. The product $\Psi P_{11}$ is a 
scaled permutation matrix with non-zero entry $\psi_i$ at 
position $(i, \sigma(i))$, where $\sigma$ is the permutation 
induced by $P_{11}$. For $\Psi P_{11}\Lambda_1$ to be 
diagonal, each non-zero entry must lie on the main diagonal, 
requiring $\sigma(i) = i$ for all $i = 1,\ldots,k$. Hence, 
$P_{11} = I_k$. Substituting, $\Psi\Lambda_1 = \tilde{\Psi}$ 
implies $\psi_i\lambda_i = \tilde{\psi}_i > 0$; since 
$B_{ii} > 0$ (sign-normalized), we obtain $\lambda_i = +1$ for all $i$, so 
$\Lambda_1 = I_k$. The allowable transformation therefore 
reduces to
\begin{equation}\label{eq:residual_transform}
    P\Lambda = \begin{bmatrix} I_k & 0 \\ 0 & P_{22}\Lambda_2 
    \end{bmatrix},
    \qquad P_{22} \in \mathcal{P}_{n-k}, \quad 
    \Lambda_2 \in \mathcal{D}_{n-k}^{\pm},
\end{equation}
and any admissible $\tilde{B}_0 \in \mathcal{B}$ takes the form
\begin{equation}\label{eq:admissible}
    \tilde{B}_0 = B_0 P\Lambda 
    = \bigl[\,B_{0,1} \;\;\; B_{0,2}P_{22}\Lambda_2\,\bigr].
\end{equation}
Since every admissible $\tilde{B}_0$ shares the first $k$ 
columns $B_{0,1}$ regardless of the choice of $P_{22}$ and 
$\Lambda_2$, the target sub-matrix $B_{0,1}$ is globally 
point-identified for all $1 \leq k \leq n$. This establishes 
result \ref{res:target}.

The remaining identification result \ref{res:residual} depends 
on the cardinality of the set of admissible 
$(P_{22}, \Lambda_2)$ pairs, which is governed by $k$ relative 
to $n$.

\smallskip
\noindent\textit{Case $k = n$.} Both $P_{22}$ and $\Lambda_2$ 
are vacuous (the block is empty), so \eqref{eq:residual_transform} 
gives $P\Lambda = I_n$ directly. Hence, $\tilde{B}_0 = B_0$ and 
$B_0$ is globally point-identified.

\smallskip
\noindent\textit{Case $k = n-1$.} $P_{22} \in \mathcal{P}_1$ 
forces $P_{22} = (1)$, so $P = I_n$. The remaining sign 
$\lambda_n \in \{-1, +1\}$ enters the $(n,n)$ entry of 
$\tilde{B}_0$ as $\lambda_n [B_0]_{nn}$. Assumption \ref{ass:norm_u} 
requires this entry to be positive, which holds if and only if 
$\lambda_n = +1$. Hence, $\Lambda_2 = (1)$, $P\Lambda = I_n$, 
and $B_0$ is globally point-identified.

\smallskip
\noindent\textit{Case $1 \leq k < n-1$.} The residual 
permutation $P_{22}$ ranges over $\mathcal{P}_{n-k}$, which 
has cardinality $(n-k)!$, and $\Lambda_2$ ranges over 
$\mathcal{D}_{n-k}^{\pm}$, which has cardinality $2^{n-k}$. 
The proxy and sign normalization conditions place no further 
restrictions on $(P_{22}, \Lambda_2)$, so the equivalence 
class $\mathcal{E}_{n-k}$ has cardinality at most 
$(n-k)!\cdot 2^{n-k}$. The residual sub-matrix $B_{0,2}$ is 
not point-identified; it is identified only up to 
$\mathcal{E}_{n-k}$. This establishes result \ref{res:residual} 
in all cases.
\end{proof}

\subsection{Reversion to set identification under proxy irrelevance}\label{appendix:irrelevance}


\begin{corollary}\label{cor:irrelevance}
Suppose the conditions of Proposition \ref{prop:identification_unified} hold,
but Assumption \ref{as:validproxy} is replaced by strict irrelevance:
$\mathbb{E}[m_t\varepsilon_t'] = 0_{k\times n}$. Then
$\mathbb{E}[g_{P,t}(\phi)] = 0$ for every $\phi \in \Phi$, the proxy
moment block imposes no identifying restrictions, and the identified set
under the joint moment system coincides exactly with $\mathcal{B}$.
The parameter $\phi_0$ is globally set-identified with identified set of
cardinality $|\mathcal{B}| \leq n!\cdot 2^n$.
\end{corollary}

\begin{proof}
Under strict irrelevance, $\mathbb{E}[m_tu_t'] =
\mathbb{E}[m_t\varepsilon_t']B_0' = 0$. For any $\phi \in \Phi$, using
the notation of Section \ref{appendix:keymatrices}:
\[
  \mathbb{E}[m_t\varepsilon_t(\phi)']
  = \mathbb{E}[m_tu_t'](H^{-1})'Q(\phi)
  = 0_{k\times n}.
\]
Accordingly, $\mathbb{E}[g_{P,t}(\phi)] = 0$ for every $\phi$. The
proxy conditions are trivially satisfied throughout $\Phi$ and
contribute no restrictions beyond those already imposed by
Assumptions \ref{as:indep} and \ref{as:gau}. The identified set therefore equals $\mathcal{B}$,
from which no element can be excluded.
\end{proof}

\subsection{Proof of Proposition \ref{prop:efficiency}}\label{appendix:efficiency}

\begin{proof}
Substituting $G_{NG} = \tilde{G}_{NG} + S_{NP}S_{PP}^{-1}G_P$ and
applying the block-matrix inversion identity $[S^{-1}]_{NN} =
\Sigma_{NG}^{-1}$ (see Section \ref{appendix:preliminaries}) yields
\[
  I_{joint} \;=\; G'S^{-1}G
  \;=\;
  \underbrace{G_P'S_{PP}^{-1}G_P}_{I_P}
  \;+\;
  \underbrace{\tilde{G}_{NG}'\Sigma_{NG}^{-1}\tilde{G}_{NG}}_{\Delta I_{NG}}.
\]

$\Sigma_{NG}=S_{NN}-S_{NP}S_{PP}^{-1}S_{PN}$ is the Schur complement of $S_{PP}$ in
$S$. Since $S\succ0$ by the regularity conditions in Section \ref{ass:regularity}, $S_{PP}\succ0$ and
$\Sigma_{NG}\succ0$; hence $\Delta I_{NG}=\tilde G_{NG}'\Sigma_{NG}^{-1}\tilde G_{NG}
\succeq0$, and $I_{joint}=I_P+\Delta I_{NG}\succeq I_P\succ0$. 
Matrix inversion is order-reversing on the positive definite cone: $A\succeq B \succ 0$ implies $B^{-1} \succeq A^{-1}$.  Therefore,  $\mathcal{V}_P =
I_P^{-1} \succeq I_{joint}^{-1} =  \mathcal{V}_{joint}$.
\end{proof}

\subsection{Proof of Proposition \ref{prop:localtozero}}
\label{appendix:proof_localtozero}
\begin{proof}
\noindent\textbf{Part \ref{localtozero_infolimit}} Under the Pitman sequence only the
first-moment relevance $\mathbb{E}[m_{t,T}\varepsilon_{\text{target},t}]$ drifts to zero so the population
proxy Jacobian
$G_{P,T}=\partial\,\mathbb{E}[g_{P}(\phi_0)]/\partial\phi'$ satisfies
$G_{P,T}=O(T^{-1/2})$. But the proxy's own second moment $\mathbb{E}[m_{t,T}^2]$ is bounded away from zero, so the
proxy-block covariance $S_{PP,T}$ stays positive definite with $S_{PP,T}^{-1}$ bounded
(and $S_{NN}\succ0$), and hence
$\mathcal{I}_{P,T}=G_{P,T}'S_{PP}^{-1}G_{P,T}=O(T^{-1})\to 0$ in operator
norm. Using the additive decomposition \eqref{eq:info_decomp},
\[
   \mathcal{I}_{\mathrm{joint},T}
   = \underbrace{G_{P,T}'\,S_{PP}^{-1}\,G_{P,T}}_{\to\, 0}
   \;+\; \tilde{G}_{NG,T}'\,\Sigma_{NG}^{-1}\,\tilde{G}_{NG,T},
   \qquad
   \tilde{G}_{NG,T}=G_{NG}-S_{NP}S_{PP}^{-1}G_{P,T}.
\]
Since $G_{P,T}\to 0$, we have $\tilde{G}_{NG,T}\to G_{NG}$, and therefore
$\mathcal{I}_{\mathrm{joint},T}\to
\mathcal{I}^{*}=G_{NG}'\Sigma_{NG}^{-1}G_{NG}$ in operator norm. By
Lemma~\ref{lem:ng_set}, $G_{NG}$ has full column rank at
$\phi_0$ and $\Sigma_{NG}\succ 0$, so $\mathcal{I}^{*}\succ 0$. 
 
\medskip
\noindent\textbf{Part \ref{localtozero_can}}
Note that under local-to-zero {relevance} the exclusion
conditions remain correctly specified, i.e., $\mathbb{E}[g_{P}(\phi_0)]=0$ holds.
Under the regularity conditions of Section \ref{ass:regularity} and
Assumptions \ref{as:validproxy}--\ref{as:gau}, the criterion
$\bar{g}_T(\phi)'\hat{S}^{-1}\bar{g}_T(\phi)$ obeys a uniform law of large
numbers with a limit uniquely minimized at $\phi_0$. The consistency theorem of
\citet[Thm.~2.6]{newey_mcfadden_1994} then gives
$\hat{\phi}_T\xrightarrow{p}\phi_0$. Combining the central limit theorem $\sqrt T\,\bar g_T(\phi_0)\xrightarrow{d}N(0,S)$, which
holds under condition 2 of Section~\ref{ass:regularity} by the CLT for $\alpha$-mixing
sequences \citep{ibragimov_some_1962}, gives us
\[
   \sqrt{T}\,(\hat{\phi}_T-\phi_0)
   \ \xrightarrow{d}\
   \mathcal{N}\!\bigl(0,\,(\mathcal{I}^{*})^{-1}\bigr),
   \qquad
   \mathcal{V}_{NG}^{*}=(\mathcal{I}^{*})^{-1}.
\]

\medskip
\noindent\textbf{Part \ref{localtozero_effoverng}}
The Schur complement satisfies
$\Sigma_{NG}=S_{NN}-S_{NP}S_{PP}^{-1}S_{PN}\preceq S_{NN}$, hence
$\Sigma_{NG}^{-1}\succeq S_{NN}^{-1}$ and
$\mathcal{I}^{*}=G_{NG}'\Sigma_{NG}^{-1}G_{NG}
\succeq G_{NG}'S_{NN}^{-1}G_{NG}$. The order-reversing property of
inversion on the positive-definite cone gives
$\mathcal{V}_{NG}^{*}=(\mathcal{I}^{*})^{-1}
\preceq (G_{NG}'S_{NN}^{-1}G_{NG})^{-1}=\mathcal{V}_{NG}$, with equality
if $S_{NP}=0$. 
\end{proof}

\begin{remark}
Part~\ref{localtozero_effoverng} states $\mathcal V_{NG}^{*}\preceq\mathcal V_{NG}$, with
equality iff $S_{NP}=0$. Under the local-to-zero sequence $S_{NP,T}\to0$, so the bound
holds with equality in the limit. For any {fixed} nonzero relevance, however,
$S_{NP}\neq0$ and the inequality is strict: the efficiency gain over the standalone
non-Gaussian estimator is a fixed-relevance property that attenuates as the instrument
weakens and disappears only in the limit.
\end{remark}

\subsection{Proof of Proposition \ref{prop:uniform}}\label{appendix:proof_uniform}
First we prove the Lemma \ref{lem:uniform_info}, then we proceed with the proof of the proposition.
\begin{lemma}[Uniform information bound]\label{lem:uniform_info}
Under Assumption \ref{as:uniform_ng},
\[
   \inf_{\gamma\in\mathcal P}\ \lambda_{\min}\!\big(G'S^{-1}G\big)\ \ge\
   \underline\kappa\ \equiv\ \underline g^{\,2}/\bar s\ >\ 0 .
\]
\end{lemma}
\begin{proof}
By the information decomposition \eqref{eq:info_decomp}, written with the
non-Gaussian block marginal,
\[
   G'S^{-1}G
   = G_{NG}'S_{NN}^{-1}G_{NG}
     + \tilde G_P'\Sigma_P^{-1}\tilde G_P
   \ \succeq\ G_{NG}'S_{NN}^{-1}G_{NG},
\]
since $\tilde G_P'\Sigma_P^{-1}\tilde G_P\succeq0$. The dropped term is the only component which is dependent on the proxy strength, so the bound is independent of $\Psi$.
Then $G_{NG}'S_{NN}^{-1}G_{NG}\succeq \lambda_{\max}(S_{NN})^{-1}\,G_{NG}'G_{NG}
\succeq (\sigma_{\min}(G_{NG})^2/\bar s)\,I_{p_\phi}\succeq(\underline g^2/\bar s)I_{p_\phi}$,
using $\lambda_{\max}(S_{NN})\le\bar s$ and Assumption \ref{as:uniform_ng}.
\end{proof}
Now we show the main proof of the proposition:
\begin{proof}[Sketch of the proof]
Lemma \ref{lem:uniform_info} places every $\gamma\in\mathcal P$ in the strongly
identified regime of \citet{andrews_estimation_2012}, with the parameter of interest
$\phi$ strongly identified by the non-Gaussian block irrespective of $\Psi$. The
four standard steps are then uniform over $\mathcal P$.
\emph{(i) Uniform consistency.} The criterion
$\bar g_T(\phi)'\hat{S}^{-1}\bar g_T(\phi)$ converges uniformly to a limit with a
well-separated minimiser at $\phi_0$ (uniform LLN, \citealp{andrews_generic_1992};
separation from the curvature bound of Lemma \ref{lem:uniform_info}), so
$\sup_{\gamma\in\mathcal P}\|\hat{\phi}_T-\phi_0\|=o_p(1)$.
\emph{(ii) Uniform linearization.} A mean-value expansion of the first-order
conditions, with $G'S^{-1}G$ uniformly nonsingular (Lemma \ref{lem:uniform_info})
and uniform stochastic equicontinuity \citep{andrews_empirical_1994}, gives
$\sqrt T(\hat{\phi}_T-\phi_0)=-(G'S^{-1}G)^{-1}G'S^{-1}T^{-1/2}\sum_t(g_t+\Gamma_\theta\zeta_t)+r_T$,
$\sup_{\gamma\in\mathcal P}\|r_T\|=o_p(1)$, where $\zeta_t$ is the first-stage
influence function of Section \ref{sec:joint_asymp}.
\emph{(iii) Uniform CLT.} The central limit theorem for triangular arrays of near-epoch dependent functions of
$\alpha$-mixing processes \citep{davidson_stochastic_1994}, applied with the mixing
coefficients and $(2+\delta)$-moment bounds that are uniform over $\mathcal P$ by
Assumption~\ref{as:uniform_ng}, yields
$T^{-1/2}\sum_t(g_t+\Gamma_\theta\zeta_t)\xrightarrow{d}\mathcal N(0,\Omega_\psi)$
uniformly, hence $\sqrt T(\hat{\phi}_T-\phi_0)\xrightarrow{d}\mathcal N(0,\mathcal V_\phi)$
uniformly, with $\mathcal V_\phi$ as in \eqref{eq:vphi_corrected}.
\emph{(iv) Uniform Studentization.} $\hat{\mathcal V}_\phi\xrightarrow{p}\mathcal V_\phi$
uniformly (with the inverse uniformly continuous by Lemma
\ref{lem:uniform_info}) over $\mathcal{P}$ by condition 4 of the regularity conditions \ref{ass:regularity} and Assumption \ref{as:uniform_ng}, and $R\mathcal V_\phi R'$ is uniformly nonsingular, so by
the generic uniformity results of \citet{andrews_generic_2020},
$W_T\xrightarrow{d}\chi^2_d$ uniformly over $\mathcal P$. The coverage statement
follows by test inversion.
\end{proof}

\subsection{Proof of Proposition \ref{prop:bias_bound_local}}\label{appendix:bias_bound_local}
\begin{proof}
Throughout, the estimator uses the efficient weight $W=S^{-1}$, and
$\mathbb{E}[g(\phi_0)]=(\,0',\,\mu_T'\,)'=(\,0',\,c'/\sqrt{T}\,)'$ with
$\mathbb{E}[g_{NG}(\phi_0)]=0$ and $\mathbb{E}[g_P(\phi_0)]=\mu_T=c/\sqrt{T}$.

Under the regularity conditions of \citet{newey_mcfadden_1994}, the pseudo-true
parameter $\phi_{*,T}$ is an interior minimizer of the population criterion
$\mathcal{Q}_T(\phi)=\mathbb{E}[g(\phi)]'S^{-1}\mathbb{E}[g(\phi)]$ and therefore
satisfies the population first-order conditions
\begin{equation}
    G(\phi_{*,T})'\,S^{-1}\,\mathbb{E}[g(\phi_{*,T})]=0,
    \label{eq:pop_foc_eff}
\end{equation}
where $G(\phi)\equiv\partial\mathbb{E}[g(\phi)]/\partial\phi'$ is partitioned
conformably as $G(\phi)=(G_{NG}(\phi)',G_P(\phi)')'$. Since
$\mathbb{E}[g(\phi_0)]=O(T^{-1/2})$, the criterion is minimised within an
$O(T^{-1/2})$ neighborhood of $\phi_0$, so $\phi_{*,T}-\phi_0=O(T^{-1/2})$.

Let $\delta_T=\phi_{*,T}-\phi_0=O(T^{-1/2})$. A first-order mean-value expansion of
$\mathbb{E}[g]$ around $\phi_0$ gives
\begin{equation}
    \mathbb{E}[g(\phi_{*,T})]
    =\mathbb{E}[g(\phi_0)]+G(\phi_0)\,\delta_T+O(\|\delta_T\|^2)
    =\mathbb{E}[g(\phi_0)]+G(\phi_0)\,\delta_T+O(T^{-1}),
    \label{eq:taylor_g_eff}
\end{equation}
and the corresponding expansion of the Jacobian gives
\begin{equation}
    G(\phi_{*,T})=G(\phi_0)+O(\|\delta_T\|)=G(\phi_0)+O(T^{-1/2}).
    \label{eq:taylor_G_eff}
\end{equation}
Substituting \eqref{eq:taylor_g_eff}--\eqref{eq:taylor_G_eff} into
\eqref{eq:pop_foc_eff} and expanding,
\begin{equation}
    G(\phi_0)'S^{-1}\mathbb{E}[g(\phi_0)]
    +G(\phi_0)'S^{-1}G(\phi_0)\,\delta_T
    +O(T^{-1/2})\!\cdot\!O(T^{-1/2})+O(T^{-1})=0,
    \label{eq:expanded_eff}
\end{equation}
where the third term is the cross-product of the $O(T^{-1/2})$ Jacobian remainder
in \eqref{eq:taylor_G_eff} with $\mathbb{E}[g(\phi_0)]=O(T^{-1/2})$. All remainder
terms are therefore $O(T^{-1})$.

Since $W=S^{-1}$ is {not} block-diagonal, the contamination is transmitted
through the off-diagonal blocks of $S^{-1}$. Writing the inverse via the Schur
complement of $S_{NN}$, $\Sigma_P\equiv S_{PP}-S_{PN}S_{NN}^{-1}S_{NP}$,
\begin{equation}
    S^{-1}=
    \begin{pmatrix}
    S_{NN}^{-1}+S_{NN}^{-1}S_{NP}\Sigma_P^{-1}S_{PN}S_{NN}^{-1}
        & -S_{NN}^{-1}S_{NP}\Sigma_P^{-1}\\[2pt]
    -\Sigma_P^{-1}S_{PN}S_{NN}^{-1} & \Sigma_P^{-1}
    \end{pmatrix},
    \label{eq:Sinv_block}
\end{equation}
so that, using $\mathbb{E}[g(\phi_0)]=(0',\mu_T')'$,
\begin{equation}
    S^{-1}\mathbb{E}[g(\phi_0)]
    =\begin{pmatrix}-S_{NN}^{-1}S_{NP}\Sigma_P^{-1}\mu_T\\[2pt]
    \Sigma_P^{-1}\mu_T\end{pmatrix}.
    \label{eq:Sinv_mu}
\end{equation}
Pre-multiplying by $G(\phi_0)'=(G_{NG}',G_P')$ and collecting terms,
\begin{equation}
    G(\phi_0)'S^{-1}\mathbb{E}[g(\phi_0)]
    =\bigl(G_P-S_{PN}S_{NN}^{-1}G_{NG}\bigr)'\Sigma_P^{-1}\mu_T
    =\tilde G_P'\,\Sigma_P^{-1}\,\mu_T,
    \label{eq:GSg_eff}
\end{equation}
where $\tilde G_P=G_P-S_{PN}S_{NN}^{-1}G_{NG}$ is the partialed proxy Jacobian and
we used $(S_{PN}S_{NN}^{-1}G_{NG})'=G_{NG}'S_{NN}^{-1}S_{NP}$ ($S_{NN}=S_{NN}'$,
$S_{PN}'=S_{NP}$). The same identity applied to $G(\phi_0)'S^{-1}G(\phi_0)$ yields
the (non-Gaussian-marginal) information decomposition
\begin{equation}
    G'S^{-1}G
    =G_{NG}'S_{NN}^{-1}G_{NG}
    +\tilde G_P'\,\Sigma_P^{-1}\,\tilde G_P,
    \label{eq:info_decomp_NGmarginal}
\end{equation}
the symmetric counterpart of \eqref{eq:info_decomp} with the roles of the two
moment blocks interchanged. By Lemma~\ref{lem:ng_set},
$G_{NG}$ has full column rank, so $G_{NG}'S_{NN}^{-1}G_{NG}\succ0$ and hence
$G'S^{-1}G\succ0$.

Substituting \eqref{eq:GSg_eff} into \eqref{eq:expanded_eff},
\begin{equation}
    \bigl(G'S^{-1}G\bigr)\,\delta_T
    =-\,\tilde G_P'\,\Sigma_P^{-1}\,\mu_T+O(T^{-1})
    =-\frac{1}{\sqrt{T}}\,\tilde G_P'\,\Sigma_P^{-1}\,c+O(T^{-1}).
    \label{eq:linear_system_eff}
\end{equation}
Multiplying by $\sqrt{T}$, setting $\mathbb{B}_T=\sqrt{T}\,\delta_T$, and
pre-multiplying by $(G'S^{-1}G)^{-1}$,
\begin{equation}
    \mathbb{B}_T
    =-\bigl(G'S^{-1}G\bigr)^{-1}\,\tilde G_P'\,\Sigma_P^{-1}\,c
    +O(T^{-1/2}),
    \label{eq:scaled_bias_eff}
\end{equation}
which establishes \eqref{eq:local_bias}. Taking $T\to\infty$ and recalling
$\delta_T=\mathbb{B}_T/\sqrt{T}\to0$ confirms $\phi_{*,T}\to\phi_0$, so $\hat{\phi}_T$
remains consistent for $\phi_0$ under local contamination.

Apply the central limit theorem of Section~\ref{ass:regularity} (condition 2) to the
recentered moment vector $\bigl(g_{NG}(\phi_0)',\,g_P(\phi_0)-\mu_T\bigr)'$, which
has mean zero and satisfies
$T^{-1/2}\sum_t\bigl(g_t-\mathbb{E}[g_t]\bigr)\xrightarrow{d}\mathcal{N}(0,S)$.
For the efficiently weighted estimator the influence-function representation is
$\sqrt{T}(\hat{\phi}_T-\phi_0)=-(G'S^{-1}G)^{-1}G'S^{-1}\,T^{1/2}\bar g_T(\phi_0)+o_p(1)$,
whose stochastic part has asymptotic variance
$(G'S^{-1}G)^{-1}G'S^{-1}\,S\,S^{-1}G(G'S^{-1}G)^{-1}=(G'S^{-1}G)^{-1}$, while the
deterministic drift $\mu_T=c/\sqrt{T}$ enters through \eqref{eq:GSg_eff} as the
non-centrality $-(G'S^{-1}G)^{-1}\tilde G_P'\Sigma_P^{-1}c$. Hence
\begin{equation*}
    \sqrt{T}(\hat{\phi}_T-\phi_0)\xrightarrow{d}
    \mathcal{N}\!\Big(-(G'S^{-1}G)^{-1}\tilde G_P'\Sigma_P^{-1}c,\
    (G'S^{-1}G)^{-1}\Big),
\end{equation*}
 When the innovations are estimated, the variance
$(G'S^{-1}G)^{-1}$ carries the generated-regressor correction of
Section~\ref{sec:joint_asymp}; this leaves the non-centrality unchanged.

Taking Euclidean norms in \eqref{eq:scaled_bias_eff} and applying
sub-multiplicativity of the spectral norm,
\begin{equation}
    \|\mathbb{B}_T\|_2
    \le\bigl\|(G'S^{-1}G)^{-1}\bigr\|_2\,
    \bigl\|\tilde G_P'\Sigma_P^{-1}c\bigr\|_2+O(T^{-1/2})
    =\frac{\bigl\|\tilde G_P'\Sigma_P^{-1}c\bigr\|_2}
          {\lambda_{\min}(G'S^{-1}G)}+O(T^{-1/2}),
    \label{eq:norm_step_eff}
\end{equation}
using $\|A^{-1}\|_2=1/\lambda_{\min}(A)$ for symmetric positive-definite $A$. By
the decomposition \eqref{eq:info_decomp_NGmarginal} and
$\tilde G_P'\Sigma_P^{-1}\tilde G_P\succeq0$,
\begin{equation}
    \lambda_{\min}\bigl(G'S^{-1}G\bigr)
    \ \ge\ \lambda_{\min}\bigl(G_{NG}'S_{NN}^{-1}G_{NG}\bigr)\ >\ 0,
    \label{eq:anchor_eff}
\end{equation}
so replacing the denominator in \eqref{eq:norm_step_eff} with
$\lambda_{\min}(G_{NG}'S_{NN}^{-1}G_{NG})$ yields the valid upper bound
\eqref{eq:local_bias_bound}. The minorant
$\lambda_{\min}(G_{NG}'S_{NN}^{-1}G_{NG})$ is the marginal non-Gaussian
information and coincides with the anchor of Lemma~\ref{lem:uniform_info}.
\end{proof}

\subsection{Limiting distributions of the test statistics}\label{app:limitingdistributions}

\subsubsection{Setup}
\label{app:B:setup}

All notation is as defined in Section \ref{appendix:preliminaries}.
We additionally maintain Assumption \ref{ass:B1} below, which holds under Assumptions~\ref{as:indep}--\ref{as:gau} by Lemma~\ref{lem:ng_set}.

\begin{assumption}\label{ass:B1}
$G_{NG} \in \mathbb{R}^{q_{NG}\times p_\phi}$ has full column rank
$p_\phi$, with $q_{NG} \geq p_\phi$.
\end{assumption}

Assumption \ref{ass:B1} implies $G$ has full column rank, so $G'S^{-1}G \succ 0$
and the annihilator $\Omega = S - G(G'S^{-1}G)^{-1}G'$ is positive
semidefinite of rank $q - p_\phi$.

\subsubsection{Properties of $L$ and $K$}
\label{app:B:LK}

As defined in Section \ref{appendix:preliminaries} Equation \eqref{eq:app_fwloperators}, 
\begin{lemma}\label{lem:LK}
Under $S > 0$:
\begin{align}
  LS &= [\Sigma_P,\; 0], \quad KS = [0,\; \Sigma_{NG}],
  \label{eq:LS_KS}\\
  SL' &= (\Sigma_P', 0)', \quad SK' = (0', \Sigma_{NG}')',
  \label{eq:SLt_SKt}\\
  \mathrm{Cov}(Lg(\phi_0),\, Kg(\phi_0)) &= LSK' = 0.
  \label{eq:orth}
\end{align}
\end{lemma}

\begin{proof}
Direct block multiplication establishes \eqref{eq:LS_KS}:
\begin{align*}
  LS &= [I_{q_P},\ {-}S_{PN}S_{NN}^{-1}]
       \begin{bmatrix}S_{PP}&S_{PN}\\S_{NP}&S_{NN}\end{bmatrix}
     = [S_{PP} - S_{PN}S_{NN}^{-1}S_{NP},\ 0]
     = [\Sigma_P,\ 0],\\
  KS &= [{-}S_{NP}S_{PP}^{-1},\ I_{q_{NG}}]
       \begin{bmatrix}S_{PP}&S_{PN}\\S_{NP}&S_{NN}\end{bmatrix}
     = [0,\ S_{NN} - S_{NP}S_{PP}^{-1}S_{PN}]
     = [0,\ \Sigma_{NG}].
\end{align*}
Equations \eqref{eq:SLt_SKt} follow by transposition. For
\eqref{eq:orth}: $LSK' = [\Sigma_P,\,0](0',\Sigma_{NG}')' = 0$.
\end{proof}

\subsubsection{Projected annihilator representations}
\label{app:B:proj}

\begin{lemma}\label{lem:proj}
Under assumptions of Proposition \ref{prop:identification_unified} and \ref{ass:B1}, with annihilator
$M = I_q - G(G'S^{-1}G)^{-1}G'S^{-1}$:
\begin{equation}
  T^{1/2}\tilde{g}_{P,T}(\hat{\phi}_T) = LM\,T^{1/2}\bar{g}_T(\phi_0)
  + o_p(1), \qquad
  T^{1/2}\tilde{g}_{NG,T}(\hat{\phi}_T) = KM\,T^{1/2}\bar{g}_T(\phi_0)
  + o_p(1).
  \label{eq:proj_rep}
\end{equation}
The asymptotic variances and cross-covariance are
\begin{align}
  \mathcal Q_P &\equiv L\Omega L'
      = \Sigma_P - \tilde{G}_P(G'S^{-1}G)^{-1}\tilde{G}_P',
      \label{eq:QP}\\
  \mathcal Q_{NG} &\equiv K\Omega K'
         = \Sigma_{NG} - \tilde{G}_{NG}(G'S^{-1}G)^{-1}\tilde{G}_{NG}',
         \label{eq:QNG}\\
  L\Omega K' &= -\tilde{G}_P(G'S^{-1}G)^{-1}\tilde{G}_{NG}'.
  \label{eq:cross}
\end{align}
\end{lemma}

\begin{proof}
Pre-multiplying the full annihilator representation
$T^{1/2}\bar{g}_T(\hat{\phi}_T) = M\,T^{1/2}\bar{g}_T(\phi_0) + o_p(1)$
by $L$ and $K$ gives \eqref{eq:proj_rep}. For the variances, use
$\Omega = MSM'= S - G(G'S^{-1}G)^{-1}G'$:
\[
  L\Omega L' = LSL' - LG(G'S^{-1}G)^{-1}G'L'.
\]
From \eqref{eq:SLt_SKt}, $LSL' = [I_{q_P},\,{-}S_{PN}S_{NN}^{-1}]
(\Sigma_P', 0)' = \Sigma_P$. With $LG = \tilde{G}_P$, equation
\eqref{eq:QP} follows. The derivation of \eqref{eq:QNG} is identical
with $K$ in place of $L$, using $KSK' = \Sigma_{NG}$ from
\eqref{eq:LS_KS}. For \eqref{eq:cross}:
$L\Omega K' = LSK' - LG(G'S^{-1}G)^{-1}G'K' =
0 - \tilde{G}_P(G'S^{-1}G)^{-1}\tilde{G}_{NG}'$,
where $LSK' = 0$ by \eqref{eq:orth}. 
\end{proof}

\begin{remark}[Non-independence]
Equation \eqref{eq:cross} shows that $J_{\mathrm{prx}}$ and $J_{NG}$ are not
asymptotically independent in general: their limiting covariance is zero
if and only if $\tilde{G}_P(G'S^{-1}G)^{-1}\tilde{G}_{NG}' = 0$,
which is not implied by the maintained assumptions.
\end{remark}

\subsubsection{Ranks of $\mathcal Q_P$ and $\mathcal Q_{NG}$}
\label{app:B:ranks}

The null space of $\Omega$ is $\mathrm{null}(\Omega) =
\mathrm{col}(S^{-1}G)$ with dimension $p_\phi$, since $\Omega v = 0
\Leftrightarrow v = S^{-1}G\alpha$ for $\alpha =
(G'S^{-1}G)^{-1}G'v$.

\begin{lemma}[Full rank of $\mathcal Q_P$]\label{lem:QP_rank}
Under Assumption \ref{ass:B1}, $\mathcal Q_P \succ 0$ and $\mathrm{rank}(\mathcal Q_P) = q_P$.
\end{lemma}

\begin{proof}
Let $\mathcal Q_P v = 0$. Since $\mathcal Q_P \succeq 0$, this gives $(L'v)'\Omega(L'v)
= 0$, so $L'v \in \mathrm{null}(\Omega) = \mathrm{col}(S^{-1}G)$.
Thus $L'v = S^{-1}G\alpha$ for some $\alpha$. Pre-multiplying by $S$
and applying \eqref{eq:SLt_SKt}:
\[
  \begin{pmatrix}\Sigma_P v \\ 0\end{pmatrix}
  = SL'v = G\alpha
  = \begin{pmatrix}G_P\alpha \\ G_{NG}\alpha\end{pmatrix}.
\]
The lower block gives $G_{NG}\alpha = 0$. Since $G_{NG}$ has full
column rank (Assumption \ref{ass:B1}), $\alpha = 0$. Then $\Sigma_P v = 0$, and
$\Sigma_P \succ 0$ gives $v = 0$. 
\end{proof}

\begin{lemma}[Rank of $\mathcal Q_{NG}$]\label{lem:QNG_rank}
Under Assumption \ref{ass:B1},
\begin{equation}
  \mathrm{rank}(\mathcal Q_{NG}) = q_{NG} - (p_\phi - \mathrm{rank}(G_P)).
  \label{eq:rank_QNG}
\end{equation}
In particular, $\mathcal Q_{NG} \succ 0$ if and only if
$\mathrm{rank}(G_P) = p_\phi$.
\end{lemma}

\begin{proof}
Let $\mathcal Q_{NG}v = 0$. The same argument gives $K'v = S^{-1}G\alpha$, and
applying \eqref{eq:SLt_SKt}:
\[
  \begin{pmatrix}0 \\ \Sigma_{NG}v\end{pmatrix}
  = SK'v = G\alpha
  = \begin{pmatrix}G_P\alpha \\ G_{NG}\alpha\end{pmatrix}.
\]
This yields the system
\begin{equation}
  G_P\alpha = 0, \qquad G_{NG}\alpha = \Sigma_{NG}v.
  \label{eq:sys_NG}
\end{equation}
The first equation confines $\alpha$ to
$\mathcal{N}_P \equiv \mathrm{null}(G_P)$, which has dimension
$p_\phi - \mathrm{rank}(G_P)$.

The map $\varphi: \mathcal{N}_P \to \mathbb{R}^{q_{NG}}$ defined by
$\varphi(\alpha) = \Sigma_{NG}^{-1}G_{NG}\alpha$ is a bijection onto
$\mathrm{null}(\mathcal Q_{NG})$:
\begin{itemize}
  \item \emph{Injective}: if $\varphi(\alpha) = 0$ then
    $G_{NG}\alpha = 0$; full column rank of $G_{NG}$
    (Assumption \ref{ass:B1}) forces $\alpha = 0$.
  \item \emph{Surjective}: every $v \in \mathrm{null}(\mathcal Q_{NG})$
    satisfies \eqref{eq:sys_NG} for some $\alpha \in \mathcal{N}_P$,
    and $v = \Sigma_{NG}^{-1}G_{NG}\alpha = \varphi(\alpha)$.
\end{itemize}
Therefore $\dim(\mathrm{null}(\mathcal Q_{NG})) = \dim(\mathcal{N}_P) = p_\phi -
\mathrm{rank}(G_P)$, giving \eqref{eq:rank_QNG}. 
\end{proof}

\begin{remark}[Degrees-of-freedom compression]
\label{rem:df}
In a SVAR with $n$ variables and $k$ proxy instruments,
$p_\phi = n(n-1)/2$ and $q_P = k(n-1)$, we have $\mathrm{rank}(G_P) < p_\phi$
generically, and $\mathcal Q_{NG}$ is rank-deficient. As a concrete example, for
$n = 3$ and $k = 1$: $q_P = 2 < p_\phi = 3$, so
$\mathrm{rank}(\mathcal Q_{NG}) \leq q_{NG} - 1$.
\end{remark}

\subsubsection{Asymptotic distributions}
\label{app:B:asym}

\begin{lemma}[Quadratic form under singular normal]\label{lem:chi2}
Let $Z \sim \mathcal{N}(0, A)$ with $A \succeq 0$ of rank $r$.
Then $Z'A^+Z \sim \chi^2(r)$.
\end{lemma}

\begin{proof}
Let $A = U\Lambda U'$ with $\Lambda = \mathrm{diag}(\lambda_1, \ldots,
\lambda_r, 0, \ldots, 0)$, $\lambda_i > 0$, $U$ orthogonal, and
$A^+ = U\Lambda^+U'$. Set $W = U'Z \sim \mathcal{N}(0, \Lambda)$ and
partition $W = (W_1', W_2')'$ with $W_1 \in \mathbb{R}^r$. Since
$\mathrm{Var}(W_2) = 0$, $Z$ is supported on $\mathrm{col}(A)$, so
$W_2 = U_2'Z = 0$ a.s.\ Therefore:
\[
  Z'A^+Z = W'\Lambda^+W = \sum_{i=1}^r \frac{W_{1,i}^2}{\lambda_i}.
\]
Each $W_{1,i}/\sqrt{\lambda_i} \sim \mathcal{N}(0,1)$ independently,
so the sum is $\chi^2(r)$. 
\end{proof}

\textbf{Proof of Proposition \ref{prop:limitdist}}
\begin{proof}
\textit{Part (i).} By \eqref{eq:proj_rep} and the CLT from the regularity conditions \ref{ass:regularity},
$T^{1/2}\tilde{g}_{P,T}(\hat{\phi}_T) \xrightarrow{d} Z_P \sim
\mathcal{N}(0, \mathcal Q_P)$, since $LMS(LM)' = L\Omega L' = \mathcal Q_P$. By
Lemma \ref{lem:QP_rank}, $\mathcal Q_P \succ 0$, so $\hat{\mathcal Q}_P^{-1}
\xrightarrow{p} \mathcal Q_P^{-1}$ and the continuous mapping theorem gives
$J_{\mathrm{prx}} \xrightarrow{d} Z_P'\mathcal Q_P^{-1}Z_P = \|\mathcal Q_P^{-1/2}Z_P\|^2 \sim
\chi^2(q_P)$.

\textit{Part (ii).} Identical to Part (i) with $K$, $\mathcal Q_{NG}$, $q_{NG}$
in place of $L$, $\mathcal Q_P$, $q_P$, using $\mathcal Q_{NG} \succ 0$ from
Lemma \ref{lem:QNG_rank}.

\textit{Part (iii).} We have $T^{1/2}\tilde{g}_{NG,T}(\hat{\phi}_T)
\xrightarrow{d} Z_{NG} \sim \mathcal{N}(0, \mathcal Q_{NG})$ with
$\mathrm{rank}(\mathcal Q_{NG}) = r_{NG}$ by Lemma \ref{lem:QNG_rank}.

Let $\mathcal Q_{NG} = V\Lambda_{NG}V'$ be the spectral decomposition, with
$\Lambda_{NG} = \mathrm{diag}(\mu_1, \ldots, \mu_{r_{NG}}, 0, \ldots,
0)$, $\mu_i > 0$, and $\Pi$ the projector onto $\mathrm{col}(\mathcal Q_{NG})$
(spanned by the first $r_{NG}$ columns of $V$). The thresholded
estimator $\hat\Pi$ retains eigenvectors of $\hat{\mathcal Q}_{NG}$ with
eigenvalues $\geq\tau_T$. Since the positive eigenvalues of $\mathcal Q_{NG}$
are bounded away from zero and $\hat{\mathcal Q}_{NG} \xrightarrow{p} \mathcal Q_{NG}$,
Weyl's inequality implies $\hat\mu_i \xrightarrow{p} \mu_i > 0$ for
$i \leq r_{NG}$ and $\hat\mu_j \xrightarrow{p} 0$ for $j > r_{NG}$.
The rate condition $\epsilon_T \to 0$, $T^{1/2}\epsilon_T \to \infty$
ensures $\hat\Pi \xrightarrow{p} \Pi$ and
$\hat\Pi\hat{\mathcal Q}_{NG}^+\hat\Pi \xrightarrow{p} \mathcal Q_{NG}^+$
\citep{stewart_rank_1977} Theorem 2.3.

Since $Z_{NG}$ is supported on $\mathrm{col}(\mathcal Q_{NG}) =
\mathrm{col}(\Pi)$ a.s., we have $\Pi Z_{NG} = Z_{NG}$ a.s. The
continuous mapping theorem then gives
\[
  J_{NG} \xrightarrow{d} Z_{NG}'\mathcal Q_{NG}^+Z_{NG} \sim \chi^2(r_{NG}),
\]
where the last step applies Lemma \ref{lem:chi2} with $A = \mathcal Q_{NG}$
and $r = r_{NG}$. 
\end{proof}

\begin{remark}[Structural interpretation of the asymmetry]
$J_{\mathrm{prx}}$ uses the jointly identified estimator $\hat{\phi}_T$; the
full-rank $G_{NG}$ absorbs all $p_\phi$ estimation degrees of freedom,
leaving the proxy moments free to span all $q_P$ dimensions. For
$J_{NG}$, only $\mathrm{rank}(G_P)$ of the $p_\phi$ identification
directions are supplied by the proxy block; the remaining $p_\phi -
\mathrm{rank}(G_P)$ directions compress the effective degrees of freedom
below $q_{NG}$.
\end{remark}

\begin{remark}[Symmetric contamination]
By the same argument, proxy mis-specification
($\mathbb{E}[g_P(\phi_0)] \neq 0$) drives $\hat{\phi}_T$ to $\phi^{**}$
with $\mathbb{E}[g_{NG}(\phi^{**})] \neq 0$ generically, inflating
$J_{NG}$. This non-separability is why both statistics are constructed
from the jointly identified $\hat{\phi}_T$.
\end{remark}

\subsubsection{Summary of degrees of freedom}
\label{app:B:dof}

Table \ref{tab:dof} collects the degrees of freedom for each test
statistic.

\begin{table}[h]
\centering
\renewcommand{\arraystretch}{1.3}
\begin{tabular}{lll}
\hline\hline
Statistic & d.f.\ (general) & d.f.\ ($\mathrm{rank}(G_P) = p_\phi$) \\
\hline
$J_{\mathrm{joint}}$ & $q - p_\phi$ & $q - p_\phi$ \\
$J_{\mathrm{prx}}$               & $q_P$        & $q_P$ \\
$J_{NG}$            & $q_{NG} - (p_\phi - \mathrm{rank}(G_P))$ & $q_{NG}$ \\
\hline\hline
\end{tabular}
\caption{Effective degrees of freedom for each test statistic. The
  general column applies whenever assumptions for Proposition \ref{prop:identification_unified} and Assumption \ref{ass:B1} hold.}
\label{tab:dof}
\end{table}

The sum of the subset degrees of freedom satisfies
\[
  q_P + \bigl(q_{NG} - (p_\phi - \mathrm{rank}(G_P))\bigr)
  = (q - p_\phi) + \mathrm{rank}(G_P).
\]
This exceeds the full-test degrees of freedom $q - p_\phi$ whenever
$\mathrm{rank}(G_P) > 0$, reflecting the fact that $J_{\mathrm{prx}}$ and $J_{NG}$
are not asymptotically independent (see the Remark following
Lemma \ref{lem:proj}): the additive decomposition $J_{\mathrm{full}}
= J_{\mathrm{prx}} + J_{NG}$ holds in distribution if and only if
$\tilde{G}_P(G'S^{-1}G)^{-1}\tilde{G}_{NG}' = 0$, which is not
generically satisfied.

\section{Generated-regressor correction}\label{app:genreg}
The first-stage estimator admits the expansion
$\sqrt{T}(\hat{\theta}_T-\theta_0)=T^{-1/2}\sum_t\zeta_t+o_p(1)$, with
$\mathcal{V}_\theta=\mathbb{E}[\zeta_t\zeta_t']$. The long-run variance of the
combined score $g_t+\Gamma_\theta\zeta_t$ is
\begin{align}
\Omega_\psi &\equiv \sum_{h=-\infty}^{\infty}\mathbb E\bigl[(g_t+\Gamma_\theta\zeta_t)(g_{t-h}+\Gamma_\theta\zeta_{t-h})'\bigr]
 = S+\Gamma_\theta \mathcal{V}_\theta\Gamma_\theta'+\Gamma_\theta \mathcal{C}'+\mathcal{C}\Gamma_\theta',
 \label{eq:Omega_corrected}\\
 \mathcal{C} &\equiv\sum_{h\ge 0}\mathbb E\bigl[g_t(\psi_0)\,\zeta_{t-h}'\bigr],
 \nonumber
\end{align}
where $S$ is the long-run variance of Section~\ref{ass:regularity} and $\mathcal{C}$
is the long-run cross-covariance between the structural moments and the first-stage
score. The sum defining $\mathcal{C}$ runs over $h\ge 0$ only, because $\zeta_t$ is a
martingale difference: $\mathbb E[g_t\zeta_{t+j}']=0$ for $j\ge 1$. The off-diagonal
block of \eqref{eq:joint_asymp} is
\begin{equation}
  \mathcal{V}_{\theta\phi}
  = -\,\big(\mathcal{V}_\theta\Gamma_\theta' + \mathcal{C}'\big)\,
     S^{-1}G_0\,(G_0'S^{-1}G_0)^{-1}.
  \label{eq:vthetaphi}
\end{equation}
The cross-covariance $\mathcal{C}$ is generically nonzero under non-Gaussian
identification. For the co-skewness coordinate $g=\varepsilon_{i}^2\varepsilon_{j}$
($i\neq j$) and $u_t=B_0\varepsilon_t$,
$\mathbb{E}[u_{k,t}\,\varepsilon_{i,t}^2\varepsilon_{j,t}]
= B_{0,kj}\,\mathbb{E}[\varepsilon_{j,t}^2\varepsilon_{i,t}^2]=B_{0,kj}\neq 0$,
so both $\mathcal{C}$ and $\mathcal{V}_{\theta\phi}$ must be retained.

The estimator $\hat{\Omega}_\psi$ in \eqref{eq:vphi_est} is
\begin{equation}
  \hat{\Omega}_\psi=\hat{S}+\hat\Gamma_\theta\hat{\mathcal V}_\theta\hat\Gamma_\theta'
    +\hat\Gamma_\theta\hat{\mathcal C}'+\hat{\mathcal C}\hat\Gamma_\theta',
\end{equation}
with
\begin{equation}
\begin{aligned}
\hat{G} &=T^{-1}\textstyle\sum_t \partial g_t(\hat\psi)/\partial\phi', &
\hat\Gamma_\theta &=T^{-1}\textstyle\sum_t \partial g_t(\hat\psi)/\partial\theta',\\
\hat\zeta_t &=\bigl(T^{-1}\textstyle\sum_s x_sx_s'\bigr)^{-1}x_t\otimes\hat{u}_t, &
\hat{\mathcal C} &=\textstyle\sum_{h=0}^{T-1}\omega_{h,T}\,T^{-1}\sum_{t=h+1}^{T}g_t(\hat\psi)\,\hat\zeta_{t-h}',
\end{aligned}
\end{equation}
where $\omega_{h,T}$ are the kernel weights of $\hat{S}_T$.

\section{Inference for structural impulse-response functions}\label{subsec:structural_irfs}

The structural VMA representation is $y_t = \mu + \sum_{h=0}^\infty
\Theta_h\varepsilon_{t-h}$, where $\Theta_h = \Xi_h B_0$ (see
Section \ref{appendix:preliminaries} for the definitions of $\Theta_h$,
$\Xi_h$, and related objects). The reduced-form VMA coefficients
satisfy the recursion
\begin{equation}
  \Xi_h = \sum_{j=1}^{\min(h,p)} A_j\Xi_{h-j}, \qquad \Xi_0 = I_n,
\end{equation}
and $[\Theta_h]_{ij}$ measures the response of variable $i$ to a unit
impulse in structural shock $j$ at horizon $h$.

Since $\Theta_h = \Xi_h(\theta)B_0(\theta,\phi)$ is a smooth function
of $\psi$, the composite Jacobian $\mathcal{J}_{\Theta_h} =
\partial\,\mathrm{vec}(\Theta_h)/\partial\psi'$ follows from the
product rule:
\begin{equation}
  \mathcal{J}_{\Theta_h}
  = (B_0'\otimes I_n)\,
    \frac{\partial\,\mathrm{vec}(\Xi_h)}{\partial\psi'}
  + (I_n\otimes\Xi_h)\,\mathcal{J}_{B_0}.
  \label{eq:irf_jac}
\end{equation}
Since $\Xi_h$ depends only on $\theta$, we have
$\partial\,\mathrm{vec}(\Xi_h)/\partial\phi' = 0$. Differentiating
the VMA recursion with respect to $\theta$ \citep{lutkepohl_2005}
gives
\begin{equation}
  \frac{\partial\,\mathrm{vec}(\Xi_h)}{\partial\theta'}
  = \sum_{j=1}^{\min(h,p)}
    \left[
      (I_n\otimes A_j)\,
      \frac{\partial\,\mathrm{vec}(\Xi_{h-j})}{\partial\theta'}
      + (\Xi_{h-j}'\otimes I_n)\,E_j
    \right],
  \label{eq:vma_deriv}
\end{equation}
where $(\Xi_{h-j}'\otimes I_n)E_j$ captures the direct effect of the
$j$-th lag matrix on $\Xi_h$, and $(I_n\otimes A_j)$ propagates the
derivative of $\Xi_{h-j}$ forward through the recursion. The delta
method applied to the joint asymptotic distribution of $\hat\psi$ with
Jacobian $\mathcal{J}_{\Theta_h}$ yields
\begin{equation}
  \sqrt{T}\,\mathrm{vec}(\hat\Theta_h - \Theta_h)
  \xrightarrow{d}
  \mathcal{N}\bigl(0,\;\mathcal{J}_{\Theta_h} \mathcal{V}_\psi
  \mathcal{J}_{\Theta_h}'\bigr).
  \label{eq:irf_asymp}
\end{equation}

\section{Residual-Based moving block bootstrap for HGMM estimator}\label{appendix:mbb}

The bootstrap follows \citet{bruggemann_inference_2016}. Their validity results cover statistics that are smooth functions of second moments of
$(u_t,m_t)$; for the higher-order moment conditions used here, validity of the
moving-block bootstrap for GMM statistics of $\alpha$-mixing data
follows from \citet{goncalves_maximum_2004}, under the moment and mixing conditions of
Section~\ref{ass:regularity}. The procedure
operates on the joint residual-proxy pairs
\[
  z_t = (\hat{u}_t', m_t')',
\]
where $\hat{u}_t = y_t - \hat\nu - \sum_{j=1}^p \hat{A}_jy_{t-j}$
for $t = 1, \ldots, T$ are OLS residuals. Construct the $T - \ell + 1$
overlapping blocks
\begin{equation}
  \mathcal{L}_i = \bigl(z_i,\, z_{i+1},\, \ldots,\, z_{i+\ell-1}\bigr),
  \qquad i = 1, \ldots, T - \ell + 1,
  \label{eq:mbb_blocks}
\end{equation}
where the block length $\ell$ satisfies $\ell \to \infty$ and $\ell^3/T
\to 0$ as $T \to \infty$. Following \citet{jentsch_asymptotically_2022}, we set
$\ell = \lfloor 5.03\,T^{1/4} \rfloor$.

For each replication $b = 1, \ldots, B$, execute the following steps.

\begin{enumerate}

\item \textbf{Block resampling.}
Draw $b_T = \lceil T/\ell \rceil$ blocks independently with replacement
from $\{\mathcal{L}_1, \ldots, \mathcal{L}_{T-\ell+1}\}$. Concatenate
and trim to length $T$, yielding the bootstrap joint sequence
$\{z_t^*\}_{t=1}^T = \{(\hat{u}_t^{*\prime}, m_t^{*\prime})'\}_{t=1}^T$.
By construction, the bootstrap blocks preserve the empirical joint
dependence between $\hat{u}_t$ and $m_t$ within each block.

\item \textbf{Bootstrap standardised innovations.}
Recursively generate the bootstrap time series $\hat{y}_t^*$ using the
autoregressive coefficients and the resampled residuals; then demean the
bootstrap residuals $\hat{u}_t^*$ and compute the bootstrap residual
covariance $\hat\Sigma^* = T^{-1}(U^{*\prime})U^*$. With the
lower-triangular Cholesky factor $H^*$ satisfying $H^*(H^*)' =
\hat\Sigma^*$, the bootstrap standardised innovations are
\begin{equation}
  \tilde{U}^* = (H^*)^{-1}(U^*)'.
  \label{eq:boot_innov}
\end{equation}

\item \textbf{$k$-step Newton--Raphson bootstrap estimation.}
Initialize the bootstrap parameter at the original estimate,
$\phi_0^* = \hat{\phi}_T$, and iterate the Newton--Raphson update
using the bootstrap moment matrix $G_m(\phi^*, \tilde{U}^*, M^*)$ and
the {fixed} original-data weighting matrix $W = \hat{S}^{-1}$. The
bootstrap gradient and Hessian are
\begin{equation}
  \nabla = G^{*\prime}W g^*, \qquad
  \mathcal{H} = G^{*\prime}W G^*,
  \label{eq:boot_nr}
\end{equation}
with update $\phi^* \leftarrow \phi^* - \mathcal{H}^{-1}\nabla$ until
both the step norm and gradient norm fall below tolerance $\epsilon =
10^{-6}$, or after a maximum of 100 iterations. This $k$-step procedure
follows \citet{andrews_2002} and achieves first-order asymptotic
equivalence to the full two-step bootstrap estimator at substantially
reduced computational cost.

\item \textbf{Recentering and bootstrap weighting matrix.}
Evaluate the recentered bootstrap moment vector
\begin{equation}
  g_{\mathrm{final}}^* = \bar{G}_m^* - \bar{g}_T(\hat{\phi}_T),
  \label{eq:boot_recenter}
\end{equation}
and the bootstrap weighting matrix
\begin{equation}
  S^* = \frac{1}{T}(G_m^* - \bar{G}_m^*)'(G_m^* - \bar{G}_m^*),
  \qquad W^* = (S^*)^{-1}.
  \label{eq:boot_S}
\end{equation}

\item \textbf{Bootstrap Schur complements and projected Jacobians.}
Partition $S^*$ and $G^*$ conformably with block sizes $(q_{NG}, q_P)$
and compute the bootstrap projection operators
\begin{equation}
  P_P^* = S_{PN}^*(S_{NN}^*)^{-1}, \qquad
  P_{NG}^* = S_{NP}^*(S_{PP}^*)^{-1}.
  \label{eq:boot_proj}
\end{equation}
The bootstrap projected Jacobians and Schur complements are
\begin{align}
  \tilde{G}_P^* &= G_P^* - P_P^*G_{NG}^*, &
  \tilde{G}_{NG}^* &= G_{NG}^* - P_{NG}^*G_P^*, \\
  \Sigma_P^* &= S_{PP}^* - P_P^*S_{NP}^*, &
  \Sigma_{NG}^* &= S_{NN}^* - P_{NG}^*S_{PN}^*.
  \label{eq:boot_schur}
\end{align}

\item \textbf{Parameter-estimation correction.}
Apply the estimation-degrees-of-freedom correction to each Schur
complement:
\begin{align}
  \mathcal Q_P^* &= \Sigma_P^*
          - \tilde{G}_P^*(G^{*\prime}W^*G^*)^{-1}(\tilde{G}_P^*)',
          \label{eq:boot_QP} \\
  \mathcal Q_{NG}^* &= \Sigma_{NG}^*
              - \tilde{G}_{NG}^*(G^{*\prime}W^*G^*)^{-1}
                (\tilde{G}_{NG}^*)'.
              \label{eq:boot_QNG}
\end{align}

\item \textbf{Bootstrap test statistics.}
The bootstrap orthogonalized moments are
\begin{equation}
  g_P^* = g_P^{*,\mathrm{fin}} - P_P^*g_{NG}^{*,\mathrm{fin}},
  \qquad
  g_{NG}^* = g_{NG}^{*,\mathrm{fin}} - P_{NG}^*g_P^{*,\mathrm{fin}}.
  \label{eq:boot_gstar}
\end{equation}
The bootstrap test statistics are
\begin{align}
  J_{\mathrm{joint}}^*
    &= T\cdot(g_{\mathrm{final}}^*)'\,W^*\,g_{\mathrm{final}}^*,
    \label{eq:boot_j_joint}\\
  J_{\mathrm{prx}}^*
    &= T\cdot(g_P^*)'\,(\mathcal Q_P^*)^+\,g_P^*,
    \label{eq:boot_j_prx}\\
  J_{NG}^*
    &= T\cdot(g_{NG}^*)'\,(\mathcal Q_{NG}^*)^+\,g_{NG}^*.
    \label{eq:boot_j_ng}
\end{align}

\end{enumerate}

\section{Estimates under Weak non-Gaussianity} \label{appendix:weak_ng_sim}
\begin{table}[h!]
\centering
\caption{Parameter Estimates under $H_0$: Weak non-Gaussianity (Pitman Drift $c = 2.24$)}
\label{tab:phi_estimates_weakng}
\begin{tabular}{lccccc}
\toprule
 & True Value & $T = 200$ & $T = 300$ & $T = 500$ & $T = 1000$  \\
\midrule
\multicolumn{6}{l}{\textit{Estimates and their standard errors}} \\[3pt]
$\hat{\phi}_1$ & $0.80$    & 0.7881    & 0.7996    & 0.8009    & 0.8002     \\
            &  & (0.0356) & (0.0295) & (0.0231) & (0.0165)  \\
$\hat{\phi}_2$ & $-0.40$   & $-$0.3976 & $-$0.3992 & $-$0.3996 & $-$0.3999 \\
                &  & (0.0358) & (0.0299) & (0.0235) & (0.0168)  \\
$\hat{\phi}_3$ & $1.20$    & 1.2013    & 1.1993    & 1.2044    & 1.1989     \\
            &  & (0.0373) & (0.0312) & (0.0243) & (0.0174)  \\[6pt]
\bottomrule
\end{tabular}
\begin{minipage}{0.95\textwidth}
\vspace{4pt}
\footnotesize
\textit{Notes:} Results based on $M = 500$ Monte Carlo replications with $N_B = 999$ bootstrap replications. Structural shocks follow a Pitman mixture: $\varepsilon_{kt} = w \cdot \varepsilon_{kt}^{NG} + (1-w) \cdot z_{kt}$, where $w = c/\sqrt{T}$ with $c = 2.24$, $\varepsilon^{NG}$ are the baseline non-Gaussian shocks, and $z_{kt} \sim N(0,1)$. As $T$ increases, the shocks approach Gaussianity ($w \approx 0.158, 0.129, 0.100, 0.071, 0.022$). Proxies are valid with relevance $\Psi = (0.70, 0.60, 0.80)'$ and noise $\sigma_v = 0.50$.
\end{minipage}
\end{table}

\begin{table}[ht]
\centering
\caption{Parameter Estimates under $H_0$: Weak Proxies and Weak non-Gaussianity (Pitman Drift $c = 2.24$)}
\label{tab:phi_estimates_weakpng}
\begin{tabular}{lccccc}
\toprule
 & True Value & $T = 200$ & $T = 300$ & $T = 500$ & $T = 1000$ \\
\midrule
\multicolumn{6}{l}{\textit{Estimates and their standard errors}} \\[3pt]
$\hat{\phi}_1$ & $0.80$    & 0.7619    & 0.7942    & 0.8001    & 0.7984 \\
               &  & (0.1218) & (0.1256) & (0.1275) & (0.1314) \\
$\hat{\phi}_2$ & $-0.40$   & $-$0.3776 & $-$0.3931 & $-$0.3685 & $-$0.3962 \\
               &  & (0.1092) & (0.1117) & (0.1163) & (0.1207) \\
$\hat{\phi}_3$ & $1.20$    & 1.0839    & 1.1105    & 1.1493    & 1.1229 \\
               &  & (0.1202) & (0.1275) & (0.1273) & (0.1320) \\[6pt]
\bottomrule
\end{tabular}
\begin{minipage}{0.95\textwidth}
\vspace{4pt}
\footnotesize
\textit{Notes:} Results based on $M = 500$ Monte Carlo replications with $N_B = 999$ bootstrap replications. Structural shocks follow a Pitman mixture: $\varepsilon_{kt} = w \cdot \varepsilon_{kt}^{NG} + (1-w) \cdot z_{kt}$, where $w = c/\sqrt{T}$ with $c = 2.24$, $\varepsilon^{NG}$ are the baseline non-Gaussian shocks, and $z_{kt} \sim N(0,1)$. As $T$ increases, the shocks approach Gaussianity ($w \approx 0.158, 0.129, 0.100, 0.071$). Proxy relevance follows a Pitman drift: $\Psi_i = c / \sqrt{T}$ with $c = 2.24$ for all instruments, yielding $\Psi \approx (0.158, 0.129, 0.100, 0.071)$ for $T = (200, 300, 500, 1000)$.
\end{minipage}
\end{table}

\section{Construction of identification-robust confidence sets}
\label{sec:emp_arsets_cons}
We construct the confidence sets by inverting an identification-robust test statistic
over structural rotation angles of the target shock, as in the weak-identification-robust SVAR
literature \citep{montiel_olea_inference_2021,drautzburg_refining_2023,hoesch_locally_2024}.

The impulse responses of the target shock depend on $B=H\,Q(\phi)$ only through
its first column,
\begin{equation}
  b_1 \;=\; H\,q_1(\phi_{(1)}),
  \qquad
  q_1(\phi_{(1)}) \;=\; Q(\phi)\,e_1 .
  \label{eq:ar_q1reduction}
\end{equation}
Under the Givens rotations parameterization of Section~\ref{subsec:svarparam}, $q_1$
depends only on the first $n-1$ angles $\phi_{(1)}$; the remaining angles rotate
the non-target shocks among themselves. Following
\citet{montiel_olea_inference_2021}, we invert the test over
$\phi_{(1)}\in(-\pi/2,\pi/2)^{n-1}$ only, on which $q_1$ ranges over the
hemisphere $q_{1,1}>0$. Since $H$ is lower triangular with positive diagonal,
this fixes a positive impact of the shock on its own target variable, matching the normalization of the point
estimate.

For a candidate $q_1$, the target shock is $e_{1,t}(q_1)=q_1'\tilde u_t$ and the
remaining shocks are $e_{\perp,t}(q_1)=N(q_1)'\tilde u_t$, where
$\tilde u_t=H^{-1}u_t$ and the columns of $N(q_1)$ are an orthonormal basis of
the orthogonal complement of $q_1$. The \textit{Proxy-only} criterion uses the
proxy exclusion conditions; \textit{HGMM} adds the higher-order conditions of
Section~\ref{subsec:nongaussianconditions} evaluated at the target shock:
\begin{equation}
  \bar g(q_1)=
  \frac{1}{T}\sum_{t=1}^{T}
  \begin{bmatrix}
    m_t\, e_{\perp,t}(q_1)\\[3pt]
    e_{1,t}^{2}(q_1)\, e_{\perp,t}(q_1)\\[3pt]
    e_{1,t}^{3}(q_1)\, e_{\perp,t}(q_1)\\[3pt]
    e_{1,t}^{2}(q_1)\bigl(\tilde u_t'\tilde u_t-e_{1,t}^{2}(q_1)\bigr)-(n-1)
  \end{bmatrix}.
  \label{eq:ar_moments}
\end{equation}
The blocks are, in order, the proxy exclusion conditions
$\mathbb E[m_t e_{j,t}]=0$, the co-skewness conditions
$\mathbb E[e_{1,t}^{2}e_{j,t}]=0$, the co-kurtosis conditions
$\mathbb E[e_{1,t}^{3}e_{j,t}]=0$, and the
aggregated symmetric co-kurtosis condition
$\mathbb E\bigl[e_{1,t}^{2}\lVert e_{\perp,t}\rVert^{2}\bigr]=n-1$, which uses
$\lVert e_{\perp,t}\rVert^{2}=\tilde u_t'\tilde u_t-e_{1,t}^{2}$. All hold at the
true rotation under Assumption~\ref{as:indep}, unit-variance shocks and proxy
exogeneity. This gives $n-1$ conditions under \textit{Proxy-only} and $3(n-1)+1$
under \textit{HGMM}: $5$ and $16$ in the oil system, $3$ and $10$ in the
monetary system.\footnote{The individual symmetric terms
$\mathbb E[e_{1,t}^{2}e_{j,t}^{2}]-1$ are not used, because a change of basis
$N(q_1)$ mixes them with the cross terms $\mathbb E[e_{1,t}^{2}e_{j,t}e_{k,t}]$;
only their sum is basis-free. As a robustness check, we omit the asymmetric
co-kurtosis conditions, leaving $2(n-1)+1$ conditions under \textit{HGMM}; see
Figure~\ref{fig:kl_arsets_nogrp2}.}

The reduction to $\phi_{(1)}$ is exact. Replacing $N(q_1)$ by $N(q_1)R$, with
$R$ orthogonal, pre-multiplies the first three blocks by $R'$ and leaves the last
unchanged, so $\bar g(q_1)$ and its bootstrap counterparts are transformed by a
common nonsingular matrix $P$, and $\widehat S(q_1)$ by $P(\cdot)P'$. The
statistic \eqref{eq:ar_stat} and its bootstrap distribution are invariant to
such transformations.

The Anderson--Rubin statistic is
\begin{equation}
  \mathrm{AR}(q_1)\;=\;T\,\bar g(q_1)'\,\widehat S(q_1)^{-1}\,\bar g(q_1),
  \label{eq:ar_stat}
\end{equation}
with $\widehat S(q_1)$ the long-run
variance of the moments evaluated at the hypothesized rotation $q_1$. It involves no
estimated Jacobian, so under the null $\sqrt{T}\,\bar g(q_1)$ is asymptotically
normal regardless of identification strength and \eqref{eq:ar_stat} is
asymptotically pivotal.

We depart from the exact implementation in two respects. First,
$\widehat S(q_1)$ is estimated by the residual-based moving-block bootstrap of
Section~\ref{sec:bootstrap} \citep{jentsch_asymptotically_2022}, which resamples
the recentered residuals and the proxy jointly in blocks of length $\ell$ and
re-estimates the VAR in each draw. The estimate therefore accounts for
estimation of $(\hat{\mathcal A},H)$, which \citet{montiel_olea_inference_2021}
handle by the delta method. With
$Z_b^{*}(q_1)=\sqrt{T}\bigl(\bar g_b^{*}(q_1)-\bar g(q_1)\bigr)$ the bootstrap
moments recentered at the sample moments \citep{hall_bootstrap_1996},
\begin{equation}
  \widehat S(q_1)=\frac{1}{B-1}\sum_{b=1}^{B}Z_b^{*}(q_1)\,Z_b^{*}(q_1)',
  \qquad
  \mathrm{AR}_b^{*}(q_1)=Z_b^{*}(q_1)'\,\widehat S(q_1)^{-1}\,Z_b^{*}(q_1),
  \label{eq:ar_boot}
\end{equation}
where the first expression is a bootstrap variance estimator
\citep{goncalves_bootstrapping_2005}. 
Second, the critical value $q_{\alpha}(q_1)$ is
the $(1-\alpha)$ quantile of $\{\mathrm{AR}_b^{*}(q_1)\}_{b=1}^{B}$, computed
separately at each candidate, as in the grid bootstrap
\citep{hansen_grid_1999,mikusheva_uniform_2007}. Bootstrap critical values for
Anderson--Rubin-type statistics remain valid under weak instruments
\citep{moreira_bootstrap_2009,davidson_bootstrap_2008}, and grid moving-block
bootstrap AR sets are used in proxy SVARs by \citet{angelini_identification_2024}.
The confidence set is
\begin{equation}
  \mathcal C_{1-\alpha}
  =\Bigl\{\,\phi_{(1)}\in(-\tfrac{\pi}{2},\tfrac{\pi}{2})^{\,n-1}
     \;:\; \mathrm{AR}\bigl(q_1(\phi_{(1)})\bigr)\le q_{\alpha}\bigl(q_1(\phi_{(1)})\bigr)\Bigr\}.
  \label{eq:ar_accept}
\end{equation}
Because the critical value is computed at each $q_1$ and is typically larger
when more conditions enter, the \textit{HGMM} and \textit{Proxy-only} sets need
not be nested: moment conditions that are nearly uninformative in some directions raise
the critical value more than the statistic \citep{kleibergen_pivotal_2002}.

The response of variable $v$ at horizon $h$ is linear in $q_1$,
\begin{equation}
  \mathrm{IRF}_{v}(h;q_1)=e_v'\,\bigl(R\,\hat{\mathcal A}^{\,h}R'\bigr)H\,q_1
  \;=\; c_{v,h}'\,q_1,
  \label{eq:ar_projection}
\end{equation}
and its confidence interval is the projection of $\mathcal C_{1-\alpha}$,
$\bigl[\min_{\mathcal C_{1-\alpha}} c_{v,h}'q_1,\;\max_{\mathcal C_{1-\alpha}} c_{v,h}'q_1\bigr]$, as in
\citet{dufour_exact_1990,dufour_projection-based_2005,granziera_inference_2018}.
Since $\mathcal C_{1-\alpha}$ need not be convex or connected
\citep{drautzburg_refining_2023} and a full grid is infeasible, we compute each
endpoint by constrained optimization \citep{kaido_confidence_2019}. We evaluate
\eqref{eq:ar_accept} at $10^{5}$ rotations, half drawn uniformly over the box
and half concentrated around the point estimate
\citep{chernozhukov_estimation_2007}, and then optimize
\eqref{eq:ar_projection} subject to $\mathrm{AR}(q_1)\le q_{\alpha}(q_1)$ from the
most extreme accepted points.

The exact projection is conservative by construction. The computed interval is
an inner approximation to it, and the endpoints are stable
when the number of draws is increased. Because the shock has unit variance,
$\lvert c_{v,h}'q_1\rvert\le\lVert c_{v,h}\rVert$, so every interval is bounded.
An interval that reaches the boundary of the box, where $q_{1,1}\to 0$, cannot
exclude a zero impact on the target variable; we report it as reaching the
normalization boundary.

\section{Additional Results}
\label{appendix:extra_irfs}

This appendix collects further results for the two main illustrations. For each
we report some evidence regarding the non-Gaussianity of the reduced-form innovations and the estimated structural shocks. We also provide impulse responses of the target shock estimated by the hybrid GMM
criterion, with moving-block-bootstrap percentile confidence intervals, alongside
the single-proxy SVAR benchmark; both estimators are computed on the same
bootstrap replications, so the two comparisons differ only in the identification step.
The point estimates and their percentile bands lead to the same economic reading
as the identification-robust sets in the main text. 

\subsection{Evidence for non-Gaussianity in the structural shocks}\label{appendix:evidence_nongauss}
Panel A of Table~\ref{tab:resid_nongauss} reports the third and fourth moments
of the reduced-form innovations and of the structural shocks recovered at the
estimate. Identification requires Assumption~\ref{as:gau}: at most one
shock may have zero skewness, and at least $n-1$ shocks must have nonzero excess
kurtosis of a common sign. Both hold at the estimates, the second exactly. Every
skewness is nonzero so
part (a) holds with the oil news shock as the single admissible symmetric
component. Five of the six excess kurtoses are positive and one, the world IP
shock, is negative, so part (b) holds with exactly $n-1$ shocks sharing a sign. The weaker
condition of \citet{comon_independent_1994}, that at most one shock have both
cumulants equal to zero holds comfortably. Taken as a system, joint Gaussianity of the shock vector is rejected
decisively (pooled $JB=35.14$ on $12$ degrees of freedom, $p=0.0004$).

Matching a Student-$t$ to the estimated kurtoses gives us
$\hat\nu\geq 9.9$ for every shock, so the sixth- and eighth-order moments that
govern the asymptotic variance of the fourth-order conditions are finite. This allows the asymmetric co-kurtosis conditions to be appended in
Section~\ref{subsubsec:oil_spec} for oil-shock illustration; they are an addition to the identifying block,
not a necessary requirement for it.
\begin{table}[htbp]
\centering
\caption{Evidence of non-Gaussianity in the reduced-form innovations and estimated structural shocks}
\label{tab:resid_nongauss}
\renewcommand{\arraystretch}{1.1}
\setlength{\tabcolsep}{3pt}
\resizebox{0.60\textwidth}{!}{%
\begin{tabular}{llrrrr}
\toprule
& & \multicolumn{2}{c}{Reduced-form $\hat u_t$} & \multicolumn{2}{c}{Structural $\hat\varepsilon_t(\hat\phi)$} \\
\cmidrule(lr){3-4}\cmidrule(lr){5-6}
& Variable / shock & Skew & Ex.\ kurt & Skew & Ex.\ kurt \\
\midrule
\multicolumn{6}{l}{\textit{Panel A: oil news shock, Kilian-corrected specification}} \\
\midrule
$1$ & Real oil price / \textbf{target} & $-0.084$ & $+0.043$ & $-0.042$ & $+0.004$ \\
$2$ & World oil production  & $-0.014$ & $-0.139$ & $+0.150$ & $+0.102$ \\
$3$ & World oil inventories & $-0.146$ & $+0.003$ & $-0.246$ & $+0.016$ \\
$4$ & World IP              & $+0.059$ & $+0.311$ & $-0.083$ & $-0.289$ \\
$5$ & U.S. IP               & $-0.295$ & $+1.862$ & $-0.308$ & $+1.021$ \\
$6$ & U.S. CPI              & $+0.165$ & $+1.141$ & $+0.035$ & $+0.860$ \\
\addlinespace
\multicolumn{6}{l}{\quad Pooled $JB=35.14$ on $12$ d.f., $p=0.0004$} \\

\midrule
\multicolumn{6}{l}{\textit{Panel B: monetary policy shock, Altavilla specification}} \\
\midrule
$1$ & 2Y OIS / \textbf{target}  & $-0.179$ & $+1.883$ & $0.105$ & $+1.258$ \\
$2$ & Euro Stoxx 50 (log)  & $0.241$ & $0.657$ & $-0.433$ & $+0.733$ \\
$3$ & EUR/USD (log)        & $0.252$ & $0.533$ & $+0.424$ & $+1.031$ \\
$4$ & 2Y infl.-linked swap & $0.038$ & $12.146$ & $-0.334$ & $+10.307$ \\
\addlinespace
\multicolumn{6}{l}{\quad Pooled $JB=2572.49$ on $8$ d.f., $p<10^{-16}$} \\

\bottomrule
\end{tabular}}
\begin{minipage}{0.96\textwidth}
\vspace{4pt}
\footnotesize
\textit{Notes:} Sample skewness and excess kurtosis of the reduced-form residuals $\hat u_t$ of \eqref{eq:VAR_emp} and of the structural shocks $\hat\varepsilon_t(\hat\phi)=Q(\hat\phi)'H^{-1}\hat u_t$ recovered at the HGMM estimate. Identification requires Assumption~\ref{as:gau}: at most one shock with zero skewness, and at least $n-1$ shocks with nonzero excess kurtosis of a common sign. The pooled statistic, the sum of the marginal $JB$ statistics, is distributed $\chi^2_{2n}$ under joint Gaussianity of the shock vector and rejects decisively in both applications.
\end{minipage}
\end{table}

Panel B of Table~\ref{tab:resid_nongauss} reports the same statistics for the
daily financial VAR. The normality is rejected for all four structural shocks, and on the raw
statistics the euro-area shocks are further from Gaussian than the oil-market
shocks of Panel A. The identifying gain reported in
Section~\ref{subsubsec:mp_results} is nonetheless the smaller of the two since
the additional moment block (asymmetric co-kurtosis) that the oil application appends is not available
here. Under the factorization of Section~\ref{subsec:nongaussianconditions},
$\mathbb{E}\|g_t\|^{2+\delta}<\infty$ needs $4+2\delta$ marginal moments for
$\mathbb{E}[\varepsilon_i^2\varepsilon_j^2]=1$ but $6+3\delta$ for
$\mathbb{E}[\varepsilon_i^3\varepsilon_j]=0$; without independence both bounds
double and the ordering is unchanged. An estimated structural shock, with excess
kurtosis of over $10$ and hence $\hat\nu \approx 4.5$, clears the first threshold and fails the
second: the identifying conditions remain admissible, but the asymmetric co-kurtosis
conditions cannot be appended as their variance is
undefined\footnote{The symmetric co-kurtosis conditions clear
$4+2\delta$ but not the $8+4\delta$ bound that holds without the factorization, so
the euro-area results may maintain weak tail dependence across shocks.}. The euro-area
criterion therefore rests on the identifying conditions of co-skewness and symmetric co-kurtosis alone.

\subsection{Oil news shock (Kilian (2024)-corrected specification)}\label{appendix:extra_kilian}
Table \ref{tab:arsign_kilian} records the horizons at which each criterion
determines the sign of a response.
Figure~\ref{fig:kl_mbb} reports the HGMM and proxy-SVAR responses with MBB
percentile bands. Figure~\ref{fig:kl_arsets_nogrp2} shows the
identification-robust sets when the asymmetric co-kurtosis conditions
$\mathbb{E}[\varepsilon_i^3\varepsilon_j]$ are omitted from the non-Gaussian
block, for both estimation and inference. The ranking of the two criteria is
unaffected, the \textit{HGMM} set remains the shorter one at $277$ of the
$306$ (variable, horizon) pairs, but the magnitude of the gain is smaller. The
median reduction in length falls from $71.6\%$ to $14.8\%$. The asymmetric block therefore carries a
large share of the identifying information in this application, which is consistent with the moment-existence argument of
Appendix~\ref{appendix:evidence_nongauss}: the oil-market shocks show finite
eighth moments and can support fourth-order conditions, whereas the euro-area 
shocks cannot. Figure~\ref{fig:kl_mbb_vs_ar_newboot} repeats the main-text comparison at the
shorter bootstrap block length $\ell=25$.

\begin{table}[htbp]
\centering
\caption{Horizons at which the sign of the response is resolved: \textit{Proxy-only} vs \textit{HGMM} (Kilian-corrected specification)}
\label{tab:arsign_kilian}
\renewcommand{\arraystretch}{1.15}
\resizebox{0.98\textwidth}{!}{%
\begin{tabular}{lcccc}
\toprule
 & \multicolumn{2}{c}{Horizons with $0\notin$ CS} & \multicolumn{2}{c}{Sign-resolved horizon ranges} \\
\cmidrule(lr){2-3}\cmidrule(lr){4-5}
Variable & Proxy & HGMM & Proxy & HGMM \\
\midrule
Real oil price & 9 & \textbf{41} & $[0,8]$ ($+$) & $[0,14]$, $[21,22]$ ($+$), $[26,30]$ ($-$), $[32,50]$ ($+$) \\
World oil production & 0 & \textbf{23} & -- & $[22,23]$ ($-$), $[29,31]$ ($+$), $[32,36]$ ($-$), $[41,50]$ ($+$) \\
World oil inventories & 23 & \textbf{37} & $[25,46]$ ($+$) & $[0,5]$, $[8,12]$ ($-$), $[26,50]$ ($+$) \\
World IP & 1 & \textbf{31} & -- & $[0,11]$, $[19,25]$, $[39,50]$ ($+$) \\
U.S. IP & 0 & \textbf{31} & -- & $[1,10]$ ($+$), $[31,50]$ ($-$) \\
U.S. CPI & 26 & \textbf{49} & $[2,13]$, $[36,50]$ ($+$) & $[0,29]$, $[32,50]$ ($+$) \\
\bottomrule
\end{tabular}}
\begin{minipage}{0.96\textwidth}
\vspace{4pt}
\footnotesize
\textit{Notes:} Columns 2--3 count the horizons $h\in\{0,\dots,50\}$ (out of $51$) at which the 90\% identification-robust confidence set excludes zero, so that the sign of the response is determined. Columns 4--5 report every maximal block of consecutive such horizons, in horizon order, with the sign of the resolved response as a superscript. Out of the $306$ (variable, horizon) pairs, \textit{HGMM} resolves a sign that \textit{Proxy-only} leaves open at $154$ pairs, while the reverse occurs at a single pair. Under \textit{Proxy-only} the sign of the world oil production and U.S. industrial production responses is undetermined at every horizon, and the world oil inventory response is never resolved as negative. Construction as in Table~\ref{tab:arlen_kilian}.
\end{minipage}
\end{table}
\begin{figure}[H]
    \centering
    \includegraphics[width=0.95\textwidth]{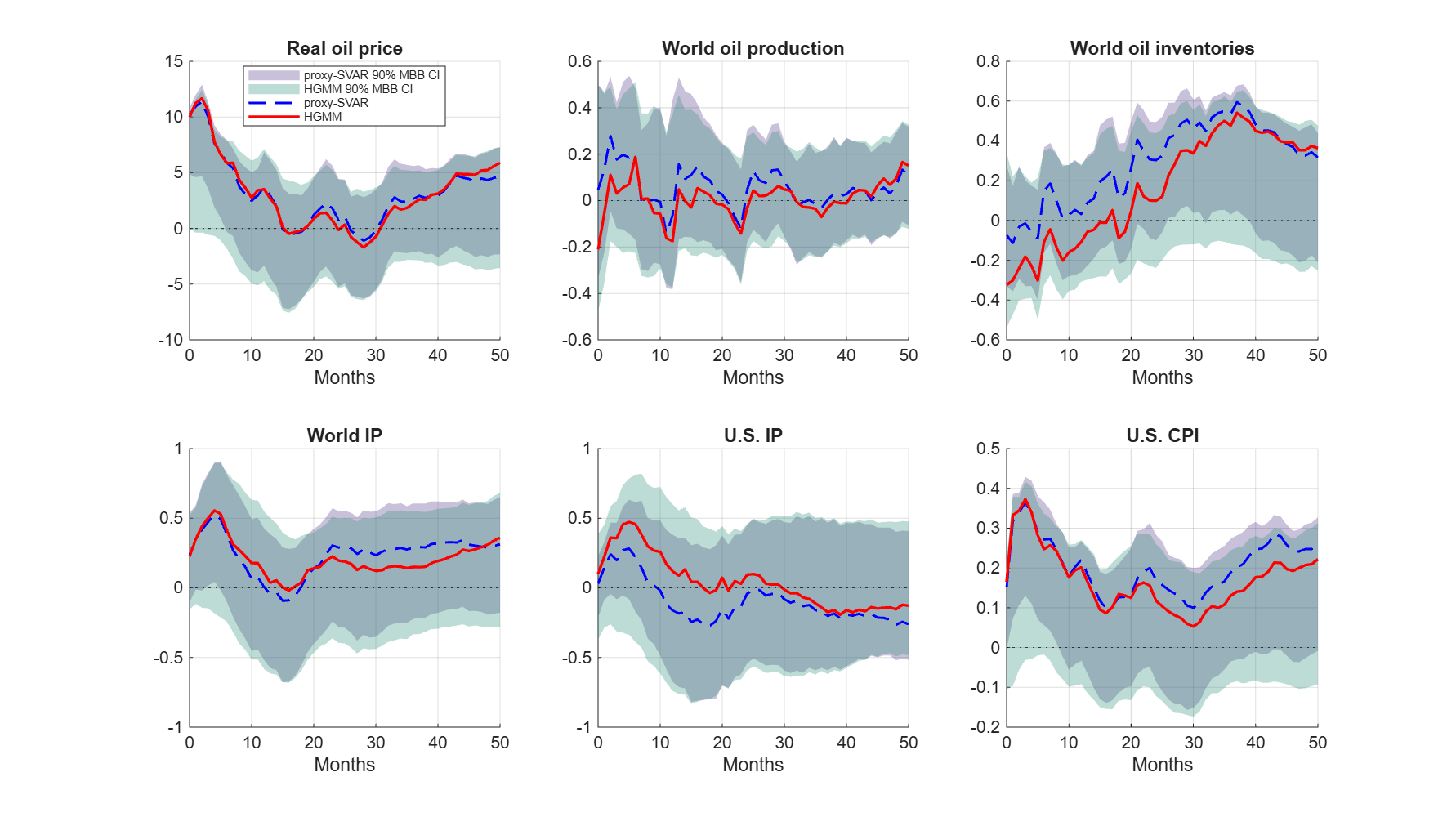}
    \caption{Impulse responses to the oil news shock: HGMM vs proxy-SVAR, with moving-block-bootstrap percentile confidence intervals (Kilian-corrected specification).}
    \label{fig:kl_mbb}
    \begin{minipage}{0.92\textwidth}
\vspace{2pt}
\scriptsize
\textit{Notes:} Point impulse responses of the oil news shock estimated by the hybrid GMM criterion (\textcolor{red}{HGMM}, solid red) and by the single-proxy SVAR (\textcolor{blue}{proxy-SVAR}, dashed blue), each with $90\%$ residual-based moving-block-bootstrap percentile confidence intervals (shaded; HGMM in green, proxy-SVAR in purple). Both estimators are computed on every bootstrap draw from the {same} resampled residuals, so the two comparisons differ only in the identification step. Bootstrap: $9999$ replications, block length $\ell=36$, VAR re-estimated within each draw. Responses are normalized to a $10\%$ increase in the real oil price on impact; the $x$-axis is the horizon in months, the $y$-axis the response in percent.
    \end{minipage}
\end{figure}

\begin{figure}[H]
    \centering
    \includegraphics[width=0.95\textwidth]{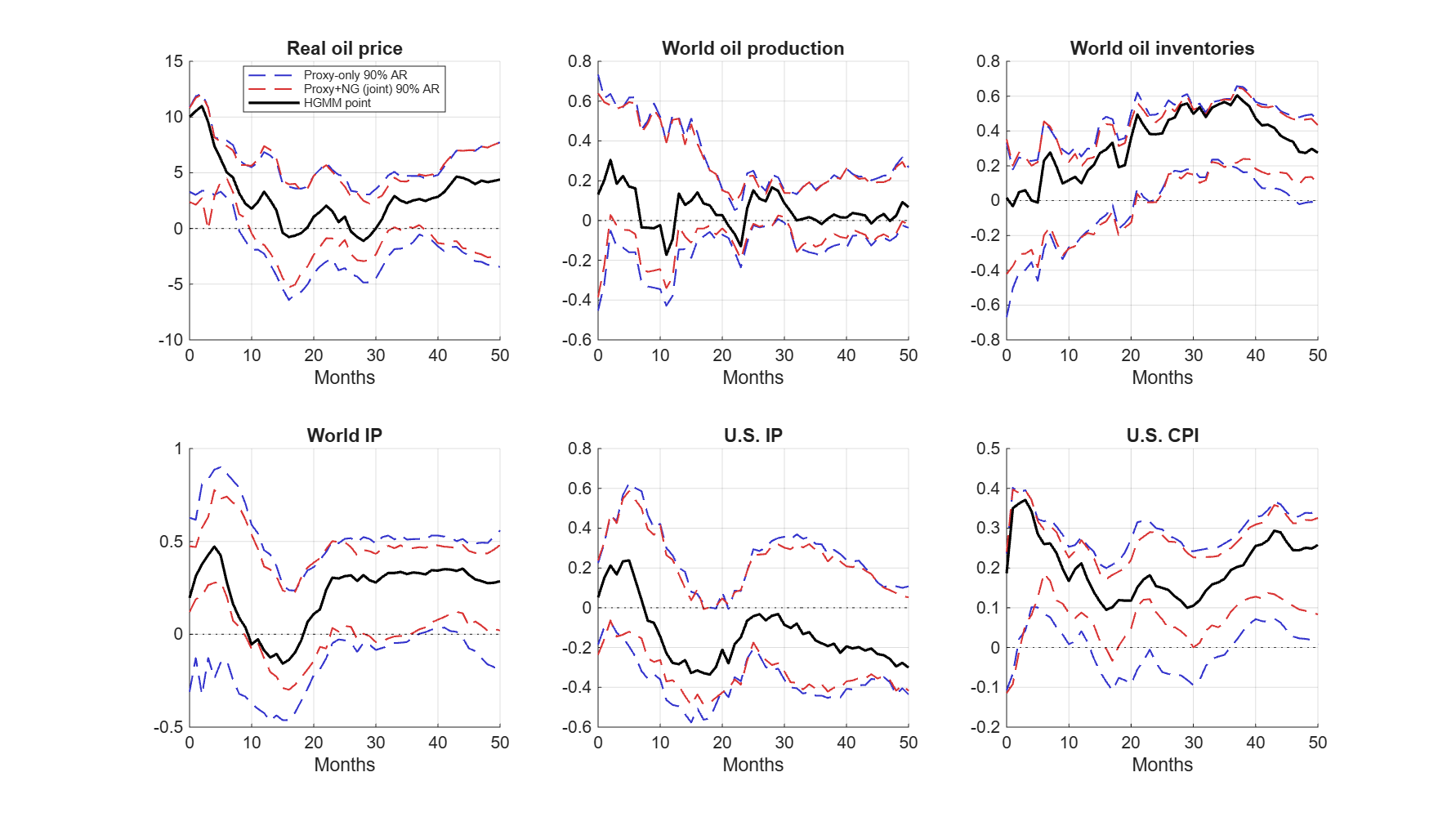}
    \caption{90\% Anderson--Rubin confidence sets for the oil news shock: \textit{Proxy-only} vs \textit{HGMM} (without asymmetric co-kurtosis conditions).}
    \label{fig:kl_arsets_nogrp2}
    \begin{minipage}{0.92\textwidth}
\vspace{2pt}
\scriptsize
\textit{Notes:} Each panel plots the 90\% identification-robust Anderson--Rubin (AR) confidence set for the impulse response of one variable to the oil news shock, at horizons $0$--$50$ months. \textcolor{blue}{\textit{Proxy-only}} (blue, dashed) inverts the AR statistic using only the five proxy-orthogonality conditions that identify the target shock; \textcolor{red}{\textit{HGMM}} (red, dashed) additionally imposes the co-skewness and (only) symmetric co-kurtosis-aggregate conditions of the target shock ($11$ conditions in total). The solid black line is the HGMM point estimate. Critical values from a residual-based moving-block bootstrap ($9999$ replications, block length $\ell = 36$). Responses are normalized to a $10\%$ increase in the real oil price on impact. The $x$-axis is the horizon in months; the $y$-axis is the response in percent. 
    \end{minipage}
\end{figure}

\begin{figure}[H]
    \centering
     \includegraphics[width=0.98\textwidth]{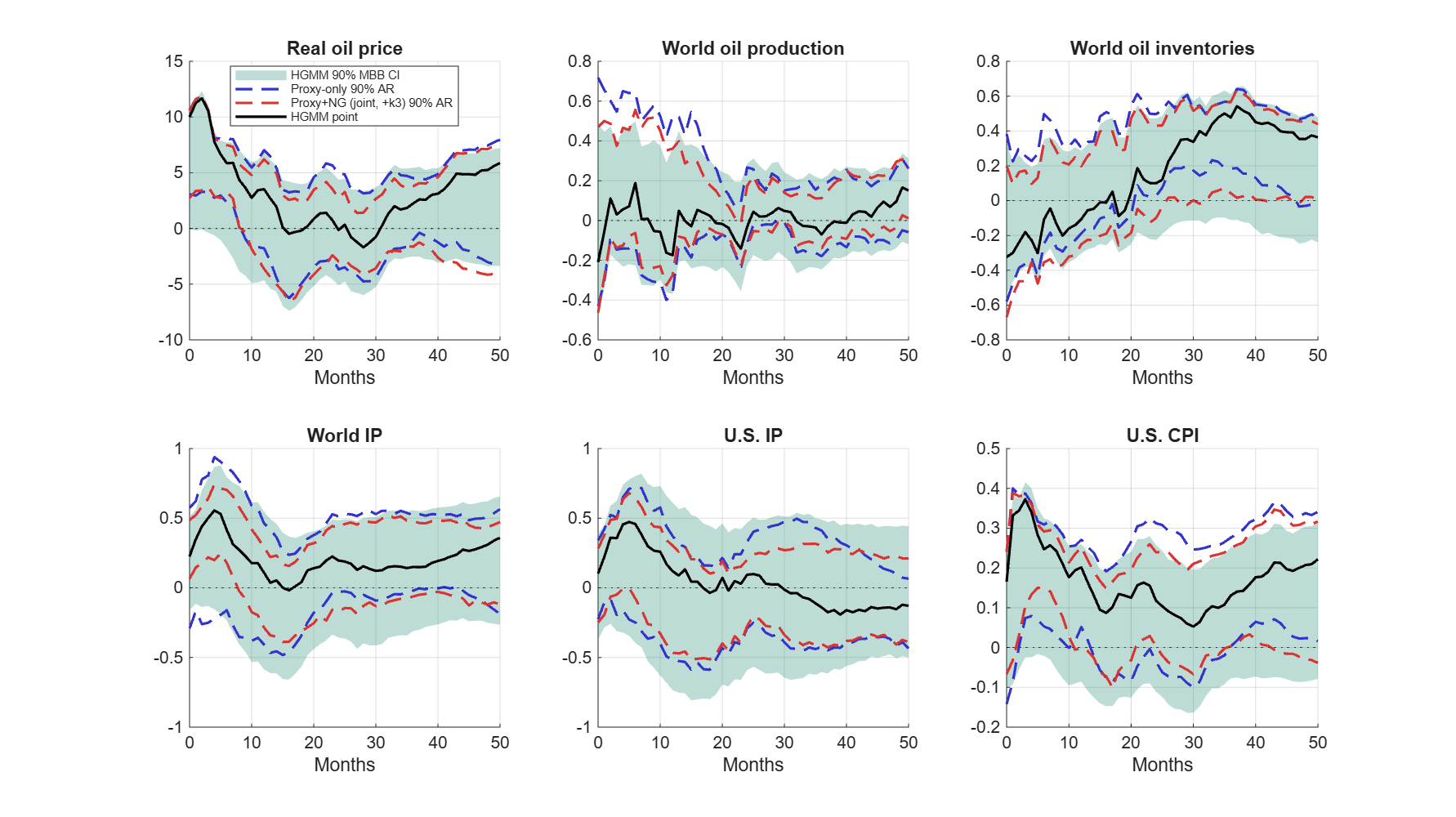}
    \caption{90\% identification-robust confidence sets with \textit{Proxy-only} and \textit{HGMM} criterions, with standard MBB percentile bands (block length $\ell = 25$)}
    \label{fig:kl_mbb_vs_ar_newboot}
    \begin{minipage}{0.92\textwidth}
\vspace{2pt}
\scriptsize
\textit{Notes:} Each panel superimposes three $90\%$ interval estimates for the impulse response of one variable to the oil news shock. The shaded band is the residual-based moving-block-bootstrap (MBB) percentile confidence interval for the HGMM {point} estimator. The two dashed lines are the identification-robust AR sets: \textcolor{blue}{\textit{Proxy-only}} (blue, dashed) inverts the AR statistic using only the five proxy orthogonality conditions that identify the target shock; \textcolor{red}{\textit{HGMM}} (red, dashed) additionally imposes the higher-order moment conditions of the target shock ($16$ conditions in total). The solid black line is the HGMM point estimate. Critical values from a residual-based moving-block bootstrap ($9999$ replications, block length $\ell = 25$). Responses are normalized to a $10\%$ increase in the real oil price on impact. The $x$-axis is the horizon in months; the $y$-axis is the response in percent.
    \end{minipage}
\end{figure}

\subsection{Monetary policy shock (Altavilla (2019) specification)}
\label{appendix:extra_altavilla}
Table \ref{tab:arsign_altavilla} reports the horizons $h\in\{0,\dots,200\}$ (out of $201$) at which the 90\% identification-robust confidence sets excludes zero, so that the sign of the response is determined. Under
\textit{Proxy-only}, the exchange rate response is not signed at any horizon, while under \textit{HGMM} it
is signed over $[103,200]$. The delayed equity appreciation is resolved from
$h=67$ under \textit{HGMM} but only from $h=119$ under \textit{Proxy-only}, a
difference of roughly two months, and the inflation-linked swap is signed over
$[76,126]$ against $[94,121]$. The policy rate is the exception: both criteria
sign it only over the first six weeks, and \textit{Proxy-only} does so at four
more horizons than \textit{HGMM}.

Figure~\ref{fig:al_mbb} reports the HGMM and proxy-SVAR responses with
percentile bands. Consistent with the identification-robust sets in the main
text, the hybrid and proxy-SVAR point estimates are close, with modest sharpening of confidence sets.

Figure \ref{fig:al_mbb_vs_ar_newboot} reports the AR sets with standard MBB CIs, for alternate bootstrap block length of $35$. The shorter block widens both criteria's sets and compresses the gain, from a median reduction of $26.4\%$ to $10.2\%$, but the median reduction stays positive for all four variables.

\begin{table}[htbp]
\centering
\caption{Horizons at which the sign of the response is resolved: \textit{Proxy-only} vs \textit{HGMM} (Altavilla monetary-policy specification)}
\label{tab:arsign_altavilla}
\renewcommand{\arraystretch}{1.15}
\resizebox{0.98\textwidth}{!}{%
\begin{tabular}{lcccc}
\toprule
 & \multicolumn{2}{c}{Horizons with $0\notin$ CS} & \multicolumn{2}{c}{Sign-resolved horizon ranges} \\
\cmidrule(lr){2-3}\cmidrule(lr){4-5}
Variable & Proxy & HGMM & Proxy & HGMM \\
\midrule
2Y OIS & \textbf{33} & 29 & $[0,32]$ ($+$) & $[0,28]$ ($+$) \\
Euro Stoxx 50 & 31 & \textbf{84} & $[119,150]$ ($+$) & $[67,150]$ ($+$) \\
EUR/USD & 0 & \textbf{98} & -- & $[103,200]$ ($+$) \\
2Y infl.-linked swap & 28 & \textbf{88} & $[94,121]$ ($-$) & $[76,126]$ ($-$), $[165,200]$ ($+$) \\
\bottomrule
\end{tabular}
}
\begin{minipage}{0.96\textwidth}
\vspace{4pt}
\footnotesize
\textit{Notes:} Columns 2--3 count the horizons $h\in\{0,\dots,200\}$ (out of $201$) at which the 90\% identification-robust confidence set excludes zero, so that the sign of the response is determined. Columns 4--5 report the corresponding horizon ranges, with the sign of the resolved response in parentheses; isolated horizons within a reported block are suppressed. Out of the $804$ (variable, horizon) pairs, \textit{HGMM} resolves a sign that \textit{Proxy-only} leaves open at $211$ pairs, while the reverse occurs at $4$ pairs, all for the two-year OIS rate. Construction as in Table~\ref{tab:arlen_altavilla}.
\end{minipage}
\end{table}

\begin{figure}[H]
    \centering
    \includegraphics[width=0.75\textwidth]{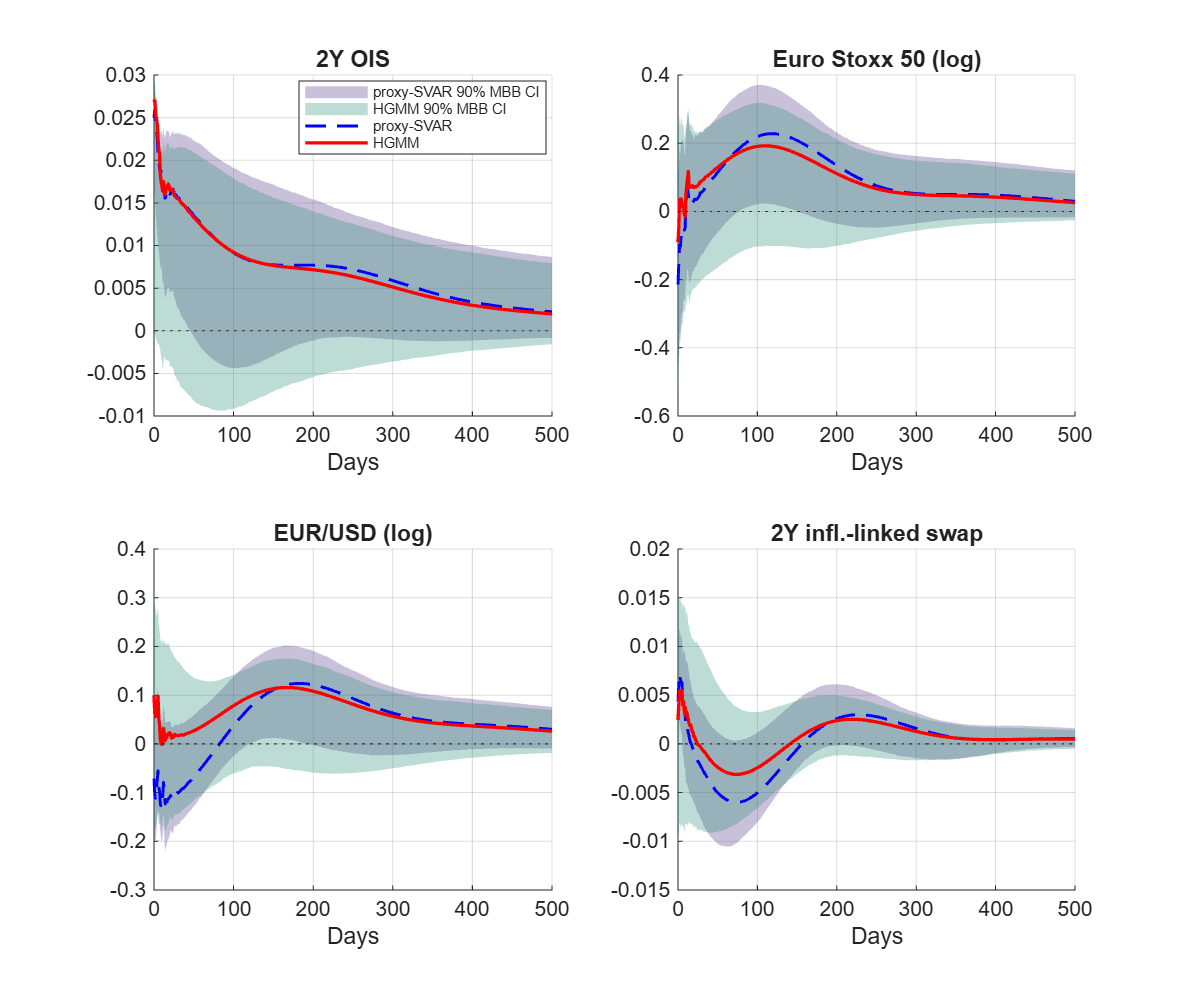}
    \caption{Impulse responses to the euro-area monetary policy shock: HGMM vs proxy-SVAR, with moving-block-bootstrap percentile confidence intervals.}
    \label{fig:al_mbb}
    \begin{minipage}{0.92\textwidth}
\vspace{2pt}
\scriptsize
\textit{Notes:} Point impulse responses of the monetary policy shock estimated by the hybrid GMM criterion (\textcolor{red}{HGMM}, solid red) and by the single-proxy SVAR (\textcolor{blue}{proxy-SVAR}, dashed blue), each with $90\%$ residual-based moving-block-bootstrap percentile confidence intervals (shaded; HGMM in green, proxy-SVAR in purple). Both estimators are computed on every bootstrap draw from the {same} resampled residuals, so the two comparisons differ only in the identification step. Bootstrap: $9999$ replications, block length $\ell=50$, VAR re-estimated within each draw. Responses are to a one-standard-deviation monetary policy shock; the $x$-axis is the horizon in business days.
    \end{minipage}
\end{figure}

\begin{figure}[H]
    \centering
    \includegraphics[width=0.75\textwidth]{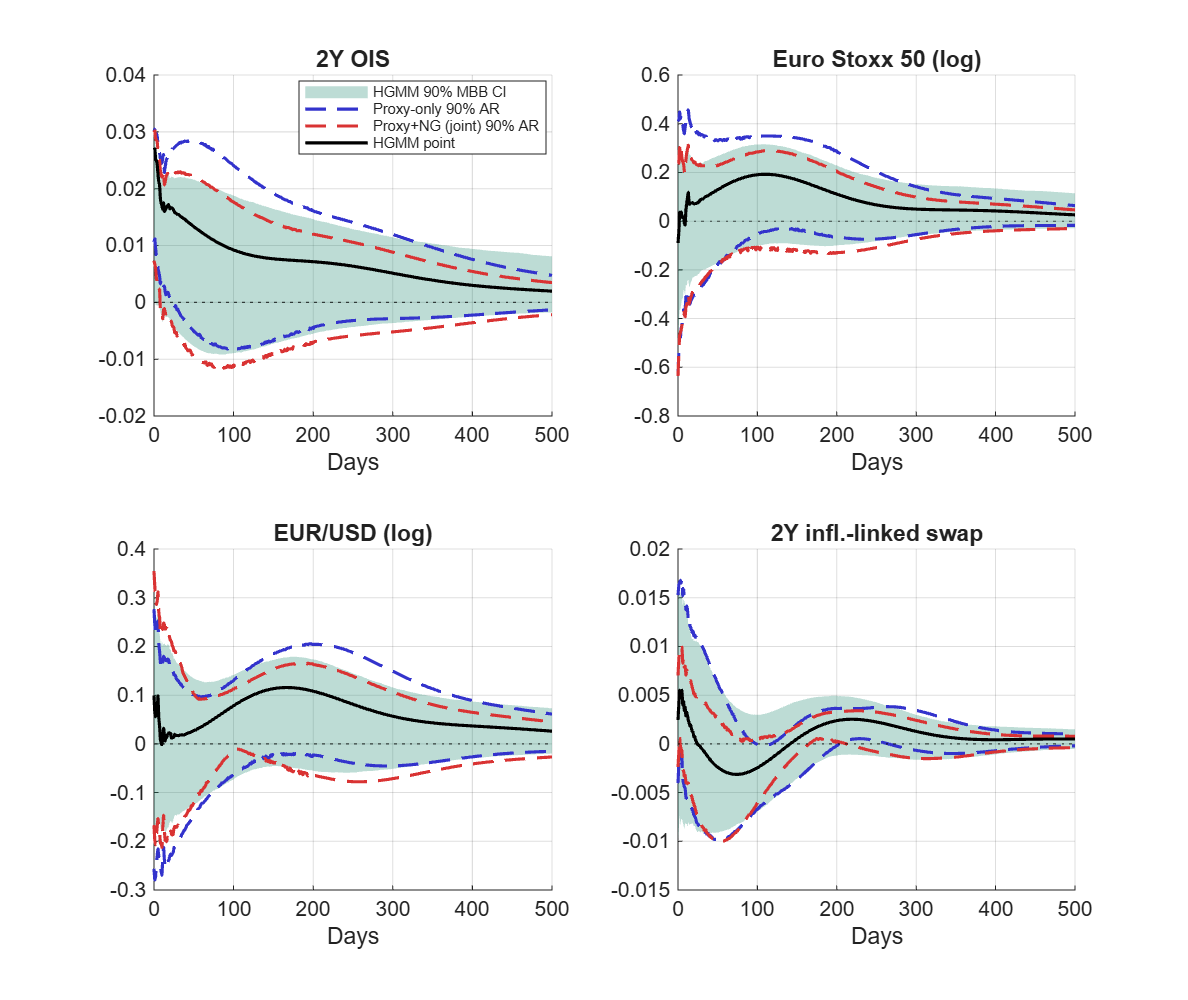}
    \caption{90\% identification-robust confidence sets with \textit{Proxy-only} and \textit{HGMM} criteria, with standard MBB percentile bands (block length $\ell = 35$).}
    \label{fig:al_mbb_vs_ar_newboot}
    \begin{minipage}{0.92\textwidth}
\vspace{2pt}
\scriptsize
\textit{Notes:} Each panel superimposes three $90\%$ interval estimates for the impulse response of one variable to the monetary policy shock. The shaded band is the residual-based moving-block-bootstrap (MBB) percentile confidence interval for the HGMM {point} estimator. The two dashed lines are the identification-robust AR sets: \textcolor{blue}{\textit{Proxy-only}} (blue, dashed) inverts the AR statistic using only the three proxy-exclusion conditions that identify the target shock; \textcolor{red}{\textit{Proxy+NG}} (red, dashed) is the HGMM criterion which additionally imposes the three co-skewness conditions of the target shock and the fourth-order aggregate, for $7$ conditions in total. Critical values are from a residual-based moving-block bootstrap ($9999$ bootstrap replications, block length $\ell = 35$), computed per candidate rotation. The solid black line is the HGMM point estimate. Responses are to a one standard deviation contractionary monetary policy shock; the $x$-axis is the horizon in business days.
    \end{minipage}
\end{figure}

\section{Robustness: the K\"anzig baseline (oil news shock)}
\label{appendix:kanzig_rob}

As a robustness check on the oil illustration, we re-run the analysis on the
\citet{kanzig_macroeconomic_2021} baseline: the full 1974:M1--2017:M12 sample,
$p=12$ lags, and the original OPEC-window instrument (without the Kilian
corrections). Figure~\ref{fig:kz_arsets} and Table~\ref{tab:arlen_kanzig} show
the same direction of effect, more modestly: across the $306$ (variable, horizon) pairs the \textit{HGMM} set is strictly shorter at $246$, of equal length at $60$, and longer at none; the median reduction in length is $2.3\%$ and the largest is $27.0\%$. The sharpening is concentrated at short horizons for world oil production, and at long horizons for U.S. CPI, where the median reduction across horizons is $15.6\%$. Neither criterion determines the sign of any response at any horizon in this
specification, against $212$ sign-resolved horizons under the Kilian-corrected
instrument.
 This is consistent with the theory that the payoff of sharper inference
depends on the marginal identifying content of those conditions offsetting the
cost of the additional degrees of freedom. The proxy exclusion restriction is not rejected
($\hat{J}_{\text{prx}}=1.90$, with asymptotic and bootstrap $p$-values of $0.86$
and $0.99$).

\begin{figure}[H]
    \centering
    \includegraphics[width=0.95\textwidth]{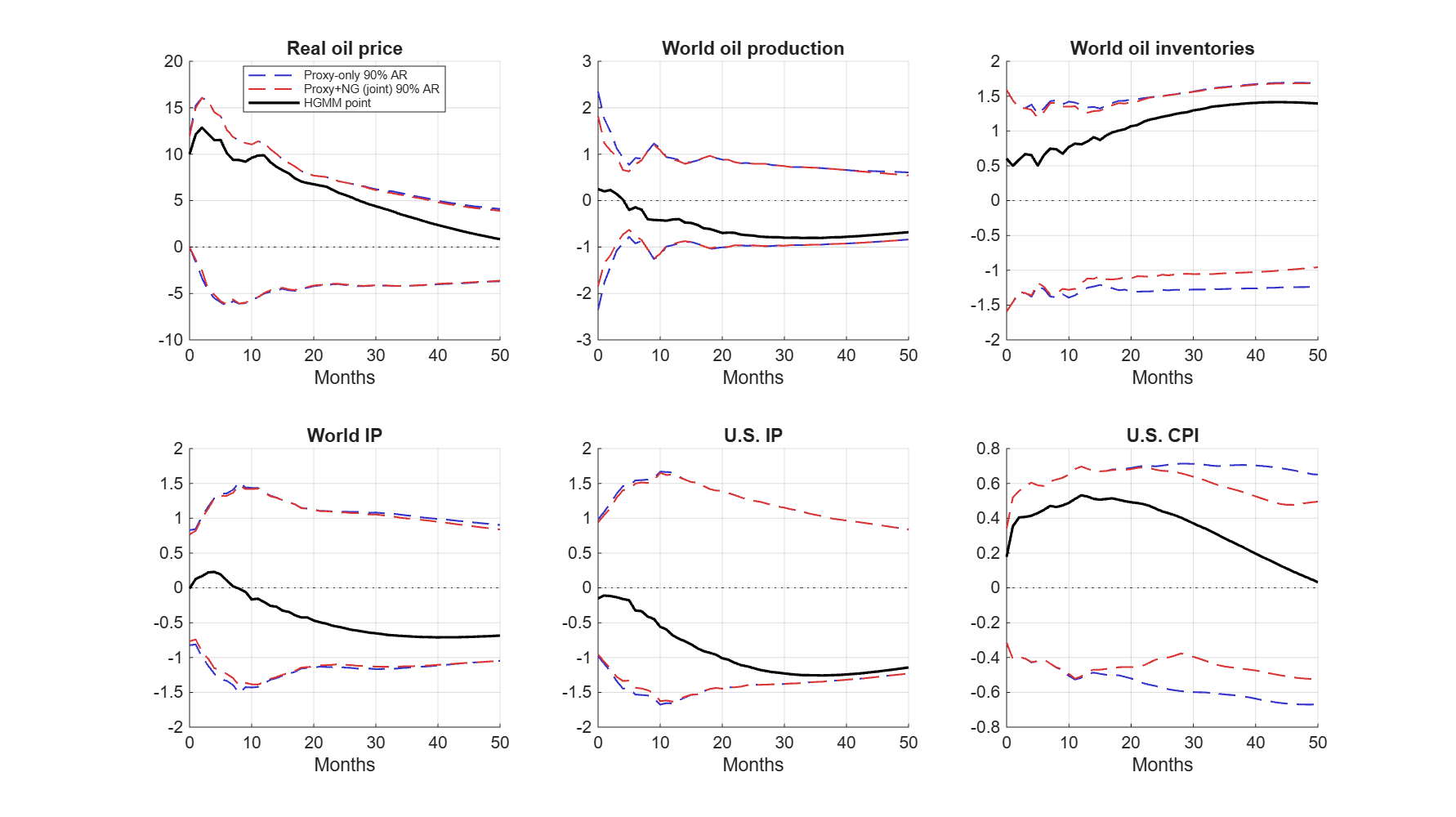}
    \caption{90\% Anderson--Rubin confidence sets for the oil news shock: \textit{Proxy-only} vs \textit{HGMM} (K\"anzig baseline specification).}
    \label{fig:kz_arsets}
    \begin{minipage}{0.92\textwidth}
\vspace{2pt}
\scriptsize
\textit{Notes:} Each panel plots the 90\% identification-robust Anderson--Rubin (AR) confidence set for the impulse response of one variable to the oil news shock, at horizons $0$--$50$ months for the K\"anzig (2021) baseline specification (VAR sample 1974:M1--2017:M12, $p=12$ lags, monthly OPEC-window surprise summed within the month). \textcolor{blue}{\textit{Proxy-only}} (blue, dashed) uses the five proxy-exogeneity conditions; \textcolor{red}{\textit{HGMM}} (red, dashed) additionally imposes the target shock's co-skewness conditions and the co-kurtosis aggregate ($11$ conditions in total). The solid black line is the HGMM point estimate. Critical values are from a residual-based moving-block bootstrap ($9999$ replications, block length $\ell = 36$); responses are normalized to a $10\%$ increase in the real oil price on impact.
    \end{minipage}
\end{figure}

\begin{figure}[H]
    \centering
    \includegraphics[width=0.95\textwidth]{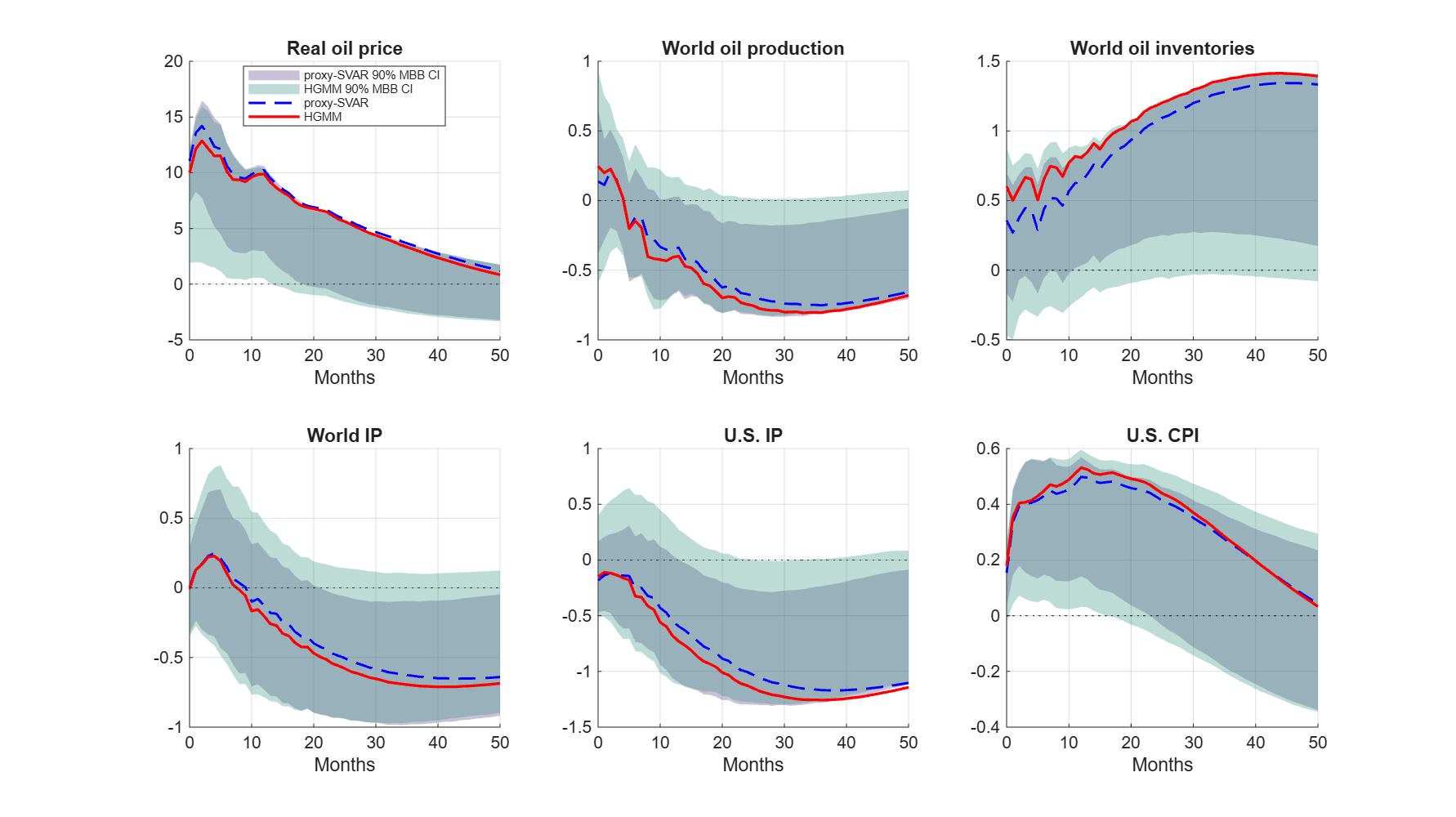}
    \caption{Impulse responses to the oil news shock: HGMM vs proxy-SVAR, with moving-block-bootstrap percentile confidence intervals (K\"anzig baseline specification).}
    \label{fig:kz_mbb}
    \begin{minipage}{0.92\textwidth}
\vspace{2pt}
\scriptsize
\textit{Notes:} As in Figure~\ref{fig:kl_mbb}, for the K\"anzig (2021) baseline specification. \textcolor{red}{HGMM} (solid red) and \textcolor{blue}{proxy-SVAR} (dashed blue) with $90\%$ moving-block-bootstrap percentile confidence intervals (shaded), computed from the same resampled residuals ($9999$ replications, block length $\ell=36$, VAR re-estimated per draw). Responses are normalized to a $10\%$ increase in the real oil price on impact.
    \end{minipage}
\end{figure}

\begin{table}[htbp]
\centering
\caption{Length comparison of 90\% Anderson--Rubin confidence sets: \textit{Proxy-only} vs \textit{HGMM} (K\"anzig baseline specification)}
\label{tab:arlen_kanzig}
\renewcommand{\arraystretch}{1.15}
\resizebox{\textwidth}{!}{%
\begin{tabular}{lcccccccccccc}
\toprule
 & \multicolumn{2}{c}{$h=0$} & \multicolumn{2}{c}{$h=1$} & \multicolumn{2}{c}{$h=2$} & \multicolumn{2}{c}{$h=3$} & \multicolumn{2}{c}{$h=4$} & \multicolumn{2}{c}{$h=5$} \\
\cmidrule(lr){2-3}\cmidrule(lr){4-5}\cmidrule(lr){6-7}\cmidrule(lr){8-9}\cmidrule(lr){10-11}\cmidrule(lr){12-13}
Variable & Proxy & HGMM & Proxy & HGMM & Proxy & HGMM & Proxy & HGMM & Proxy & HGMM & Proxy & HGMM \\
\midrule
Real oil price & 12.11 & \textbf{11.92} & 16.89 & \textbf{16.39} & 19.39 & \textbf{18.57} & 20.40 & \textbf{20.27} & 20.04 & \textbf{19.65} & 20.02 & \textbf{19.91} \\
World oil production & 4.70 & \textbf{3.66} & 3.53 & \textbf{2.59} & 2.90 & \textbf{2.25} & 2.20 & \textbf{1.86} & 1.85 & \textbf{1.39} & 1.55 & \textbf{1.26} \\
World oil inventories & 3.18 & 3.18 & 2.90 & 2.90 & 2.66 & \textbf{2.64} & 2.65 & 2.65 & 2.76 & \textbf{2.67} & 2.47 & \textbf{2.37} \\
World IP & 1.65 & \textbf{1.53} & 1.66 & \textbf{1.56} & 2.02 & \textbf{1.92} & 2.28 & \textbf{2.15} & 2.52 & \textbf{2.44} & 2.66 & \textbf{2.50} \\
U.S. IP & 1.96 & \textbf{1.89} & 2.19 & \textbf{2.11} & 2.42 & \textbf{2.30} & 2.69 & \textbf{2.57} & 2.91 & \textbf{2.74} & 2.91 & \textbf{2.75} \\
U.S. CPI & 0.66 & 0.66 & 0.93 & 0.93 & 0.95 & 0.95 & 0.99 & 0.99 & 1.03 & 1.03 & 1.01 & 1.01 \\
\bottomrule
\end{tabular}}
\begin{minipage}{0.96\textwidth}
\vspace{4pt}
\footnotesize
\textit{Notes:} Each entry is the length (upper minus lower bound) of the 90\% identification-robust Anderson--Rubin confidence set for the impulse response of the row variable to the target shock at horizon $h$ (months), on the \citet{kanzig_macroeconomic_2021} baseline: the full 1974:M1--2017:M12 sample, $p=12$ lags, and the original OPEC-window instrument. \textit{Proxy} inverts the AR statistic using only the five proxy exclusion conditions that identify the target shock; \textit{HGMM} (Proxy+NG) additionally imposes the target shock's five co-skewness conditions and the co-kurtosis aggregate, for $11$ conditions in total. The smaller length in each (variable, horizon) pair is in \textbf{bold}. Responses are normalized to a $10\%$ increase in the real oil price on impact. Sets are obtained by inverting the AR statistic over the target-shock rotation subspace with per-candidate moving-block-bootstrap critical values; see the construction in Appendix~\ref{sec:emp_arsets_cons}.
\end{minipage}
\end{table}

\begin{table}[htbp]
\centering
\caption{Ratio of 90\% Anderson--Rubin confidence-set lengths: \textit{Proxy-only} / \textit{HGMM} (K\"anzig baseline specification)}
\label{tab:arlen_kanzig_ratio}
\renewcommand{\arraystretch}{1.0}
\resizebox{0.80\textwidth}{!}{%
\begin{tabular}{lcccccc|cc}
\toprule
 & \multicolumn{6}{c|}{Ratio at horizon $h$} & \multicolumn{2}{c}{Over $h=0,\dots,50$} \\
\cmidrule(lr){2-7}\cmidrule(lr){8-9}
 Variable & $h=0$ & $h=1$ & $h=2$ & $h=3$ & $h=4$ & $h=5$ & Median & Max \\
\midrule
Real oil price & 1.02 & 1.03 & 1.04 & 1.01 & 1.02 & 1.01 & 1.01 & 1.04 \\
World oil production & 1.28 & 1.36 & 1.29 & 1.18 & 1.33 & 1.23 & 1.00 & 1.36 \\
World oil inventories & 1.00 & 1.00 & 1.01 & 1.00 & 1.03 & 1.04 & 1.09 & 1.11 \\
World IP & 1.08 & 1.06 & 1.05 & 1.06 & 1.03 & 1.06 & 1.03 & 1.08 \\
U.S. IP & 1.04 & 1.04 & 1.05 & 1.05 & 1.06 & 1.06 & 1.00 & 1.06 \\
U.S. CPI & 1.00 & 1.00 & 1.00 & 1.00 & 1.00 & 1.00 & 1.18 & 1.37 \\
\bottomrule
\end{tabular}}
\begin{minipage}{0.96\textwidth}
\vspace{4pt}
\footnotesize
\textit{Notes:} Each entry is the ratio $\text{len(Proxy)}/\text{len(HGMM)}$ for the impulse response of the row variable to the target shock. Values above one indicate that the \textit{Proxy-only} set is longer. The last two columns summarize the ratio over all $51$ horizons $h=0,\dots,50$. Construction as in Table~\ref{tab:arlen_kanzig}.
\end{minipage}
\end{table}


\section{Robustness: identification of a monetary policy shock, Gertler--Karadi (2015)}
\label{appendix:gk}

As a robustness check on the monetary policy illustration, we apply the hybrid GMM framework to the \citet{gertler_monetary_2015} system. This is a case in which the higher-order moment conditions carry too little signal to contribute for the degrees of freedom they add, and we report it because it shows the criterion does not invent identifying content in such cases.

\subsection{Specification}
\label{appendix:gk_spec}
The specification contains six variables ordered with the policy indicator
first: the two-year Treasury rate\footnote{In one of the specifications of the
original paper the authors use the one-year Treasury rate as the policy
indicator; the choice is not consequential for the relative
(non-)informativeness of the higher-order moment conditions.} (the policy
indicator), the consumer price index and industrial production (both in
$100\times\log$ levels), the excess bond premium, and two credit spreads. The
reduced-form VAR is estimated at monthly frequency with $p=12$ lags on
1979:M7--2012:M6. The instrument is the high-frequency fed-funds futures surprise
(FF4), available on 1991:M1--2012:M6; with $n=6$ and $k=1$ this gives
$p_\phi = 15$ rotation angles, $q_P = 5$ proxy conditions and $q_{NG}=75$
higher-order moment conditions, for $q=80$ in total. Inference is at the nominal $90\%$ confidence
level using a residual-based MBB with block length $\ell=15$
and $9999$ bootstrap replications (VAR re-estimated within each draw); the identification-robust
sets invert the AR statistic over the target-shock subspace at $10^5$ candidate
rotations. Impulse responses are reported for a one-standard-deviation monetary
policy shock. 
\subsection{Results}
\label{appendix:gk_results}
Figure~\ref{fig:gk_arsets} and Table~\ref{tab:arlen_gk} report the
identification-robust confidence sets. In contrast to the oil applications, the
\textit{Proxy-only} sets are not improved on, and adding the higher-order moment
conditions does {not} narrow them; across the $294$ (variable, horizon) pairs the \textit{HGMM} set is shorter at $46$, of equal length at $47$, and longer at $201$; the median length change is $+2.2\%$ and the maximum ratio anywhere is $1.04$. Neither criterion determines the sign of any response at any horizon. Again, this is consistent as when the non-Gaussianity of the shocks does not translate into robust sample-wide moment
restrictions, the extra conditions add degrees of freedom to the AR
statistic without a corresponding increase in identifying signal, so the joint set
cannot improve on the proxy-only set. We regard this as reassuring. The hybrid criterion tightens the confidence sets in the oil and euro-area applications, where the higher-order moment conditions are informative. A criterion that narrowed the sets in every application regardless of the data would be the worrying outcome.

The point responses remain those of a contractionary monetary policy shock and
are similar across the hybrid and proxy-SVAR estimators
(Figure~\ref{fig:gk_mbb}): the policy rate rises on impact, industrial production
and the price level decline, and the excess bond premium and credit spreads
widen, consistent with \citet{gertler_monetary_2015}. The specification test does
not reject the proxy exclusion restriction ($\hat{J}_{\text{prx}} = 3.09$, with
asymptotic and bootstrap $p$-values of $0.68$ and $0.99$).

\begin{figure}[H]
    \centering
    \includegraphics[width=0.95\textwidth]{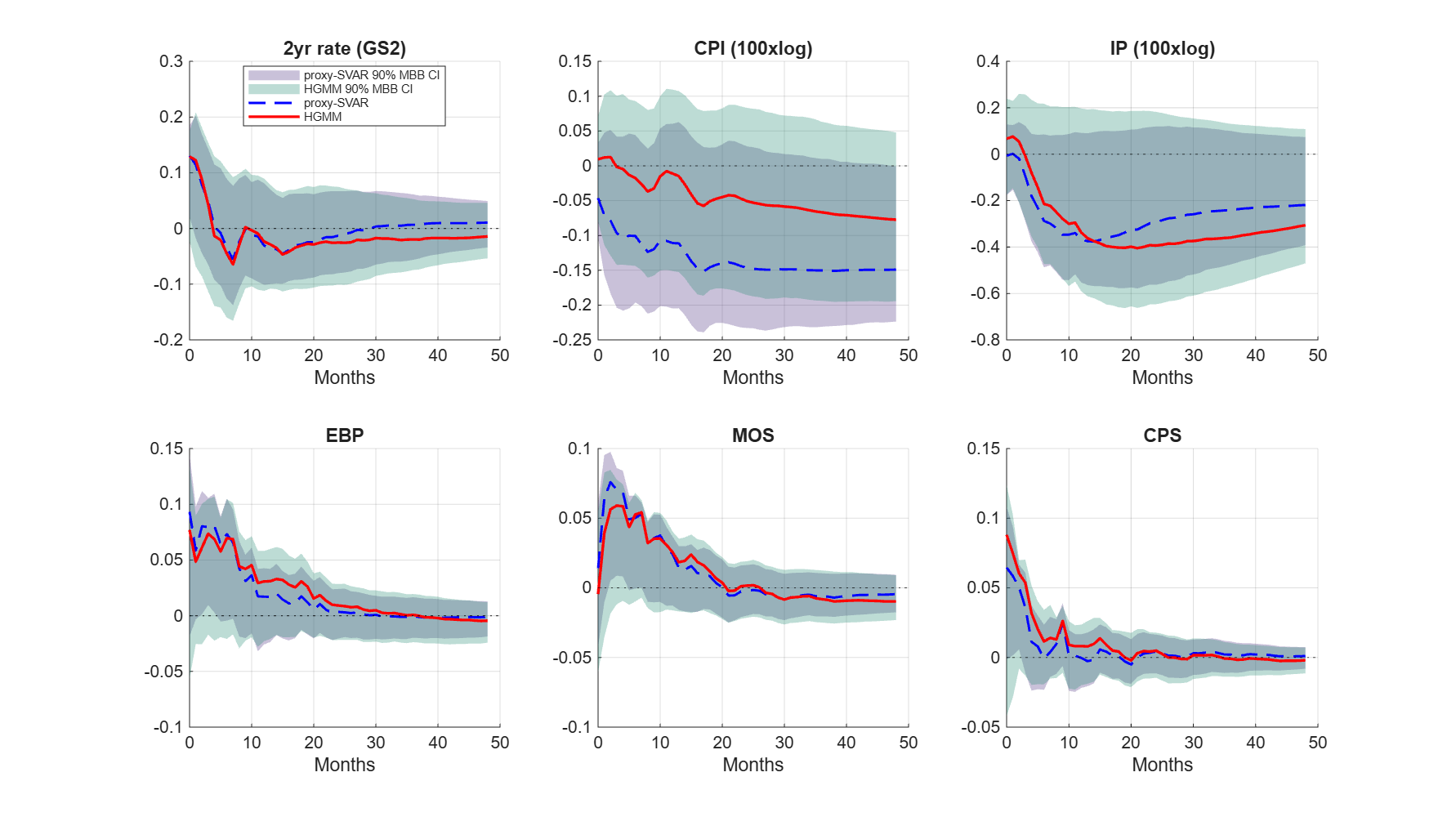}
    \caption{Impulse responses to the monetary policy shock: HGMM vs proxy-SVAR, with moving-block-bootstrap percentile confidence intervals.}
    \label{fig:gk_mbb}
    \begin{minipage}{0.92\textwidth}
\vspace{2pt}
\scriptsize
\textit{Notes:} As in Figure~\ref{fig:kl_mbb}, for the Gertler--Karadi (2015) system. \textcolor{red}{HGMM} (solid red) and \textcolor{blue}{proxy-SVAR} (dashed blue) impulse responses to a one-standard-deviation monetary policy shock, with $90\%$ moving-block-bootstrap percentile confidence intervals (shaded), computed from the same resampled residuals ($9999$ replications, block length $\ell=15$, VAR re-estimated per draw). The $x$-axis is the horizon in months.
    \end{minipage}
\end{figure}

\begin{figure}[H]
    \centering
    \includegraphics[width=0.95\textwidth]{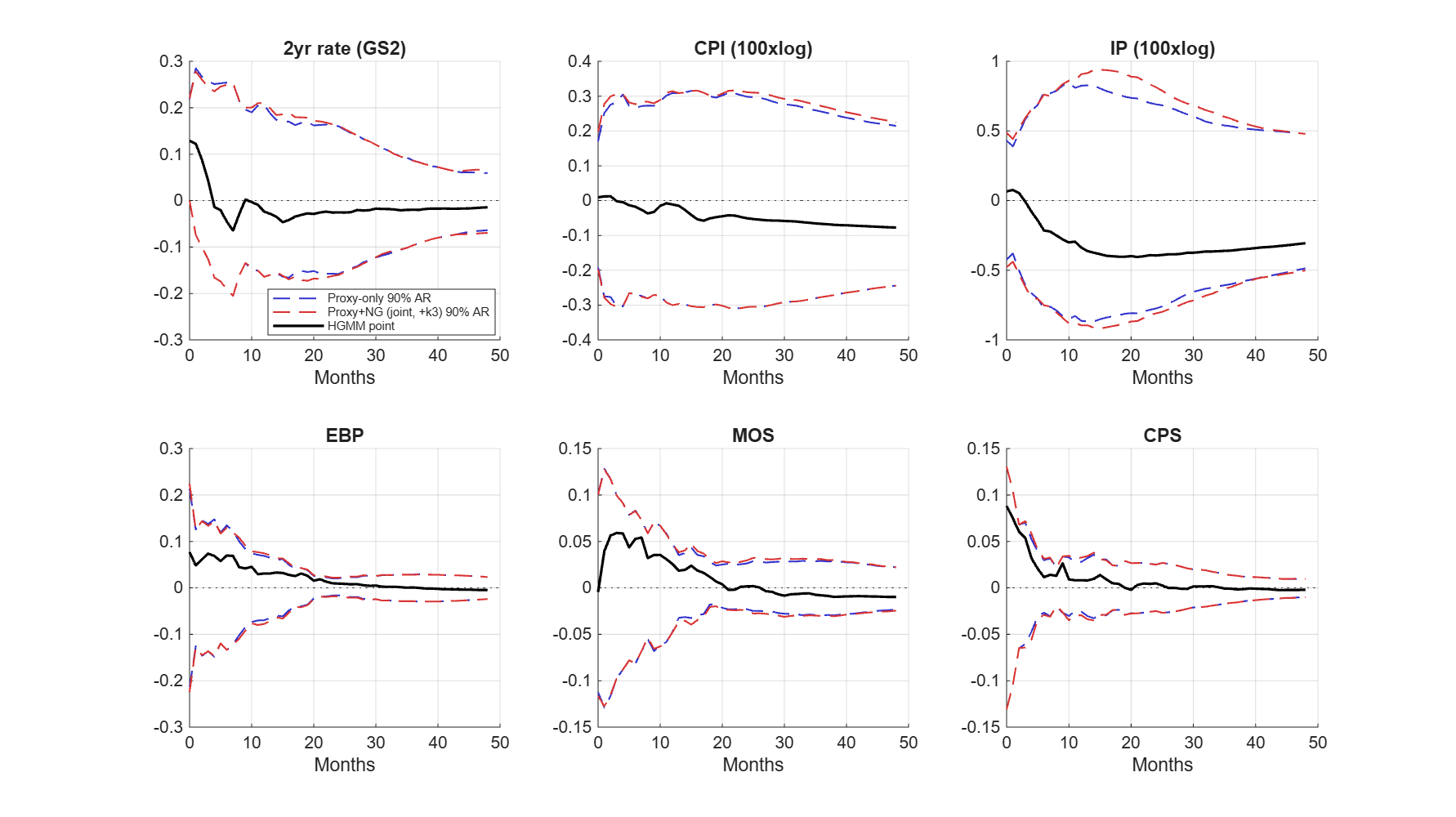}
    \caption{90\% Anderson--Rubin confidence sets for the monetary policy shock: \textit{Proxy-only} vs \textit{HGMM} (Gertler--Karadi application).}
    \label{fig:gk_arsets}
    \begin{minipage}{0.92\textwidth}
\vspace{2pt}
\scriptsize
\textit{Notes:} 90\% identification-robust Anderson--Rubin confidence sets for the impulse responses to the monetary policy shock in the Gertler--Karadi (2015) system (monthly VAR, $p=12$, 1979:M7--2012:M6; FF4 high-frequency surprise, 1991:M1--2012:M6). \textcolor{blue}{\textit{Proxy-only}} (blue, dashed) uses the five proxy orthogonality conditions; \textcolor{red}{\textit{HGMM}} (red, dashed) adds the co-skewness, asymmetric co-kurtosis, and co-kurtosis aggregate conditions of the target shock. The solid black line is the HGMM point estimate. Critical values are from a residual-based moving-block bootstrap ($9999$ replications, block length $\ell = 15$). Responses are to a one-standard-deviation monetary policy shock; the $x$-axis is the horizon in months.
    \end{minipage}
\end{figure}

\begin{table}[htbp]
\centering
\caption{Length comparison of 90\% Anderson--Rubin confidence sets: \textit{Proxy-only} vs \textit{HGMM} (Gertler--Karadi application)}
\label{tab:arlen_gk}
\renewcommand{\arraystretch}{1.15}
\resizebox{\textwidth}{!}{%
\begin{tabular}{lcccccccccccc}
\toprule
 & \multicolumn{2}{c}{$h=0$} & \multicolumn{2}{c}{$h=1$} & \multicolumn{2}{c}{$h=2$} & \multicolumn{2}{c}{$h=3$} & \multicolumn{2}{c}{$h=4$} & \multicolumn{2}{c}{$h=5$} \\
\cmidrule(lr){2-3}\cmidrule(lr){4-5}\cmidrule(lr){6-7}\cmidrule(lr){8-9}\cmidrule(lr){10-11}\cmidrule(lr){12-13}
Variable & Proxy & HGMM & Proxy & HGMM & Proxy & HGMM & Proxy & HGMM & Proxy & HGMM & Proxy & HGMM \\
\midrule
2yr rate (GS2) & 0.22 & 0.22 & 0.36 & \textbf{0.35} & 0.37 & \textbf{0.36} & 0.38 & \textbf{0.37} & 0.42 & \textbf{0.40} & 0.43 & \textbf{0.42} \\
CPI (100xlog) & \textbf{0.36} & 0.40 & \textbf{0.53} & 0.56 & \textbf{0.55} & 0.60 & \textbf{0.59} & 0.61 & 0.61 & 0.61 & \textbf{0.54} & 0.55 \\
IP (100xlog) & \textbf{0.85} & 0.97 & \textbf{0.77} & 0.88 & \textbf{0.99} & 1.04 & \textbf{1.19} & 1.21 & \textbf{1.31} & 1.35 & \textbf{1.39} & 1.40 \\
EBP & \textbf{0.43} & 0.45 & 0.25 & 0.25 & 0.29 & 0.29 & 0.27 & 0.27 & 0.30 & \textbf{0.29} & 0.24 & 0.24 \\
MOS & \textbf{0.21} & 0.22 & 0.26 & 0.26 & 0.23 & 0.23 & 0.20 & 0.20 & 0.18 & 0.18 & 0.16 & 0.16 \\
CPS & 0.26 & 0.26 & 0.21 & 0.21 & 0.13 & 0.13 & \textbf{0.13} & 0.14 & \textbf{0.10} & 0.11 & 0.07 & 0.07 \\
\bottomrule
\end{tabular}}
\begin{minipage}{0.96\textwidth}
\vspace{4pt}
\footnotesize
\textit{Notes:} Each entry is the length (upper minus lower bound) of the 90\% identification-robust Anderson--Rubin confidence set for the impulse response of the row variable to the target shock at horizon $h$ (months). \textit{Proxy} inverts the AR statistic using only the five proxy exclusion conditions that identify the target shock; \textit{HGMM} (Proxy+NG) additionally imposes the target shock's five co-skewness conditions, five asymmetric co-kurtosis conditions, and the co-kurtosis aggregate, for $16$ conditions in total. The smaller length in each (variable, horizon) pair is in \textbf{bold}. Responses are to a one-standard-deviation monetary policy shock. Sets are obtained by inverting the AR statistic over the target-shock rotation subspace with per-candidate moving-block-bootstrap critical values; see the construction in Appendix~\ref{sec:emp_arsets_cons}.
\end{minipage}
\end{table}

\begin{table}[htbp]
\centering
\caption{Ratio of 90\% Anderson--Rubin confidence-set lengths: \textit{Proxy-only} / \textit{HGMM} (Gertler--Karadi application)}
\label{tab:arlen_gk_ratio}
\renewcommand{\arraystretch}{1.0}
\resizebox{0.80\textwidth}{!}{%
\begin{tabular}{lcccccc|cc}
\toprule
 & \multicolumn{6}{c|}{Ratio at horizon $h$} & \multicolumn{2}{c}{Over $h=0,\dots,48$} \\
\cmidrule(lr){2-7}\cmidrule(lr){8-9}
 Variable & $h=0$ & $h=1$ & $h=2$ & $h=3$ & $h=4$ & $h=5$ & Median & Max \\
\midrule
2yr rate (GS2) & 1.00 & 1.02 & 1.02 & 1.03 & 1.04 & 1.02 & 0.99 & 1.04 \\
CPI (100xlog) & 0.92 & 0.95 & 0.93 & 0.96 & 1.00 & 0.98 & 0.98 & 1.00 \\
IP (100xlog) & 0.88 & 0.87 & 0.95 & 0.98 & 0.97 & 1.00 & 0.94 & 1.03 \\
EBP & 0.95 & 0.99 & 1.01 & 1.01 & 1.02 & 1.01 & 0.99 & 1.02 \\
MOS & 0.99 & 1.01 & 1.00 & 1.00 & 1.00 & 1.00 & 0.96 & 1.02 \\
CPS & 1.00 & 1.00 & 1.00 & 0.96 & 0.87 & 0.93 & 1.00 & 1.02 \\
\bottomrule
\end{tabular}}
\begin{minipage}{0.96\textwidth}
\vspace{4pt}
\footnotesize
\textit{Notes:} Each entry is the ratio $\text{len(Proxy)}/\text{len(HGMM)}$ for the impulse response of the row variable to the target shock. Values \emph{below} one indicate that the \textit{HGMM} set is the longer of the two. The last two columns summarize the ratio over all $49$ horizons $h=0,\dots,48$. Construction as in Table~\ref{tab:arlen_gk}.
\end{minipage}
\end{table}

\section{Evidence of weak instruments and non-Gaussianity in the macroeconomic literature}\label{appendix:evidencelit}
\begin{table}[htbp]
\centering
\caption{Instrument relevance in leading proxy-SVARs}
\label{tab:weak_instruments}
\renewcommand{\arraystretch}{1.15}
\resizebox{\textwidth}{!}{%
\begin{tabular}{llccl}
\toprule
Study & Identified shock (instrument) & $F$ & Robust/eff.\ $F$ & Source of statistic \\
\midrule
\citet{gertler_monetary_2015}            & Monetary policy (HF FF4 surprise)        & $21.5$ & $17.5$ & own first stage (one-year rate) \\
\quad\textit{same instrument}, LP-IV, no controls & Monetary policy                  & $1.7$  & ---    & \citet{stock_identification_2018}, Tab.\ 1$^{a}$ \\
\citet{montiel_olea_inference_2021}      & Oil supply (illustration)                & ---    & $9.4$  & robust stat.\ $<10$; $\xi_1=4.4$ \\
\citet{kanzig_macroeconomic_2021}        & Oil supply news (OPEC HF surprise)       & $22.7$ & $10.6$ & own first stage \\
\citet{romer_macroeconomic_2010}                 & Tax (narrative)                          & $1.6$  & ---    & \citet{ramey_chapter_2016}$^{b}$ \\
Government-spending proxies$^{c}$         & Fiscal (narrative / HF)                  & $<10$  & ---    & \citet{ramey_chapter_2016}, p.\ 120 \\
\bottomrule
\end{tabular}}
\begin{minipage}{0.96\textwidth}
\vspace{4pt}
\footnotesize
\textit{Notes:} \enquote{$F$} is the standard (homoskedastic) first-stage $F$; \enquote{Robust/eff.\ $F$} is the
heteroskedasticity-robust or \citet{montiel_olea_inference_2021}-type effective first-stage statistic.
The \citet{stock_testing_2002} $10\%$ critical value is $\approx 10.3$; the common rule of thumb is $F>10$.
Leading instruments are strong in their original specification, but relevance is fragile under robust
statistics or alternative specifications.
$^{a}$ Rises to $F=23.7$ once four lags of $Y_t$ and $Z_t$ are included as controls.
$^{b}$ First stage of tax revenue on the \citet{romer_macroeconomic_2010} narrative shocks.
$^{c}$ \citet{fisher_using_2010}, Ben~Zeev \& Pappa (2017), and \citet{barth_iii_cost_2000}; see
\citet{ramey_chapter_2016}, p.\ 120, Fig.\ 4 Panel~B.
\end{minipage}
\end{table}

\begin{table}[htbp]
\centering
\caption{non-Gaussianity in macroeconomic and financial shocks}
\label{tab:nongaussianity}
\renewcommand{\arraystretch}{1.15}
\resizebox{\textwidth}{!}{%
\begin{tabular}{lll}
\toprule
Study & Series / shocks & Non-Gaussianity measure (estimate) \\
\midrule
\citet{lanne_identification_2017}      & Macro SVAR structural errors        & Gaussianity rejected; ML favors heavy-tailed errors \\
\citet{gourieroux_statistical_2017}     & SVAR shocks (ICA)                   & Components significantly non-Gaussian ($\le 1$ Gaussian) \\
\citet{bekaert_macro_2021}   & Aggregate supply \& demand shocks   & Demand shocks negatively skewed, leptokurtic; higher moments sig. \\
\citet{salgado_skewed_2019}      & Output / employment / sales growth  & Strongly procyclical skewness; left tail widens in recessions \\
\citet{fagiolo_napoletano_roventini_2008} & OECD output growth (US GDP, IP)   & Laplace tails (Subbotin $b\approx1$); excess kurtosis $\approx 3$ \\
\citet{guvenen_etal_2021}               & Labor-income shocks                 & Left-skewed; kurtosis up to $\approx 18$ (vs $3$ Gaussian) \\
\bottomrule
\end{tabular}}
\begin{minipage}{0.96\textwidth}
\vspace{4pt}
\footnotesize
\textit{Notes:}The first three rows concern structural SVAR shocks directly; the remainder
document non-Gaussianity at the aggregate-output and labor-income level. A Gaussian benchmark has
skewness $0$ and kurtosis $3$.
\end{minipage}
\end{table}

\end{document}